\documentclass[11pt]{article}
\pdfoutput=1
\usepackage[UKenglish]{babel}
\usepackage[
letterpaper,
top=1in,
bottom=1in,
left=1in,
right=1in]{geometry}
\usepackage[noadjust]{cite}
\usepackage{amsmath}
\usepackage{comment}
\usepackage{amsfonts}
\usepackage{anyfontsize}
\usepackage{empheq}
\usepackage{centernot}
\usepackage{ifxetex}
\ifxetex
\usepackage[bigdelims,vvarbb]{newpxmath}
\usepackage{mathspec}
\setallmainfonts(Digits,Latin){Noto Serif}
\else
\usepackage{mathpazo}
\fi

\usepackage{amsthm}
\usepackage{amssymb}
\usepackage{mathtools}
\usepackage{microtype}
\usepackage{bm}
\usepackage{enumitem}
\usepackage{complexity}
\usepackage{mathrsfs}
\usepackage{caption}
\usepackage{subcaption}
\usepackage{makecell}
\usepackage{diagbox}

\usepackage{graphicx} 
\usepackage{scalefnt}
\usepackage{stmaryrd}

\let\orignot\not 
\usepackage{mathabx}  
\let\not\orignot

\usepackage{faktor}
\usepackage[allcommands, old-arrows]{overarrows}
\usepackage[colorlinks=true, linkcolor=blue,citecolor=blue,urlcolor=blue]{hyperref}
\usepackage{cleveref}

\DeclareMathOperator{\pcsp}{PCSP}
\DeclareMathOperator{\spcsp}{sPCSP}

\DeclareMathOperator{\csp}{CSP}
\DeclareMathOperator{\scsp}{sCSP}
\DeclareMathOperator{\pol}{Pol}
\DeclareMathOperator{\Hom}{Hom}

\DeclareMathOperator{\id}{id}

\DeclareMathOperator{\partialmap}{\rightharpoonup}

\newcommand{\overpartialmap}[1]{\mathrel{\raisebox{-1pt}{
$\xrightharpoonup{~#1~}$}}
}

\newcommand{\pj}{\mathrm{pj}}
\newcommand{\eg}{\mathrm{eg}}
\newcommand{\io}{\mathrm{io}}

\newcommand{\qclass}[1]{\langle#1 \rangle}

\renewcommand{\A}{{\bm{A}}}
\newcommand{\B}{{\bm{B}}}
\newcommand{\bF}{\bm{F}}
\newcommand{\bG}{\bm{G}}
\newcommand{\bT}{\bm{T}}

\newcommand{\bS}{\bm{S}}

\newcommand{\Rfrak}{\mathfrak{R}}

\newcommand{\GS}{S}
\newcommand{\bGS}{\bm{S}}

\newcommand{\LS}{\mathrm{S}}
\newcommand{\bLS}{\bm{\LS}}

\renewcommand{\C}{{\bm{C}}}
\newcommand{\I}{{\bm{I}}}

\newcommand{\bK}{\bm{K}}

\newcommand{\Dcal}{\mathcal{D}}
\newcommand{\Ccal}{\mathcal{C}}
\newcommand{\Qcal}{\mathcal{Q}}
\newcommand{\Mcal}{\mathcal{M}}

\newcommand{\Jcal}{\mathcal{J}}
\newcommand{\Lcal}{\mathcal{L}}
\newcommand{\Ical}{\mathcal{I}}
\newcommand{\Mscr}{\mathscr{M}}

\newcommand{\Nscr}{\mathscr{N}}
\newcommand{\Cscr}{\mathscr{C}}
\newcommand{\NN}{\mathbb{N}}
\newcommand{\ZZ}{\mathbb{Z}}
\newcommand{\QQ}{\mathbb{Q}}
\newcommand{\arty}{\mathrm{ar}}
\newcommand{\sPMC}{\mathrm{sPMC}}
\newcommand{\MC}{\mathrm{MC}}
\newcommand{\PMC}{\mathrm{PMC}}
\newcommand{\Wscr}{\mathscr{W}}

\newcommand\ScaleExists[1]{\vcenter{\hbox{\scalefont{#1}$\exists$}}}

\DeclareMathOperator*\bigexists{%
  \vphantom\sum
  \mathchoice{\ScaleExists{2}}{\ScaleExists{1.4}}{\ScaleExists{1}}{\ScaleExists{0.75}}}

\newcommand\Crestrict[2]{{
  \left.\kern-\nulldelimiterspace 
  #1 
  \right|_{#2} 
  }}
\makeatletter
\newcommand{\oset}[3][0ex]{%
  \mathrel{\mathop{#3}\limits^{
    \vbox to#1{\kern-2\ex@
    \hbox{$\scriptstyle#2$}\vss}}}}
\makeatother

\makeatletter
\DeclareRobustCommand{\Div}{\mathrel{\Div@}}
\DeclareRobustCommand{\nDiv}{\mathrel{\nDiv@}}

\newcommand{\Div@}{%
  \mkern2mu\nonscript\mkern-2mu 
  \mathpalette\Div@@\relax
  \mkern2mu\nonscript\mkern-2mu
}
\newcommand{\Div@@}[2]{%
  \hbox{%
    \sbox\z@{$#1T$}%
    \vbox to \ht\z@{%
      \offinterlineskip\m@th
      \hbox{$#1.$}\vfil
      \hbox{$#1.$}\vfil
      \hbox{$#1.$}%
    }%
  }%
}

\newcommand{\nDiv@}{\centernot\Div@}
\makeatother

\DeclareMathAlphabet{\mathbbmsl}{U}{bbm}{m}{sl}
\newcommand{\Ebb}{\mathbbmsl{E}}

\newcommand{\AIP}{\text{AIP}}
\newcommand{\Ip}{\text{IP}}
\newcommand{\BLP}{\text{BLP}}

\usepackage{thm-restate}

\newtheorem{theorem}{Theorem}[section]


\newtheorem{proposition}[theorem]{Proposition}


\newtheorem{lemma}[theorem]{Lemma}


\newtheorem{corollary}[theorem]{Corollary}




\newtheorem{fact}[theorem]{Fact}


\newtheorem{observation}{Observation}

\theoremstyle{definition}

\newtheorem{ppty}{Property}
\newtheorem{desc}{Description}
\newtheorem{interpret}{Interpretation}

\graphicspath{ {./figures/} }
\title{Ineffectiveness for Search and Undecidability of PCSP Meta-Problems}
\author{Alberto Larrauri \thanks{This work was supported by the UKRI grant
EP/X024431/1, and the UK EPSRC grant EP/X033201/1. The research leading to these results was conducted at the University of Oxford between 2024 and 2025, and continued during the summer of 2025 at Durham University.}}
\date{ \small University of Zaragoza}
\begin{document}

\maketitle
\thispagestyle{empty}
\begin{abstract}

It is an open question whether the search and decision versions of promise CSPs are equivalent. 
Most known algorithms for PCSPs solve only their \emph{decision} variant, and it is unknown whether they can be adapted to solve \emph{search} as well. 
The main approaches, called $\BLP, \AIP$ and $\BLP+\AIP$, handle a PCSP by finding a solution to a relaxation of some integer program. We prove that rounding those solutions to a proper search certificate can be as hard as any problem in the class TFNP. In other words, these algorithms are ineffective for search. Building on the algebraic approach to PCSPs, we find sufficient conditions that imply ineffectiveness for search. Our tools are tailored to algorithms that are characterized by minions in a suitable way, and can also be used to prove undecidability results for meta-problems. This way, we show that the families of templates solvable via $\BLP, \AIP$, and $\BLP+\AIP$ are undecidable. \par

Using the same techniques we also analyze several algebraic conditions that are known to guarantee the tractability of finite-template CSPs. We prove that several meta-problems related to cyclic polymorphims and WNUs are undecidable for PCSPs. In particular, there is no algorithm deciding whether a finite PCSP template (1) admits cyclic a polymorphism, (2) admits a WNU. 

\end{abstract}

\clearpage
\pagenumbering{arabic} 

{
\small
\tableofcontents
}
\newpage
\section{Introduction}

The Dichotomy Theorem \cite{Zhuk20:jacm,Bulatov17:focs} and, more generally, the algebraic approach to constraint satisfaction \cite{krokhin2005complexity,barto2017polymorphisms} show that finite-template constraint satisfaction problems (CSPs) form a particularly well-behaved class of NP problems. This class includes a wide range of natural problems relevant across many domains, such as variants of the Boolean satisfiability problem, graph coloring problems, and systems of equations over finite algebraic structures. Unconditionally, the search and decision variant of each finite-template CSPs are polynomial-time equivalent. Each problem in this class is either NP-complete or in P, and there is an explicit procedure that correctly assigns one of these two cases to each given CSP.
This stands in contrast to the well-known fact that there are NP-intermediate problems if P$\neq$ NP
\cite{ladner1975structure}, and there are problems in FNP (the search analogue of NP) which we believe to be hard but whose decision variants are trivial, captured by the TFNP class \cite{goldberg2017tfnp_update}. Therefore, it is natural to ask how far (and in which directions) can finite-template CSPs be generalized while still keeping their nice properties. 
 \par

\emph{Promise} CSPs (PCSPs) are qualitative relaxations of CSPs generalizing the task of coloring a $c$-colorable graph using $d$ colors for fixed integers $d\geq c$, called the \emph{approximate graph coloring} (AGC) problem. A PCSP is given by a pair of relational structures $(\A, \B)$, the \emph{template}, where the first maps homomorphically into the second, denoted $\A \rightarrow \B$. In the decision variant of $\pcsp(\A,\B)$ the task is to distinguish input structures $\I$ satisfying $\I \rightarrow \A$ from those satisfying $\I \not\rightarrow \B$. In the search variant, the goal is to find an explicit homomorphism from $\I$ to $\B$ given the promise that $\I \rightarrow \A$. The algebraic approach, instrumental in much of the CSP theory and both proofs of the Dichotomy Theorem, has recently been extended to the PCSP framework \cite{BBKO21}. Roughly, this approach studies the complexity of $\pcsp(\A, \B)$ by analyzing the algebraic properties of the set $\pol(\A, \B)$ of 
homomorphisms $f: \A^n \rightarrow \B$ from some direct power of $\A$ to $\B$, called \emph{polymorphisms}. This extension has motivated a surge of activity in the area that has yielded partial complexity classifications \cite{BG21:sicomp,Barto21:stacs,ficak2019dichotomy,NZ23:lics,brakensiek2023conditional,LZ24solving,NZ24symmetric_functional}, efficient algorithms \cite{BBKO21,BGWZ20,cz23sicomp:clap,conghaile2022cohomology}, and hardness conditions \cite{brandts2021complexity,Barto22:soda,KOWZ23,banakh2024injective}. However, many basic questions remain unanswered. For example, the complexity classifications of Boolean and graph PCSPs (including the complexity of AGC) are still open, despite relevant progress in those directions \cite{Ficak19:icalp,KOWZ23}. Crucially, the relationship between search and decision PCSPs is not well understood, and it is unknown whether there is always an efficient way of solving a finite-template search PCSP whose decision variant is tractable. 
\par

In this work we analyze several PCSP algorithms, as well as some algebraic conditions that guarantee tractability in the CSP setting. It is known that, unless P=NP, there are tractable finite-template PCSPs which cannot be reduced via gadgets to a tractable finite-template CSP, as exemplified by  ``$1$-in-$3$-SAT vs not-all-equal-SAT'' \cite{barto2019promises,asimi2021finitely,mottet2025algebraic,pinsker2025three}. Therefore, there is a need to develop algorithms that go beyond the finite CSP case. Solving the decision variant of $\pcsp(\A, \B)$ involves solving a tractable relaxation of $\csp(\A)$ which should be a restriction of $\csp(\B)$. The main relaxations used for this purpose are direct relaxations of the basic integer programming formulation of $\csp(\A)$ \cite{BBKO21,BGWZ20,bgs_robust23stoc}, the local consistency algorithm \cite{barto2009boundedwidth,Atserias22:soda}, and algorithmic hierarchies built on top of these two previous approaches \cite{cz23soda:minions,cz23sicomp:clap,conghaile2022cohomology}. Virtually all these algorithms present two inherent limitations. The first is that they only solve the decision version of $\pcsp(\A, \B)$: when an instance $\I$ is accepted by the relaxation, we do not directly obtain a homomorphism $\I \rightarrow \B$. To obtain such a homomorphism we must \emph{round} the solution to the relaxation, and there is no obvious efficient way of doing so except for in a few known cases \cite{BG21:sicomp,brakensiek2019algorithmic}. The second limitation of these algorithms is that, despite some of them admitting nice algebraic characterizations, we do not know of a way to recognize the problems $\pcsp(\A, \B)$ that they solve. This is known as the \emph{meta-problem} related to these algorithms. \par

In our analysis of algorithms we focus on the \emph{affine integer programming} (AIP) and the \emph{basic linear programming} (BLP) relaxations for PCSPs, as well as on the $\BLP+\AIP$ relaxation, which combines the power of the previous two. We present results of three kinds: hardness, undecidability, and non-computability. We show that rounding the output of these algorithms to search certificates can be as hard as any problem in the \emph{total} FNP (TFNP) class, and that all these algorithms have undecidable meta-problems. 

\begin{theorem}[Main algorithmic result, informal]
Let $\Qcal\in \{ \AIP, \BLP, \BLP+\AIP\}$, and let $\mathcal{S}^\Qcal$ be the family of finite templates $(\A, \B)$ such that $\Qcal$ solves $\pcsp(\A, \B)$. Then, given any problem $\Lambda$ in the TFNP class, there is a finite template $(\A_\Lambda, \B_\Lambda)\in \mathcal{S}^\Qcal$ such that $\Lambda$ is many-one reducible in polynomial time to the problem of finding a homomorphism $\I \rightarrow \B_\Lambda$ for an input structure $\I$ accepted by $\Qcal$. Furthermore, $\mathcal{S}^\Qcal$ is undecidable.    
\end{theorem}

This is the first \emph{hardness-of-rounding} result in the PCSPs literature. 
If TFNP contains some problem that has no polynomial-time solution (which is widely conjectured \cite{megiddo1991total,hubacek2017journey_tfnp}), then this gives a negative answer to the question of whether the output of these algorithms can always be rounded to a search solution efficiently, posed in \cite{krokhin2022invitation}
\footnote{A previous version of this work also claimed the results presented here yielded negative answers to the questions in \cite{brakensiek2019algorithmic,BGWZ20}, without further clarification. This was incorrect and misleading.}. 
Furthermore, for $\AIP$ and $\BLP$ we show that there are finite-template PCSPs solved by those algorithms where no computable function is a valid rounding rule. This remains true even for \emph{left-Boolean} templates $(\A, \B)$ (i.e., those where $\A$'s domain contains only two elements). This means that directly rounding the output of these algorithms through consistent families of well-behaved polymorphisms such as the ones defined in \cite{brakensiek2019algorithmic} is not always possible. However, changing the domain of our relaxations to other \textit{efficiently computable rings}, as proposed in \cite{brakensiek2019algorithmic}, could be a way to bypass the issues presented in this paper. This is expanded upon in the discussion section (\Cref{sec:discussion}). 
We also point out that while our results include small bounds on $\A$'s domain size, the domain of the right structure $\B$ can grow quite large. Hence, our results do not rule out the possibility that $\BLP, \AIP$, and $\BLP+\AIP$ could always be directly adapted for search in the Boolean setting. \par

Our techniques leverage the characterizations of the algorithms $\AIP$, $\BLP$, and $\BLP+ \AIP$ by means of objects called \emph{minions} (or minor-closed classes \cite{pippenger2002galois}) that lie at the heart of the algebraic method. On a high level, our results follow from reducing tiling problems (e.g., \cite{hanf1974nonrecursive}) to \emph{promise minor condition ($\PMC$) problems} \cite{BBKO21}.
We remark that, although the connection between the TFNP class and tiling problems is not difficult to prove, to our knowledge it has not been previously used to obtain hardness results within TFNP. 
To achieve our reductions we develop a new way of encoding PCSPs in $\PMC$ problems by means of a sheaf-like \cite{bredon2012sheaf} construction on minions where elements represent partial homomorphisms, and minoring operations ensure consistency between the corresponding local homomorphisms. \par

We also study some well-known algebraic conditions following the same approach as in our analysis of algorithms. A polymorphism $f:\A^n \rightarrow \B$ is called \emph{cyclic} if it is invariant under any cyclic permutation of its arguments, and \emph{weak near-unanimity} (WNU) if the value of $f(x,\dots, x , y, x,\dots, x)$ is independent of the position of the lone $y$. The existence of a WNU and the existence of a cyclic polymorphism are equivalent conditions that characterize the tractability of finite-template CSPs if P$\neq$ NP \cite{Bulatov17:focs, Zhuk20:jacm}.
There is a rich network of equivalences between algebraic conditions in this setting (e.g.,  \cite[Theorems 41 and 47]{barto2017polymorphisms}) which were established using tools from universal algebra. These equivalences are a powerful aid in proving the decidability of various \emph{meta-questions}. For instance, it is known that a finite structure $\A$ admits a cyclic polymorphism if and only if it admits a cyclic polymorphism of every prime arity larger than the domain size of $\A$ \cite{barto2012absorbing}, and hence one can easily decide the existence of a WNU or a cyclic polymorphism. Another important algebraic condition is the presence of WNUs of all arities $k\geq 3$, which characterizes bounded width CSPs in the finite-template setting \cite{barto2009boundedwidth,kozik2015characterizations}. This condition is equivalent to the presence of WNUs of arities $3$ and $4$ satisfying a particular relation \cite{kozik2015characterizations} and is, hence, decidable. The complexity of these and other meta-questions has been examined in \cite{chen2017metaquestions}. \par
For PCSPs the picture is significantly less structured. We know that, unless P$\neq$ NP, no finite family of polymorphisms can guarantee tractability  \cite{BBKO21}. Hence, all algebraic tractability conditions must involve an infinite family of polymorphisms. Even then, admitting WNUs of all arities $k\geq 3$ is no longer a sufficient condition for bounded width (but remains a necessary one) \cite{Atserias22:soda}, and does not even guarantee tractability \cite{ciardo2024periodic}. Since many of the tools from universal algebra no longer apply to the PCSP setting, few non-trivial implications between algebraic conditions have been shown (with some exceptions, e.g. \cite[Theorem 4]{BGWZ20}). In particular, the decidability of most meta-questions related to natural algebraic conditions remains open for PCSPs. Our main algebraic result is the following. 

\begin{restatable}{theorem}{minors}
\label{th:minor_identities_main}
The following problems are undecidable. Given an input finite template $(\A, \B)$ with $|A|\leq 3$, determine whether $\pol(\A, \B)$ contains: 
\begin{enumerate}[label=(\arabic*)]
    \item cyclic polymorphisms 
    \label{item:main_cyclic}
    \begin{enumerate}[label=(\roman*)]
        \item for every prime arity $p$,
         \label{item:main_cyclic_all}
        \item for at least one arity $k$,
        \label{item:main_cyclic_one}
        \item for all but finitely many prime arities $p$, 
        \label{item:main_cyclic_eg}
        \item for infinitely many prime arities $p$.
        \label{item:main_cyclic_io}
    \end{enumerate}
    \item  weak near‑unanimity polymorphisms (WNUs)
    \label{item:main_wnu}
    \begin{enumerate}[label=(\roman*)]
        \item for every arity $k\geq 3$,
        \label{item:main_wnu_all}
        \item for at least one arity $k\geq 3$,
        \label{item:main_wnu_one}
        \item for all but finitely many arities $k\geq 3$,
        \label{item:main_wnu_eg}
        \item for infinitely many arities $k\geq 3$.
        \label{item:main_io}
    \end{enumerate}
\end{enumerate}  
\end{restatable}

\section{Preliminaries}
\label{sec:prelims}
The set of natural numbers $\NN$ starts at $1$, and $[k]$ is the set $\{1,2,\ldots,k\}$. Given sets $S, T$, the set of maps $f:S \rightarrow T$ is denoted by $T^S$.
A \emph{partial function} from $S$ to $T$ is denoted as $f: S\partialmap T$. 
We write $\id_X$ for the identity map on a set $X$. We identify tuples $\bm t=(t_1,\dots, t_n)\in T^n$ with functions in $T^{[n]}$ in the natural way (i.e., $\bm t(i)=t_i$).
Disjoint unions are denoted using $\sqcup$. We write $U^*$ for the set of finite strings over a finite alphabet $U$, i.e. $U^* = \bigsqcup_{n\geq 0} U^n$. A partial map $f: U^* \partialmap U^*$ is said to be \emph{computable} if there is a Turing machine on an alphabet containing $U$ that computes $f(\bm{x})$ for any input $\bm{x}\in U^*$ where $f$ is defined, and runs forever on $\bm{x}\in U^*$ where $f$ is undefined. We informally say that a partial map $f: S \partialmap T$ between countable sets $S,T$ is computable if $f$ is computable under some implicit encoding. That is, we implicitly refer to a finite alphabet $U$, and to injective maps $\alpha:S\rightarrow U^*$ and $\beta: T\rightarrow U^*$, and mean that the partial map $\beta \circ f \circ \alpha^{-1}$ is computable \cite{cutland1980computability}.

\paragraph{Search Problems}

We refer to \cite[Section 10.3]{papadimitriou1994complexity} for an introduction to search complexity and the classes FP, FNP, and TFNP.
Let $\mathfrak{R}\subseteq U^* \times V^*$ be a binary relation on
finite words, and $\bot$ a special ``reject'' symbol outside the alphabets $U, V$. The \emph{search problem} $\Lambda_\Rfrak$,
is defined as: given a string $\bm x\in U^*$,
find some $\bm y\in V^*$ satisfying $(\bm x,\bm y)\in \mathfrak{R}$, or reject $\bm x$ and return $\bot$ if no such $\bm y$ exists. The class of \emph{total functional NP} (TFNP) problems consists of all problems $\Lambda_\Rfrak$ where $\Rfrak$ is (1) a total relation, meaning that for all input strings $\bm x$ there is some $\bm y$ satisfying $(\bm x,\bm y)\in \Rfrak$, (2) polynomially bounded, meaning that there is some polynomial $p(n)$ such that $(\bm x, \bm y)\in \Rfrak$ implies $|\bm y| \leq p(|\bm x|)$, and (3) $\Rfrak$ is recognizable in polynomial time. The class of \emph{tally TFNP} problems, denoted TFNP$_1$, consists of all problems $\Lambda_\Rfrak \in \text{TFNP}$ such that $\Rfrak \subseteq U^* \times V^*$ is a relation where $U$ is the unary alphabet $\{ 1\}$. \par

The class TFNP has been studied extensively (e.g., \cite{goldberg2017tfnp_update,hubacek2017journey_tfnp,goos2024separations_tfnp}) 
and contains several problems for which no efficient solution is expected such as integer factoring, or the problem of computing a Nash equilibrium~\cite{daskalakis2009complexity}. 
There is compelling reason to believe that $\text{TFNP} \not\subseteq \text{FP}$ \cite{hubacek2017journey_tfnp} (FP being the class of search problems solvable in polynomial time), which would imply that $\text{P}\neq \text{NP}$.
On the other hand, there are oracles with respect to which $\text{TFNP} \subseteq \text{FP}$ but the polynomial hierarchy is infinite \cite{buhrman2010oracle}.  
To our knowledge, the class TFNP$_1$ has not been studied explicitly, but also contains problems that have no known polynomial-time solution, such as \textsc{Prime}, where $1^n$ is given as input and the task is to find a prime number on $n$ bits \cite{gat2011probabilistic}. \par

A (polynomial-time) \emph{many-one reduction} from $\Lambda_{\Rfrak_1}$ to $\Lambda_{\Rfrak_2}$ consists of a pair $(\alpha,\beta)$ of polynomial-time computable functions satisfying that for all $\bm x$ (1)
if $(\alpha(\bm x),\bm z)\in \mathfrak{R}_2$, then $(\bm x, \beta(\bm x, \bm z))\in \mathfrak{R}_1$,
and (2) if $(\alpha(\bm x),\bm z)\not\in \mathfrak{R}_2$ for all $\bm z$, then
$(\bm x, \bm y)\notin\in \mathfrak{R}_1$ for all $\bm y$. Similarly, a \emph{generalized many-one reduction} from $\Lambda_{\Rfrak_1}$ to $\Lambda_{\Rfrak_2}$ is a pair of polynomial-time computable functions $(\alpha,\beta)$ satisfying that whenever $\bm z$ is a valid answer to $\alpha(\bm x)$ in the problem $\Lambda_{\Rfrak_2}$
(including the case where $\bm z = \bot$), then $\beta(\bm x,\bm z)$ is a valid answer to $\bm x$ in the problem $\Rfrak_1$ (including again the case where $\beta(\bm x,\bm z)= \bot$). Given a class $\mathcal{C}$ of search problems, and a set $\mathcal{F}$ of search problems, we say that $\mathcal{F}$ is $\mathcal{C}$-hard if for any $\mathfrak{R}\in \mathcal{C}$ there is some $\mathfrak{R}^\prime \in \mathcal{F}$ such that $\Lambda_\Rfrak$ has a many-one reduction to $\Lambda_{\Rfrak^\prime}$. A problem $\Lambda_\Rfrak$ is called $\Ccal$-hard if the family $\{ \Lambda_\Rfrak \}$ is $\Ccal$-hard. The class TFNP is conjectured to contain no TFNP-hard problems \cite{megiddo1991total}, so our results show TFNP-hardness of families. 
\par

We also consider \emph{promise} search problems, which can be seen as ``partial'' search problems. The notions in this section extend to promise problems in the natural way, along the lines of e.g., \cite{even1984crypto_promise,goldreich2006promisesurvey}.

\paragraph{Promise Constraint Satisfaction}

A \emph{relational signature} $\Sigma$ 
is a finite set of symbols where each $R\in \Sigma$ has some \emph{arity} $\arty(R)\in \NN$. A $\Sigma$-\emph{structure} $\A$ consists of: a set $A$ called its \emph{universe}, and a relation $R^A\subseteq A^{\arty(R)}$ for each $R\in \Sigma$. Given two \emph{similar} (i.e., with the same signature) structures $\A, \B$, a homomorphism $h:\A \rightarrow \B$ is a map from $A$ to $B$ satisfying $(h(e(1)), \dots, h(e(\arty(R))))\in R^B$     
for each $R\in \Sigma$ and each tuple $e\in R^A$. We write $\A \rightarrow \B$ to denote there is a homomorphism from $\A$ to $\B$. 
The \emph{$n$-th power} of $\A$, denoted $\A^n$, is a structure similar to $\A$ whose universe is $A^n$, and where 
$R^{A^n}$ consists of the tuples $(\bm{a}_1, \dots , \bm{a}_{\arty(R)})$, satisfying $(a_{1,j},\dots a_{\arty(R),j})\in R^A$ for all $j\in [n]$. \par
\par

\emph{Templates} are pairs of structures $(\A, \B)$ satisfying $\A \rightarrow \B$. The template is
\emph{finite} if both $\A, \B$ are finite. 
The \emph{decision promise constraint satisfaction problem} (PCSP) defined by a template $(\A, \B)$, denoted
$\pcsp(\A,\B)$, is the problem of, given an input finite structure $\I$, to accept it
if $\I \rightarrow \A$, and to reject it if $\I\not\rightarrow \B$. Similarly, in the \emph{search PCSP} defined by $(\A, \B)$, denoted $\spcsp(\A, \B)$, the goal is to find a homomorphism $F: \I \rightarrow \B$, if $\I\rightarrow \A$, or to reject $\I$ if $\I \not\rightarrow \B$. Observe that this only makes sense if $\B$ is finite or a suitable encoding is fixed. We define the problems $\csp(\A)$ and $\scsp(\A)$ as $\pcsp(\A, \A)$ and $\spcsp(\A, \A)$ respectively.\par

\paragraph{Efficient Algorithms and Rounding Problems}
Let $\A$, $\I$ be finite $\Sigma$-structures. The following system of equations over the integers $\{0,1\}$, denoted $\Ip_{\A}(\I)$, is satisfiable if and only if $\I \rightarrow \A$: 
\begin{align}
\nonumber
   \text{Variables: }\, \, 
   &
\{x_{v,a}  \, \vert \, v\in I, a\in A\} 
\quad \sqcup \\ \nonumber 
  &  \{
   x_{r_A,r_I} \, \vert \, R\in \Sigma, r_I\in R^I,  r_A\in R^A \}.  \\ \nonumber &
\\
   \label{eq:IP}
   \text{Equations: }
   & \begin{array}{l r}
    \sum_{a\in A} x_{v,a} = 1,  \\
    \sum_{ r_A \in R^A} x_{r_I, r_A} = 1, \\
    \sum_{r_A \in R^A, r_A(i)=\alpha}
    x_{r_I, r_A} = x_{r_I(i), \alpha}, 
\end{array} \\ 
\nonumber
 \text{ for each } & \text{$v\in I$, $R\in \Sigma$, $r_I\in R^I$, $i\in [\arty(R)]$, $\alpha\in A$.}
\end{align}

Solving this system is as difficult as solving $\csp(\A)$, but if we allow the variables to take arbitrary values in $\ZZ$ or in the rational interval $[0,1]$ then the task can be carried out in polynomial time. We write $\AIP_\A(\I)$ and $\BLP_\A(\I)$ for the system 
$\Ip_\A(\I)$ when the domain of the variables is $\ZZ$ or $[0,1]\subseteq \QQ$ respectively. \par
Given an input structure $\I$ for $\pcsp(\A, \B)$, where $(\A, \B)$ is a finite template, the $\AIP$  algorithm (resp., $\BLP$) \cite{BBKO21} solves $\AIP_\A(\I)$ (resp., $\BLP_\A(\I)$) and accepts $\I$ if and only if this system is satisfiable. The algorithm $\BLP+\AIP$ \cite{BGWZ20} combines the power of the previous two, and checks in polynomial time whether $\AIP_\A(\I)$ and $\BLP_\A(\I)$ have compatible solutions, and accepts $\I$ when this is the case. Compatibility here means that whenever a variable is assigned to $0$ in $\BLP_\A(\I)$, it is also assigned to $0$ in $\AIP_\A(\I)$. An algorithm $\Qcal\in \{ \AIP, \BLP, \BLP+\AIP\}$ solves $\pcsp(\A, \B)$ if it always outputs a correct answer in this problem, meaning that the following two implications hold: (1) if there is a homomorphism $\I \rightarrow \A$ then $\Qcal$ accepts $\I$, and (2) if $\Qcal$ accepts $\I$ then there is a homomorphism $\I\rightarrow \B$. In this situation we define the \emph{rounding problem} $\spcsp_\Qcal(\A, \B)$ as the problem of, given an instance $\I$ accepted by $\Qcal$, to find a homomorphism from $\I$ to $\B$.

\paragraph{Minions}

 A \emph{minion} $\Mscr$ consists of a collection of disjoint sets $\Mscr(n)$ indexed by the natural numbers $n\in \NN$, and a map $\pi^\Mscr:
    \Mscr(n) \rightarrow \Mscr(m)$ 
    for each
    pair of numbers $n,m\in \NN$ and each function $\pi \in [m]^{[n]}$, satisfying that (1) 
     $\pi^{\Mscr}=\id_{\Mscr(n)}$ if $\pi=\id_{[n]}$, and (2) $\pi^\Mscr = \pi_1^\Mscr \circ 
    \pi_2^\Mscr$
    whenever $\pi= \pi_1 \circ \pi_2$. The elements $f\in \Mscr(n)$ are called \emph{$n$-ary}, the functions $\pi^\Mscr$ are called \emph{minoring operations}, and an element $\pi^\Mscr(f)$ is called a \emph{minor} of $f$. When $\Mscr$ is clear from the context, we write $f^\pi$ instead of $\pi^\Mscr(f)$. The minion $\Mscr$ is called \emph{locally finite} if $\Mscr(n)$ is a finite set for all $n\in \NN$.  An element $p\in \Mscr(n)$ is called \emph{cyclic}
    if $p=p^{(n,1,2,\dots,n-1)}$ (we remind the reader that we represent maps $\pi\in [n]^{[n]}$ as tuples). A
    \emph{weak near-unanimity} (WNU) element is some $p\in \Mscr(n)$ satisfying
    that $p^{\sigma_i}=p^{\sigma_j}$  for all $i,j\in [n]$, where $\sigma_i\in [2]^{[n]}$ sends $i$ to $1$ and all other $o\in [n]$ to $2$. Given a template $(\A, \B)$, the \emph{polymorphism minion} $\pol(\A, \B)$ is a minion whose $n$-ary elements are the homomorphisms $f: \A^n \rightarrow \B$, which are called \emph{polymorphisms}. Given an $n$-ary polymorphism $f$, and a map $\pi:[n]\rightarrow [m]$, the minor $f^\pi$ is defined by $f^\pi(\bm a)=f(\bm a\circ \pi)$
for every $f\in A^n$. Finally, a \emph{minion homomorphism} 
    $F:\Mscr \rightarrow \Nscr$
    is a map from elements of $\Mscr$ to elements of $\Nscr$ that preserves arities and minoring operations, i.e., satisfying that
    $F(f)^\pi= F(f^\pi)$
    for each suitable $f,\pi$. Similarly, given $h\in \NN$, a \emph{partial homomorphism}
    $F:\Mscr \overpartialmap{h} \Nscr$ \emph{up to arity $h$} is a partial map defined on all elements $f\in \Mscr$ of arity at most $h$ that preserves arities and minoring operations.     
    \par

We define three minions that characterize the power of the algorithms $\AIP, \BLP, \BLP+\AIP$.
In the minion $\Mscr_\AIP$ the $n$-ary elements are the tuples $f\in \ZZ^n$ 
satisfying $\sum_{i\in [n]} f(i) = 1$. 
Given $f\in \Mscr_\AIP(n)$, and $\pi\in [m]^{[n]}$, minoring is defined by the identity  $(f^\pi)(i) = \sum_{j\in \pi^{-1}(i)} f(j)$ for each $i\in [m]$.
In the minion $\Mscr_\BLP$ the $n$-ary elements are the the tuples $f\in [0,1]^n$ of rational numbers for which $\sum_{i\in [n]} f(i) = 1$. Minoring is defined as for $\Mscr_\AIP$, i.e., by the identity $(f^\pi)(i) = \sum_{j\in \pi^{-1}(i)} f(j)$. Finally, in the minion $\Mscr_{\BLP+\AIP}$ the $n$-ary elements are pairs $(f,g)$, where $f\in \Mscr_\BLP(n), g\in \Mscr_\AIP(n)$, and $f(i)=0$ implies $g(i)=0$ for each $i\in [n]$.
Minoring in $\Mscr_{\BLP+\AIP}$ is defined component wise. That is, $(f, g)^\pi = (f^\pi, g^\pi)$, where $f^\pi= \pi^{\Mscr_\BLP}(f)$, and $g^\pi= \pi^{\Mscr_\AIP}(g)$. The following theorem has been shown in \cite{BBKO21} for the algorithms $\AIP, \BLP$ and in \cite{BGWZ20} for $\BLP+\AIP$.
\begin{theorem}
\label{th:algorithm_characterization}
Let $\Qcal\in \{\AIP, \BLP, \BLP+\AIP\}$, and let $(\A, \B)$ be a finite template. Then the algorithm $\Qcal$ solves $\pcsp(\A, \B)$ if and only if $\Mscr_\Qcal \rightarrow \pol(\A, \B)$.
\end{theorem}

In (non-)computability results we consider the plain encoding of $\Mscr_\BLP, \Mscr_\AIP, \Mscr_{\BLP+\AIP}$, where tuples are represented as comma separated lists, delimited with parentheses, integers are represented with their decimal representations, and rational numbers are represented as irreducible fractions, written as two numbers separated by a forward slash (i.e., $n/m$).

\paragraph{Minor Conditions}
    Another useful way of looking at minions is to consider them as multi-sorted structures \cite{dalmau2024local}, where the sort of each element $p\in \Mscr$ is its arity $\arty(p)$, and $\cdot^\pi$ is a function from $m$-ary elements to $n$-ary elements for each map $\pi\in [n]^{[m]}$. We consider the multi-sorted first-order (FO) language $\Lcal_{\MC}$
    of minions. For some background of multi-sorted (or many-sorted) FO logic we refer to \cite{gallier2015logic}, but we require only the very basics. Formulas in $\Lcal_\MC$ are built using variables, each of which has an arity (i.e., its sort), Boolean connectives, and function symbols $\cdot^\pi$ for each $n,m\in \NN$ and each map $\pi \in [m]^{[n]}$. Variables represent elements of minions, and each symbol $\cdot^\pi$ represents the corresponding minoring operation. We write $\exists^n x$ rather than $\exists x$ to make explicit that $x$ is a $n$-ary variable, and $\phi(x_1^{n_1},\dots,x_k^{n_k})$ to express that the free variables $x_1,\dots, x_k$ of the formula $\phi$ have arities $n_1,\dots, n_k$ respectively. For example, the formula
    $\exists^3 x \left(  x = x^{(2,3,1)} 
    \right)$ expresses the existence of a $3$-ary cyclic element.  The \emph{maximum arity} of a formula $\phi\in \Lcal_\MC$ is the largest maximum arity of any of its sub-terms.  A \emph{primitive positive formula} is one that does not include disjunction, negation, or universal quantification, and a \emph{sentence} (or a \emph{closed formula}) is a formula with no free variables. \emph{Minor conditions} are closed pp-formulas in $\Lcal_\MC$. We remark that more commonly the notion of minor condition is introduced using bipartite Label Cover instances (e.g., \cite{BBKO21}). \par
    Given a minion $\Mscr$, a formula $\phi(x_1,\dots,x_k)\in \Lcal_\MC$, and elements $f_1,\dots, f_k\in \Mscr$ such that
    $f_i$ has the same arity as $x_i$ for all $i\in [k]$, we write $\Mscr\models \phi(f_1, \dots, f_k)$ to express that
    $\Mscr$ satisfies $\phi$ when substituting each $x_i$ with the element $f_i$. A \emph{pp-definition} of a set $Q\subseteq \Mscr$ is a pp-formula $\Phi(x)$ such that $\Mscr\models \Phi(f)$ if and only if $f\in Q$. Suppose that $\Mscr \models \phi$ for a minor condition $\phi$. A \emph{satisfying assignment} of $\phi$ in $\Mscr$ maps each occurrence of the existential quantifier $\exists x \varphi(x)$ in $\phi$ to an element $f$ of $\Mscr$ of the same arity as $x$ in such a way that $\phi$ is satisfied after substituting $\exists x \varphi(x)$ with $\varphi(f)$ in each sub-formula. When there is no ambiguity, we will simply treat assignments as maps from variables of $\phi$ to elements of $\Mscr$.  Given two minions $\Mscr \rightarrow \Nscr$ and a number $h\in \NN$, in the \emph{promise minor condition problem} $\PMC_h(\Mscr, \Nscr)$ we consider an input minor condition $\phi$ whose maximum arity is at most $h$, and the task is to accept it if $\Mscr\models \phi$ and reject it if $\Nscr \not\models \phi$. Additionally, if $\Nscr$ is locally finite, we define the \emph{search promise minor condition problem} $\sPMC_h(\Mscr, \Nscr)$  as the problem of either finding a satisfying assignment of $\phi$ in $\Nscr$ or to reject $\phi$ if $\Mscr\not\models \phi$. \par

\section{Main Results}

We are ready to state our main results about the algorithms $\AIP, \BLP,$ and $\BLP+\AIP$. These are summarized in \Cref{fig:results_algs}. Given $\Qcal\in \{\AIP, \BLP, \BLP+\AIP\}$ and $k\in \NN$,
$\mathcal{S}^\Qcal_k$ denotes the family of finite templates $(\A, \B)$ with $|A|\leq k$ such that $\Qcal$ solves $\pcsp(\A, \B)$.

\begin{theorem}[Main result for AIP]
\label{th:AIP_main}
The following hold:
\begin{enumerate}[label=(\arabic*)]
    \item the family of rounding problems $\spcsp_{\AIP}(\A, \B)$
    for $(\A, \B)\in \mathcal{S}^\AIP_4$ is TFNP-hard,
    \item the family $\spcsp_{\AIP}(\A, \B)$
    for $(\A, \B)\in \mathcal{S}^\AIP_2$ is TFNP$_1$-hard,
    \item  the family $\mathcal{S}^\AIP_2$ is undecidable, and
    \item there is a template $(\A, \B)\in \mathcal{S}^\AIP_2$
    for which there is no computable minion homomorphism from $\Mscr_\AIP$ to $\pol(\A, \B)$.  
\end{enumerate}
\end{theorem}

\begin{theorem}[Main result for BLP]
\label{th:BLP_main}
The following hold:
\begin{enumerate}[label=(\arabic*)]
    \item the family of rounding problems $\spcsp_{\BLP}(\A, \B)$
    for $(\A, \B)\in \mathcal{S}^\BLP_5$ is TFNP-hard, 
    \item the family $\spcsp_{\BLP}(\A, \B)$
    for $(\A, \B)\in \mathcal{S}^\BLP_2$ is TFNP$_1$-hard,
    \item the family $\mathcal{S}^\BLP_2$ is undecidable, and
    \item  there is a template $(\A, \B)\in \mathcal{S}^\BLP_2$
    for which there is no computable minion homomorphism from $\Mscr_\BLP$ to $\pol(\A, \B)$.  
\end{enumerate} 
\end{theorem}

\begin{theorem}[Main result for $\BLP + \AIP$]
\label{th:BLP+AIP_main}
The following hold:
\begin{enumerate}[label=(\arabic*)]
    \item the family of rounding problems $\spcsp_{\BLP+\AIP}(\A, \B)$
    for $(\A, \B)\in \mathcal{S}^{\BLP+\AIP}_5$ is TFNP-hard, and 
    \item the family $\mathcal{S}^{\BLP+\AIP}_5$ is undecidable.
\end{enumerate}    
\end{theorem}

We also recall here our main result about cyclic polymorphisms and WNUs. We remark that $\pol(\A, \B)$ admitting
some cyclic polymorphism is equivalent to it admitting one of some prime arity. 

\minors*

\section{Overview of the Proofs}
\label{sec:warm-up}
We sketch the proof of the following theorem while outlining the ideas used in our main results. 
\begin{theorem}
\label{th:warm-up}
    Let $\mathcal{S}$ be the set of finite templates $(\A, \B)$ with $|A|\leq 3$ such that $\pcsp(\A, \B)$ is solved by $\AIP$. Then (1) the family of problems $\spcsp_{\AIP}(\A, \B)$ where $(\A, \B)\in \mathcal{S}$ is TFNP$_1$-hard, (2) $\mathcal{S}$ is undecidable, and (3) there is a template
    $(\A, \B)\in \mathcal{S}$ for which there is no computable homomorphism $F: \Mscr_\AIP \rightarrow \pol(\A, \B)$.
\end{theorem}

The proof of this result is a reduction from tiling problems. The signature $\Sigma_\Gamma$  consists of a unary symbol $O$ and binary symbols $E_1,E_2$, and $\bm \Gamma$ is the $\Sigma_\Gamma$-structure whose universe is the upper-right quadrant $\NN^2$, and where $O^\Gamma=\{ (1,1) \}$ marks the origin, and $E_1^\Gamma=\{ ((m,n),(m+1,n)) \, \vert \, (m,n)\in \NN^2 \}$,  $E_2^\Gamma=\{ ((m,n),(m,n+1)) \, \vert \, (m,n)\in \NN^2 \}$ 
are the horizontal and vertical adjacency relations. We write $\Hom(\bm \Gamma, \cdot)$ for the set of finite structures $\bT$ satisfying $\bm \Gamma \rightarrow \bT$. A $\Sigma_\Gamma$-structure $\bT$ can be seen as a set of tiles, equipped with horizontal and vertical constraints, and a set of distinguished tiles that can occupy the initial position. This way, a homomorphism $F: \bm \Gamma \rightarrow \bT$ represents a tiling of the upper-right quadrant of the plane that satisfies all the imposed restrictions.

\begin{restatable}{proposition}{gridprop}
\label{prop:grid}
The following hold:
\begin{enumerate}[label=(\arabic*)]
    \item the family of problems $\spcsp(\bm \Gamma, \bT)$ where $\bT\in \Hom(\bm \Gamma, \cdot)$ is TFNP$_1$-hard, 
    \label{item:grid_hardness}
    \item $\Hom(\bm \Gamma, \cdot)$ is undecidable, and
    \label{item:grid_undec}
    \item there exists $\bT\in \Hom(\bm \Gamma, \cdot)$ for which there is no computable homomorphism $F: \bm \Gamma \rightarrow \bT$.
    \label{item:grid_noncomp}
\end{enumerate}
\end{restatable}

At a high-level, this is follows from the fact that, given a non-deterministic Turing machine $M$, we can construct a finite structure $\bT_M$ for which the homomorphisms $F:\bm \Gamma \rightarrow \bT_M$ represent non-halting runs of $\bT_M$. \par
In order to prove \Cref{th:warm-up} using this proposition we first represent $\bm \Gamma$ inside the minion $\Mscr_\AIP$, and then we find a way of encoding finite $\Sigma_\Gamma$-structures in finite templates $(\A, \B)$ in a suitable way. The first task is the simpler one. We can identify each pair $\bm m= (m_1,m_2)$ in $\NN^2$ with the $3$-ary element $ f_{\bm m} = (m_1,m_2, 1-m_1 -m_2)\in \Mscr_\AIP$. Similarly, we can represent each pair $(\bm m, \bm m^\prime)\in E_i^\Gamma$ with the $4$-ary element $g_{E_i(\bm m, \bm m^\prime)}=(m_1,m_2,1, - m_1 - m_2)$ with the idea that $f_{\bm m}= g_{E_i(\bm m, \bm m^\prime)}^{(1,2,3,3)}$ and 
$f_{\bm m^\prime}= g_{E_i(\bm m, \bm m^\prime)}^{(1,2,i,3)}$ in mind. We call this way of representing a relational structure inside a minion an \emph{interpretation} (\Cref{sec:interpretations}) and we denote it by $\Ical$. We define $U^\Ical= \{ f_{\bm m} \, \vert \, \bm m\in \NN^2\}$, $O^\Ical=\{ f_{(1,1)} \}$, and $E_i^\Ical = \{ g_{E_i(\bm m, \bm m^\prime)} \, \vert \, (\bm m, \bm m^\prime)\in E_i^\Gamma \}$ for $i = 1,2$.
\par

The interpretation $\Ical$ is \emph{almost} pp-definable in $\Mscr_\AIP$. Indeed, the element $(1,0)\in \Mscr_\AIP(2)$ is the only witness of the pp-formula $\phi_1(x^2)\equiv x=x^{(1,1)}$, so the $4$-ary elements of the form $(m_1,m_2,1, - m_1 - m_2)\in \Mscr_\AIP$ are precisely the witnesses of $\phi_E(y^4)\equiv \phi_1(y^{(2,2,1,2)})$. However, this does not take into account the restriction that $m_1, m_2\in \NN$. Under closer inspection it is not hard to see that the set of elements $(m,n)\in \Mscr_\AIP$ with $m>0$ is not pp-definable, so, in fact, the interpretation $\Ical$ itself is not pp-definable either \footnote{This could be circumvented by having defined $\bm \Gamma$ on the whole integer plane $\ZZ^2$ instead of $\NN^2$, but we want to illustrate the point that our reductions are not uniform pp-definitions. This will be useful in the more involved proofs.}. However, we get something almost as good: for any element $\bm m\in \NN^2$ we can construct in polynomial time a pp-definition $\psi_{\bm m}(x^3)$ of $f_{\bm m}$ in $\Mscr_\AIP$. For $\bm m=(1,1)$, we simply define $\psi_{\bm m}(x^3)\equiv \phi_1(x^{(1,2,2)}) \wedge \phi_1(x^{(2,1,2)})$. 
Given $\bm m\in \NN^2$, for $\bm m^\prime=(m_1+1,m_2)$, we define
$\psi_{\bm m^\prime}(x^3)$ as the formula
\begin{equation}
\bigexists^{4} y \bigexists^3 z \left(\psi_{\bm m}(z) \wedge \phi_{E}(y) \wedge z = y^{(1,2,3,3)} \wedge
x = y^{(1,2,1,3)}\right), \label{eq:reference_warm-up}
\end{equation}
and so on. This way, given a finite substructure $\bm G \subset \bm \Gamma$, we can construct in polynomial time a minor condition $\Psi_{\bm G}$ containing an existentially-quantified variable $x_{\bm m}$ for each vertex $\bm m\in G$ an existentially-quantified variable $y_{E_i(\bm m, \bm m^\prime)}$ for each edge $(\bm m, \bm m^\prime)\in E_i^G$, such that the satisfying assignments of $\Psi_{\bm G}$ on $\Mscr_\AIP$ must map 
each variable to the minion element that encodes the corresponding vertex or edge. That is, 
$x_{\bm m}\mapsto f_{\bm m}$ and
$y_{E_i(\bm m, \bm m^\prime)} \mapsto g_{E_i(\bm m, \bm m^\prime)}$. The minor condition $\Psi_{\bm G}$ can be defined as
\begin{multline}
\label{eq:pattern_warm-up}
\bigexists_{\bm m \in G}^3 x_{\bm m}
\bigexists_{i\in [2], (\bm m, \bm m^\prime)\in E_i^G}^4
y_{E_i(\bm m, \bm m^\prime)} 
\left( \bigwedge_{\bm m\in G}
\psi_{\bm m}(x_{\bm m})
\right)  \bigwedge \\
 \Biggl( \bigwedge_{i\in [2], (\bm m, \bm m^\prime)\in E_i^G}
\psi_{E}(y_{E_i(\bm m,\bm m^\prime)}) \wedge  x_{\bm m}= y_{E_i(\bm m,\bm m^\prime)}^{(1,2,3,3)}
\wedge  x_{\bm m^\prime}= y_{E_i(\bm m,\bm m^\prime)}^{(1,2,i,3)} 
\Biggr).
\end{multline}
This construction is called a \emph{pattern} (\Cref{sec:patterns}). Having found a nice way to represent $\bm \Gamma$ inside $\Mscr_\AIP$, the next step is to develop a construction that, given a finite $\Sigma_\Gamma$-structure $\bT$, builds a suitable finite template $(\A_{\bT}, \B_{\bT})$. We would like $\Mscr_{\bT}= \pol(\A_{\bT}, \B_{\bT})$ to satisfy the following.
 \begin{ppty}[Conditions that enable our reductions] ~ \\
 \vspace{-1em}
    \label{ppty:warm-up_conditions}
    \begin{enumerate}[label=({\Roman*})]
        \item
        \label{ppty:warm-up_conditions_item1}
        $\Mscr_\AIP\rightarrow \Mscr_{\bT}$ if and only if $\bm \Gamma \rightarrow \bT$.
        \item
        \label{ppty:warm-up_conditions_item2}
        A homomorphism $F: \bm \Gamma \rightarrow \bT$ can be computed given oracle access to a homomorphism $H: \Mscr_\AIP\rightarrow \Mscr_{\bT}$.
        \item
        \label{ppty:warm-up_conditions_item3}
        Given a finite substructure $\bm G \subset \bm \Gamma$, there is a 
        polynomial-time reduction from the task of finding  a homomorphism $F: \bm G \rightarrow \bT$ 
        to the task of finding a satisfying assignment  of $\Psi_{\bm G}$ in $\Mscr_{\bT}$.
    \end{enumerate}
    \end{ppty}
  It is easy enough to see that the conditions \ref{ppty:warm-up_conditions_item1} and \ref{ppty:warm-up_conditions_item2} allow us to respectively reduce the undecidability and non-computability parts of \Cref{th:warm-up} to those of \Cref{prop:grid}. In to bridge the gap between the TFNP$_1$-hardness statements in those results, we use condition \ref{ppty:warm-up_conditions_item3} together with the following additional result, which is a consequence of \Cref{th:spcsp_to_spmc} (\Cref{sec:reductions}). We note that this is a standard modification of an analogous result, \cite[Theorem 3.12]{BBKO21}, that has been applied widely in the PCSP literature. 
\begin{proposition}
\label{prop:PMC_reduction_warm-up}
    Suppose that $\AIP$ solves $\pcsp(\A, \B)$ for a finite template $(\A, \B)$. Then $\spcsp_\AIP(\A,\B)$ is log-space equivalent to $\sPMC_N(\Mscr_\AIP, \pol(\A, \B))$, where $N$ is at least as large as $|A|$ and $|R^A|$ for all relation symbols $R$.
\end{proposition}
We remark that we do not always achieve the conditions from \Cref{ppty:warm-up_conditions} in tandem: for example, sometimes we obtain templates that satisfy both \ref{ppty:warm-up_conditions_item1} and 
\ref{ppty:warm-up_conditions_item3}, but not \ref{ppty:warm-up_conditions_item2}, meaning that we are able to obtain undecidability, and hardness-of-rounding results for a given algorithm, but not a related non-computability result. This is the case, for example, of $\BLP+\AIP$. With our methods, the condition \ref{ppty:warm-up_conditions_item1} used to prove undecidability has the weakest requirements, while the conditions \ref{ppty:warm-up_conditions_item2} and \ref{ppty:warm-up_conditions_item3} used to show non-computability and hardness have non-comparable requirements. \par

Rather than constructing the template $(\A_{\bT}, \B_{\bT})$ directly, we focus on its polymorphism minion instead.
We start by defining a minion $\Nscr_{\bT}$ that satisfies \Cref{ppty:warm-up_conditions} of $\pol(\A_{\bT}, \B_{\bT})$. The $n$-ary elements of $\Nscr_{\bT}$ are pairs $(f,\chi)$, where $f\in \Mscr_{\AIP}(n)$, and $\chi: [3]^{[n]}\partialmap T$ is a partial map
defined on the elements $\gamma\in [3]^{[n]}$ such that $f^\gamma\in U^\Ical$ that satisfies the following properties:  if $f_{(1,1)}=f^\gamma$, then $\chi(\gamma)\in O^T$, and if $g_{E_i(\bm m, \bm m^\prime)}=f^\pi$, then 
$(\chi(\gamma), \chi(\gamma^\prime))\in E_i^T$, where $\gamma=(1,2,3,3)\circ \pi$, and $\gamma^\prime=(1,2,i,3)\circ \pi$. Intuitively, the element $f\in \Mscr_{\AIP}$ \emph{covers} some part of $\bm \Gamma$ (as interpreted by $\Ical$), which consists of the elements $f_{\bm m}\in U^\Ical$ that can be obtained by minoring $f$. Similarly, we also think of $f$ as covering the edges $(\bm m, \bm m^\prime)\in E_i^\Gamma$ where $g_{E_i(\bm m, \bm m^\prime)}\in E_i^\Ical$ is a minor of $f$. Then, the map $\chi$ represents a homomorphism from the substructure $\bLS_f\subseteq \bm \Gamma$ covered by $f$ to $\bT$. Given a map $\pi\in [m]^{[n]}$, we define the minor $(f,\chi)^{\pi}$ as the pair $(f^\pi, \chi^\pi)$, where $\chi^\pi: [3]^{[m]} \partialmap T$ is the partial map defined by $\gamma \mapsto \chi(\gamma \circ \pi)$. The rationale behind this construction is that the structure $\bLS_{f^\pi}$ covered by $f^\pi$ is a substructure of $\bLS_{f}$, so $\chi^\pi$ is defined so that it represents the restriction of $\chi$ to $\bLS_{f^\pi}$.
If the substructures $\bLS_f$ formed a topology over $\bm \Gamma$ (which, we remark, is not the case), then the minion $\Nscr_{\bT}$ would encode a \emph{sheaf} \cite{bredon2012sheaf} over $\bm \Gamma$ whose \emph{sections} correspond to partial homomorphisms to $\bT$. We keep this topological intuition in mind throughout the proof. We call $\Nscr_{\bT}$ the \emph{manifold minion} \footnote{In previous versions of this work, this was called the exponential minion. We changed the name to appease a  category-theoretically inclined reader who would point out that these minions are not, in fact, exponential objects in the category of minions.} given by the interpretation $\Ical$ and the structure $\bT$ (\Cref{sec:manifold_minion}). \par
\begin{proposition}
    The minion $\Nscr_{\bT}$ satisfies \Cref{ppty:warm-up_conditions}.
\end{proposition}
Let us sketch the proof of this fact. It is not difficult to show that every element $f\in \Mscr_\AIP$ is pp-definable, so the only minion homomorphism $F: \Mscr_\AIP \rightarrow \Mscr_\AIP$ is the identity.  Hence, any homomorphism $F:\Mscr_\AIP \rightarrow \Nscr_{\bT}$ composed with the left projection must yield the identity over $\Mscr_\AIP$. In other words, $F$ must be of the form $f \mapsto (f,\chi_f)$. If such homomorphism $F$ exists, then the charts $\chi_f$ must be compatible with each other, yielding a global homomorphism from $\bm \Gamma$ to $\bT$. Conversely, if there is a homomorphism $H:\bm \Gamma \rightarrow \bT$, one can restrict it to each local structure $\bLS_f$ to obtain compatible charts $\chi_f$ and define a minion homomorphism $F$. This proves $\Nscr_{\bT}$ satisfies \ref{ppty:warm-up_conditions_item1}. To show \ref{ppty:warm-up_conditions_item2}, suppose that we have oracle access to a homomorphism $F:\Mscr_\AIP \rightarrow \Nscr_{\bT}$. Then, by the previous reasoning, in order to compute a homomorphism $H:\bm \Gamma \rightarrow \bT$, given an element $\bm m\in \NN^2$ we just need to query $F(f_{\bm m})=(f_{\bm m}, \chi_{f_{\bm m}})$ and consider the chart $\chi_{f_{\bm m}}$. Finally, to see that \ref{ppty:warm-up_conditions_item3} holds, consider a finite substructure $\bG \subset \bm \Gamma$, and a satisfying assignment $x\mapsto (f_x,\chi_x)$ of $\Psi_{\bG}$ in $\Nscr_{\bT}$. The map $x\mapsto f_x$
is a satisfying assignment of $\Psi_{\bG}$ in $\Mscr_\AIP$, so by construction of $\Psi_{\bG}$ it must hold that $f_{x}=f_{\bm m}$ when $x=x_{\bm m}$ for each $\bm m\in G$, and $f_{x}=g_{E_i(\bm m, \bm m^\prime)}$ when $x= y_{E_i(\bm m, \bm m^\prime)}$ for each $i\in [2]$, $(\bm m, \bm m^\prime)\in E_i^G$. Hence, when $x$ ranges over all variables $\{ x_{\bm m}\}_{\bm m\in G} \cup \{ y_{E_i(\bm m, \bm m^\prime)}\}_{i\in [2],(\bm m, \bm m^\prime)\in E_i}$, the elements $f_x$ cover the whole structure $\bG$ (as induced by the interpretation $\Ical$), and the charts $\chi_x$ must 
piece together a homomorphism from $\bG$ to $\bT$. \par

The main issue at this point is that, despite having the desired properties, the minion $\Nscr_{\bT}$ is not (isomorphic to) the polymorphism minion of any finite template $(\A_{\bT}, \B_{\bT})$. In fact, it is not even locally finite. The first step towards constructing $(\A_{\bT}, \B_{\bT})$ from the minion $\Nscr_{\bT}$ is to find a locally finite quotient $\Nscr^\prime_{\bT}$ of $\Nscr_{\bT}$ that still satisfies \Cref{ppty:warm-up_conditions}. More precisely, $\Nscr^\prime_{\bT}$ will be obtained by performing the manifold construction on a locally finite quotient of $\Mscr_\AIP$ by some equivalence relation $\sim$. Let us spell this out. The equivalence $\sim$ defines a quotient minion $\Mscr_\AIP/\sim$ where each element $\qclass{f}$ is the $\sim$-class of an element $f$ in the original minion $\Mscr_\AIP$. Similarly, given a subset $S\subseteq \Mscr_\AIP$, we write $\qclass{S}$
for the subset of $\Mscr_\AIP/\sim$ consisting of the elements $\qclass{f}$ for each $f\in S$.
The relation $\sim$ also induces an equivalence on $\bm \Gamma$ through the interpretation $\Ical$ defined as $\bm m \sim \bm m^\prime$ whenever $f_{\bm m}\sim f_{\bm m^\prime}$. Finally, we can also speak of the quotient interpretation $\Ical/\sim$ over the quotient minion $\Mscr_\AIP/\sim$, which is given by $U^{\Ical/\sim} = \qclass{U^\Ical}$, $O^{\Ical/\sim}= \qclass{O^\Ical}$, and $E_i^{\Ical/\sim} = \qclass{E_i^\Ical}$ for $i = 1,2$. It can be shown that $\Ical/\sim$ induces the quotient structure $\Gamma/\sim$ on the quotient minion $\Mscr_\AIP/\sim$. 
Then, we define $\Nscr^\prime_{\bT}$ as the manifold minion given by the interpretation $\Ical/\sim$
and the structure $\bT$, similarly to $\Nscr_{\bT}$. \par
The following conditions ensure that $\Nscr^\prime_{\bT}$ still satisfies \Cref{ppty:warm-up_conditions} and that this property can be further preserved when transforming $\Nscr^\prime_{\bT}$ into a polymorphism minion in the final step. \par

    \begin{ppty}[Requirements for a good quotient]~\\
    \vspace{-1em}
    \label{ppty:warm-up_conditions2}
    \begin{enumerate}[label=({\Alph*})]
    \item \emph{The interpretation $\Ical$ is $\sim$-stable.} It holds that  $\qclass{f}\in \qclass{U^\Ical}$ if and only if $f\in U^\Ical$, 
    $\qclass{f}\in \qclass{O^\Ical}$ if and only if $f\in O^\Ical$, and 
    $\qclass{f}\in \qclass{E_i^\Ical}$ if and only if $f\in E_i^\Ical$ for all $i\in [2]$.
    \label{ppty:warm-up_conditions2_item0}
    \item \emph{The interpretation $\Ical$ is internal with respect to $\sim$ at arity $4$.}
        For each partial homomorphism
        $F: \Mscr_\AIP \partialmap \Mscr_\AIP/\sim$ defined up to arity $4$, 
        it holds that
        $F(U^\Ical)\subseteq \qclass{U^\Ical}$, 
        $F(O^\Ical)\subseteq \qclass{O^\Ical}$,
        and
        $F(E_i^\Ical)\subseteq \qclass{E_i^\Ical}$ for each $i\in [2]$.
        \label{ppty:warm-up_conditions2_item1}
        \item \emph{The map $\bm G \mapsto \Psi_{\bm G}$ defined \eqref{eq:pattern_warm-up}
        is an internal pattern with respect to $\sim$.
        } Given a finite substructure $\bm G \subset \bm \Gamma$, any satisfying assignment of $\Psi_{\bm G}$ in $\Mscr_{\AIP}/\sim$ must correspond to a homomorphism from $\bm G$ into $\Gamma/\sim$\footnote{We warn the reader that, for the sake of exposition in this introductory section, the statement in \ref{ppty:warm-up_conditions2_item2} is not fully precise.}.
        \label{ppty:warm-up_conditions2_item2}
    \end{enumerate}
    \end{ppty}

Let us briefly go over these conditions. Condition \ref{ppty:warm-up_conditions2_item0} ensures that if there is a homomorphism $F:\Gamma \rightarrow \bT$ then there is a minion homomorphism
$H:\Mscr_{\AIP} \rightarrow \Nscr^\prime$. Indeed, in this case we can define $H$ as the map
$f\mapsto (\qclass{f}, \chi_f)$, where $\chi_f$ is obtained by looking at the restriction of $F$ to the local structure induced by $\Ical$ on $f$, and then factoring the local homomorphism through the quotient. The reason this construction works is that the quotient map $\bm\Gamma \rightarrow \bm \Gamma/\sim$ preserves non-relations (i.e., sends non-edges to non-edges and sends non-origin elements to non-origin elements) so partial homomorphisms on local structures can be factored through the quotient \footnote{Making this statement formal requires being more careful in our definition of local structures. The direct approach, which we follow in this proof sketch, of defining the local structure $\bS_f$ as the substructure of $\bm \Gamma$ induced on the minors of $f$ does not work. The issue is that the quotient map $f\mapsto \qclass{f}$ might identify minors of $f$, which prevents us from factoring partial homomorphisms through the quotient.}. Condition \ref{ppty:warm-up_conditions2_item1} implies that every minion homomorphism $H:\Mscr_\AIP \rightarrow \Nscr^\prime_{\bT}$ corresponds to a homomorphism $F:\bm\Gamma \rightarrow \bT$. Indeed, intuitively, this condition means that every homomorphism $\Mscr_\AIP \rightarrow \Mscr_\AIP/\sim$ must induce a homomorphism between the structures $\bm \Gamma$ and $\bm \Gamma/\sim$ interpreted on those minions. So, if the homomorphism $H$
maps $f\in \Mscr_\AIP$ to some pair $(\qclass{g}, \chi)\in \Nscr^\prime_{\bT}$, then it must hold that the local structure 
$\bGS_{\qclass{g}}$ given by $\qclass{g}$ in $\Mscr_\AIP/\sim$ is a homomorphic image of the local structure $\bGS_f$ of $f$ in $\Mscr_\AIP$, meaning that the local homomorphism $\chi: \bGS_{\qclass{g}}\rightarrow \bT$ can be lifted to another one $\tilde{\chi}: \bGS_f \rightarrow \bT$. It can be seen that these lifted local homomorphisms form a consistent family, so they define a global homomorphism from $\bm \Gamma$ to $\bT$. Finally, condition \ref{ppty:warm-up_conditions2_item2} requires that any assignment of $\Psi_{\bm G}$ in $\Mscr_{\AIP}/\sim$ maps, for each $\bm m\in G$, the variable $x_{\bm m}$ to some element $\qclass{f_{\bm m^\prime}}$ in such a way that the map $\bm m \mapsto \qclass{\bm m^\prime}$ is a homomorphism from $\bG$ to $\bm \Gamma/\sim$. Arguing again that you can factor local homomorphism through the quotient, this condition allows us to transfer \ref{ppty:warm-up_conditions_item3} in \Cref{ppty:warm-up_conditions} from the minion $\Nscr_{\bT}$ to $\Nscr^\prime_{\bT}$.
\par

Another important observation is that \ref{ppty:warm-up_conditions2_item1} would be weaker if we considered non-partial homomorphisms instead. Indeed, every homomorphism can be restricted to a partial homomorphism, but not every partial homomorphism can, in principle, be extended to a fully-defined homomorphism. 
We consider partial homomorphisms planning for a future step where, in order to obtain a polymorphism minion, we will have to let go of high-arity terms in a certain way. The choice of arity $4$ corresponds
to the fact that the minor conditions $\Psi_{\bm G}$ have maximum arity $4$. \par

In order to construct the relation $\sim$ we select some relevant \emph{predicates} (related to $\Ical$) containing the information that we would like to preserve in the quotient. These predicates will be sets of binary elements $D_1, D_{\NN}\subseteq \Mscr_\AIP(2)$ defined as $D_1=\{(1,0)\}$, and $D_\NN=\{(m,1-m) \, \vert \, m\in \NN \}$. We call the set of predicates $\Dcal=\{ D_1, D_\NN\}$ a \emph{description} (\Cref{sec:descriptions}). Then, we write $f\sim_\Dcal g$ for two elements $f, g\in \Mscr_\AIP$ of the same arity $n$ whenever for all $\pi\in [2]^{[n]}$ and all $P\in \Dcal$ the inclusion $f^\pi\in P$ holds if and only if $g^\pi\in P$ as well. 
We shorten expressions of the form $\, \cdot \, /\sim_{\Dcal}$ to $\, \cdot\, / \Dcal$. \Cref{fig:quotient_warm-up}
displays the quotient $\bm \Gamma/\Dcal$.
\begin{fact}
    The equivalence relation $\sim_\Dcal$ satisfies \Cref{ppty:warm-up_conditions2}.
\end{fact}

The fact that $\sim_\Dcal$ meets condition \ref{ppty:warm-up_conditions2_item0} in \Cref{ppty:warm-up_conditions2} follows easily from the definition of this relation. 
Items \ref{ppty:warm-up_conditions2_item1} and \ref{ppty:warm-up_conditions2_item2}, hold roughly, because the formula
$\psi_{\bm m}(x)$ given in \eqref{eq:reference_warm-up}
is a pp-definition of $\qclass{f_{\bm m}}$ in $\Mscr_\AIP/\Dcal$ for each $\bm m\in\NN^2$
\footnote{This is a stronger statement than what we really need. What we will use is the existence of
a kind of formulas which we call ``internal references''. These are introduced in \Cref{sec:descriptions}.}, 
and there are analogous pp-definitions for the elements $\qclass{g_{E_i(\bm m_1, \bm m_2)}}$ representing each edge. Let us sketch this fact. First, we can see that $\phi_1(x^2)\equiv x= x^{(1,1)}$ defines the element $\qclass{(1,0)}$. Indeed,
if $\Mscr_\AIP/\Dcal \models \phi_1(\qclass{f})$ then $f\sim_\Dcal f^{(1,1)}$. However, $f^{(1,1)}\in D_1$ for all $f\in \Mscr_\AIP(2)$, so $f$ must belong to $D_1$ as well, meaning that $f=(1,0)$. Following this argument, 
we see that $\psi_{(1,1)}(x)$ defines $\qclass{f_{(1,1)}}$, and the other values of $\bm m\in \NN^2$ can be handled by induction. \par
The fact that $\psi_{\bm m}(x)$ is a pp-definition of $\qclass{f_{\bm m}}$ in $\Mscr_\AIP/\Dcal$ forces any partial homomorphism $H: \Mscr_{\AIP}\partialmap \Mscr_{\AIP}/\Dcal$ defined up to arity $4$ (which is the maximum arity of $\phi_{\bm m}(x)$) to satisfy $H(f_{\bm m})=\qclass{f_{\bm m}}$ for all $\bm m \in \NN^2$. Indeed, it holds that $\Mscr_\AIP \models \psi_{\bm m}(f_{\bm m})$, and this is witnessed by elements of arity at most $4$, so $\Mscr_\AIP/\Dcal \models \psi_{\bm m}(H(f_{\bm m})$, and it must be that $H(f_{\bm m})=\qclass{f_{\bm m}}$. This shows that condition \ref{ppty:warm-up_conditions2_item1} holds. Finally, one can also show that \ref{ppty:warm-up_conditions2_item2} holds via a similar argument using the fact that the minor conditions $\Psi_{\bG}$ are built-up from the formulas $\psi_{\bm m}(x)$. \par

Having obtained a good equivalence relation, we construct the locally finite minion $\Nscr^\prime_{\bT}$ as outlined previously. To reiterate, we consider the quotient interpretation $\Ical/\Dcal$ on the quotient minion $\Mscr_\AIP/\Dcal$ and define $\Nscr^\prime_{\bT}$ as the manifold minion given by $\Ical/\Dcal$ and the structure $\bT$. The fact that $\sim_\Dcal$ satisfies \Cref{ppty:warm-up_conditions2}, means that $\Nscr^\prime_{\bT}$
manages to preserve \Cref{ppty:warm-up_conditions}, which $\Nscr_{\bT}$ previously satisfied. However, we are still not done: The minion $\Nscr^\prime_{\bT}$ is locally finite, but it may not be isomorphic to a polymorphism minion. The reason is that we do not know whether it is finitizable in the sense of \cite{BG21:sicomp}. We circumvent this issue by proving the following.

\begin{proposition}
\label{prop:warm-up_minion_to_template}
    There is a finite template $(\A, \B)$ with $|A|=3$ and $|R^A|\leq 4$ for any relation symbol $R$ satisfying that (1) there is a partial minion isomorphism $F:\Nscr^\prime_{\bT} \partialmap \pol(\A, \B)$ defined up to arity $4$, and (2) for any minion $\Mscr$, it holds that $\Mscr \rightarrow \pol(\A, \B)$ if and only if 
    there is a partial homomorphism $H:\Mscr \partialmap \Nscr^\prime_{\bT}$ defined up to arity $4$.
\end{proposition}

This is a consequence of the more general result \Cref{th:minion_to_template} (\Cref{sec:minion_closures}). 
The template $(\A, \B)$ is constructed as follows. We consider a single relation symbol $R$ of arity $3^4$. Then $\A$ is \emph{the most general structure} on $3$ elements containing a single relation $R^A$ of size $4$, following a construction given in \cite{BG21:sicomp}. The fact that $|A|=3$ in this result has to do with the fact that any two elements $f,g\in \Nscr^\prime_{\bT}$ of the same arity are equal if and only if all their $3$-ary minors are equal. Hence, $\Nscr^\prime_{\bT}$ can be seen as a function minion on a domain of size $3$. The structure $\B$ is defined as the \emph{free-structure} \cite{BBKO21} of $\Nscr^\prime_{\bT}$ generated by $\A$. \par

The fact that we do not obtain a full isomorphism in \Cref{prop:warm-up_minion_to_template} may seem concerning, but we already accounted for this by considering partial homomorphisms in \Cref{ppty:warm-up_conditions2}. Crucially, the second item in this proposition, ensures that $\Mscr_\AIP \rightarrow \pol(\A, \B)$ exactly when there is a partial homomorphism $H: \Mscr_\AIP \partialmap \Nscr^\prime_{\bT}$ defined up to arity $4$. Putting together some of the implications we have shown so far shows that $\pol(\A, \B)$ satisfies condition \ref{ppty:warm-up_conditions_item1} in \Cref{ppty:warm-up_conditions}, and the rest of the conditions follow similarly.  
This way, the following holds and finishes the proof sketch of \Cref{th:warm-up}.

\begin{proposition}
    Let $(\A, \B)$ be a template satisfying the statement of \Cref{prop:warm-up_minion_to_template}. Then $\pol(\A, \B)$ satisfies \Cref{ppty:warm-up_conditions}.
\end{proposition}

\begin{figure*}[ht]
    \centering
    \bgroup
\def\arraystretch{2}
\begin{tabular}{ |l| c |c | c | } 

 \hline
 \diagbox[width=15em]{ Result }{Algorithm $\Qcal$} & $\Qcal=$ \AIP & $\Qcal=$ \BLP & $\Qcal=$ \BLP+\AIP \\
 \hline 
 \makecell{ Undecidability of $\Mscr_\Qcal\rightarrow \pol(\A, \B)$} & $|A|=2$ & $|A|=2$ & $|A|=5$ \\[5pt] \hline
\makecell{Non-computability of  homomorphisms \\ $F: \Mscr_\Qcal\rightarrow \pol(\A, \B)$} & $|A|=2$ & $|A|=2$ & $--$ \\[5pt]
 \hline
 \makecell{TFNP$_1$-hardness of $\spcsp_\Qcal(\A, \B)$} & $|A|=2$ & $|A|=2$ & $|A|=5$ \\[5pt]
 \hline
 \makecell{ TFNP-hardness of $\spcsp_\Qcal(\A, \B)$} & $|A|=4$ & $|A|=5$ & $|A|=5$ \\[5pt]
 \hline
\end{tabular}
\egroup
 \caption{Main algorithmic results.}
 \label{fig:results_algs}
\end{figure*}

\begin{figure}[ht]
    \centering
    \includegraphics[width=0.4\linewidth]{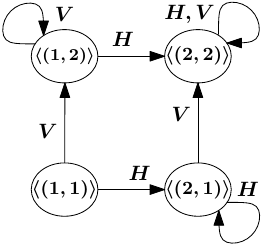}
    \caption{The quotient $\bm \Gamma/\sim_\Dcal$}
    \label{fig:quotient_warm-up}
\end{figure}

\subsection{Organization of the Paper}
The rest of the paper follows roughly the same structure as the proof sketch in this section.
We begin in \Cref{sec:sources} by introducing the sources of hardness, undecidability and non-computability that are the starting point of our reductions. 
In \Cref{sec:definitions} we properly introduce the notions required to prove the main results, including \emph{interpretations}, \emph{manifold minions}, \emph{descriptions}, and \emph{patterns}. Each subsection defines some concepts and includes proofs of related auxiliary results. In \Cref{sec:minion_closures} we describe how to obtain polymorphism minions that are partially isomorphic to a given locally-finite minion. In \Cref{sec:reductions} we show the reductions that allow us to transfer undecidability, non-computability, and hardness results from tiling problems to rounding problems. In \Cref{sec:main_proofs} we put everything together and prove our main results. Finally \Cref{sec:discussion} discusses in greater depth the link between our results and some open questions in the area, outlining some research directions.

\section{Sources of Undecidability, Non-Computability and Hardness}
\label{sec:sources}

Given a family $\mathbb{S}$ of similar structures, $\Hom(\mathbb{S}, \cdot)$ denotes the set of finite structures
satisfying $\bS \rightarrow \I$ for some $\bS\in \mathbb{S}$. When $\mathbb{S}$ is a singleton family $\{\bS \}$, we write $\Hom(\bS, \cdot)$ rather than $\Hom(\{ \bS \}, \cdot)$. We write $\Hom_{\eg}(\mathbb{S}, \cdot)$, where $\eg$ stands for \emph{eventually globally},
for the set of finite structures $\I$ satisfying $\bS\rightarrow \I$ for all but finitely many $\bS\in \mathbb{S}$. Finally, we define $\Hom_{\io}(\mathbb{S}, \cdot)$, where $\io$ stands for \emph{infinitely often}, for the set of finite structures $\I$ satisfying $\bS\rightarrow \I$ holds for infinitely many $\bS\in \mathbb{S}$.
\par
When $\mathbb{S}$ is a finite family of finite structures, or $\bS$ is a finite structure, all the sets described above can easily be recognized in polynomial time.
However, when $\mathbb{S}$ is an infinite family, or $\bS$ is an infinite structure, the previous problems can be undecidable. We deal mostly with these later cases.

Let $\bS$ and $\bT$ be two similar structures. We say that $\bS$ and $\bT$ are \emph{finitely equivalent} if
$\I \rightarrow \bS$ and $\I \rightarrow \bT$ are equivalent conditions 
for every finite structure $\I$. A standard argument shows that, given a finite structure $\I$, 
the fact that $\bS \rightarrow \I$ is equivalent to $\bm G \rightarrow \I$ 
for every finite substructure $\bm G \subseteq \bS$.  See \cite[Remark 7.13]{BBKO21} for a proof of this fact in the countable case, which is the only one we use. This yields the following result.

\begin{lemma}
\label{le:finitely_equivalent}
Let $\bS$, $\bT$ be other two
finitely equivalent
$\Sigma$-structures. Let $\I$ be a finite $\Sigma$-structure. Then $\bS \rightarrow \I$ if and only if $\bT \rightarrow \I$.
\end{lemma}
\begin{proof}
    We show that $\bS \rightarrow \I$ implies that $\bT \rightarrow \I$. The reverse implication follows analogously. For any finite substructure $\bG \subseteq \bS$ it holds that $\bG \rightarrow \bS$. By the definition of finitely equivalent structures, it must hold that $\bG \rightarrow \bT$ as well. Given that we have assumed that $\bT \rightarrow \I$, composing homomorphisms we obtain that $\bG \rightarrow \I$. This way, any finite substructure of $\bS$ maps homomorphically to $\I$, meaning that $\bS \rightarrow \I$ as well.
\end{proof}

The proofs of the remaining results in this section can be found in~\Cref{ap:sources}. The techniques are standard, and essentially follow the idea from \cite{wang1990dominoes} that runs of a Turing machine $M$ can be encoded in tilings of the plane by using consecutive horizontal lines to describe consecutive configurations of the machine $M$. 
\par
Recall the definition of the grid structure $\bm{\Gamma}$ from \Cref{sec:warm-up}. We repeat here the main result about that structure for the sake of completeness.

\gridprop*

Items \ref{item:grid_undec} and \ref{item:grid_noncomp} were shown in \cite{wang1990dominoes} and \cite{hanf1974nonrecursive} respectively. In \ref{item:grid_noncomp} we consider the plain encoding of $\bm \Gamma$, which represents pairs $(m_1, m_2)\in \NN^2$ as comma-separated lists delimited by parentheses, where the integers are written in decimal notation. \par

It is worth remarking that the notion of finite equivalence introduced earlier preserves both the undecidability and the hardness parts in this result. Indeed, if $\bG$ is finitely equivalent to $\bm \Gamma$ then $\Hom(\bG, \cdot)$ and
$\Hom(\bm \Gamma, \cdot)$ are the same family, and the problems $\spcsp(\bG, \bT)$ and $\spcsp(\bm \Gamma, \bT)$
are the same whenever $\bG \rightarrow \bT$ for some finite structure $\bT$. However, the non-computability part of this result is not preserved by finite equivalence. To see this, define $\bG$ as a disjoint union of increasingly large grids. For instance, for each $n\in \NN$, let $\bG_n$ be the substructure of $\bm \Gamma$ induced on $[n]$, and let 
$\bG = \bigsqcup_{n=1}^\infty \bG_n$. It is not difficult to see that $\bG$ is finitely equivalent to $\bm \Gamma$. Even so, it also holds that
there exists a computable homomorphism $F:\bG \rightarrow \bT$ whenever $\bG \rightarrow \bT$ for a finite structure $\bT$. To see this, observe that when querying the value of a global homomorphism on an element $v$, an algorithm just needs to compute a homomorphism from $v$'s connected component, which is finite, to $\bT$, and store it in memory. This yields a global homomorphism that can be computed on the fly. \par

In order to obtain TFNP-hardness instead of TFNP$_1$-hardness we utilize a three-dimensional grid with extra constraints corresponding to ``doubling'' each coordinate. The \emph{super-grid} $\bm{\Gamma}^+$ has signature
$\Sigma_{\Gamma^+}=\{
O, E_1, E_2, E_3, \Ebb_1, \Ebb_2, \Ebb_3
\}$, where the symbol $O$ is unary and all other symbols are binary. The universe $\Gamma^+$ is the set of triples $\NN^3$. The relations of $\bm{\Gamma}^+$ are defined as follows. We have 
$O^{\Gamma^+}=\{(1,1,1)\}$. The relations $E_1^{\Gamma^+}, E_2^{\Gamma^+}, E_3^{\Gamma^+}$
describe unit increments in the first, second, and third coordinate respectively. I.e., $E_1^{\Gamma^+} = 
\{ ((m, n, o),(m+1, n, o)) \, \vert \, (m,n,o)\in \NN^3 \}$, and so on. Similarly, the relations $\Ebb_1^{\Gamma^+}, \Ebb_2^{\Gamma^+}, 
\Ebb_3^{\Gamma^+}$
describe doubling increments in the first, second, and third coordinate respectively. That is, 
$\Ebb_1^{\Gamma^+} = 
\{ ((m, n, o),(2m, n, o)) \, \vert \, (m,n,o)\in \NN^3 \}$, and so on. 

\begin{proposition}
    \label{prop:3-grid}
    The following hold: \begin{enumerate}[label=(\arabic*)]
        \item the family of problems $\spcsp(\bm \Gamma^+, \bT)$ where $\bT\in \Hom(\bm \Gamma^+, \cdot)$ is TFNP-hard,
        \label{item:3-grid_hardness}
        \item $\Hom(\bm \Gamma^+, \cdot)$ is undecidable, and \label{item:3-grid_undec}
        \item there exists $\bT\in \Hom(\bm \Gamma^+, \cdot)$ for which there is no computable homomorphism $F: \bm \Gamma^+ \rightarrow \bT$.
        \label{item:3-grid_noncomp}
    \end{enumerate}
\end{proposition}

Again, in the non-computability result we consider the plain encoding of $\bm \Gamma^+$. Here we point out that the undecidability and non-computability parts of this result just follow from \Cref{prop:grid}. Indeed, given a $\Sigma_\Gamma$-structure $\bT$ it is easy to construct a
$\Sigma_{\Gamma^+}$-structure $\bT^+$ such that
$\bm{\Gamma} \rightarrow \bT$ if and only if $\bm{\Gamma}^+ \rightarrow \bT^+$,
and where a homomorphism $F:\bm \Gamma \rightarrow \bT$ can be computed given oracle access to a homomorphism 
$H:\bm \Gamma^+ \rightarrow \bT^+$. Indeed, $\bT^+$ can be obtained extending $\bT$ by interpreting each symbol $R\in \Sigma_{\Gamma^+} \setminus \Sigma_{\Gamma}$ as the total relation of arity $\arty(R)$ over $T$. \par

Finally, we need one last source of undecidability results, which will be given by a family of growing triangular slices of the two-dimensional grid. Given $m\in \NN$, the structure $\bm{\nabla}_m$ has signature
$\Sigma_\nabla=\{O, W, E_1, E_2\}$, where
$O, W$ are unary symbols and $E_1, E_2$ are binary. The universe $\nabla_m$ consists of all pairs $(n,o) \in \NN^2$
satisfying $n+o \leq m$. The relations $O^{\nabla_m}, E_1^{\nabla_m}, E_2^{
\nabla_m}$ are defined as in $\bm{\Gamma}$. That is,
$O^{\nabla_m}=\{(1,1)\}$, $E_1^{\nabla_m}$ consists of all pairs of the form $((n, o),(n+1,o))$ and  
$E_2^{\nabla_m}$ contains the pairs
$
((m,o),(m,o+1))$. Finally, the relation $W^{\nabla_m}$ contains all pairs $(n,o)$ satisfying $n+o=m$ (i.e., the upper-right boundary of the triangle).

\begin{proposition}
    \label{prop:triangles_undecidability}
    Let $(a_n)_{n\in \NN}$ be a strictly increasing sequence of natural numbers. Then the following families are undecidable (1) $\Hom(\{ \nabla_{a_n} \,\vert\, n\in \NN\}, \cdot)$, (2) $\Hom_{\eg}(\{ \nabla_{a_n} \,\vert\, n\in \NN\}, \cdot)$, and (3) $\Hom_{\io}(\{ \nabla_{a_n} \,\vert\, n\in \NN\}, \cdot)$.
\end{proposition}

\section{Main Definitions}
\label{sec:definitions}

\subsection{Interpretations over Minions}
\label{sec:interpretations}

We begin by introducing formally the notion of interpretation that we outlined in the proof sketch.  Let $\Mscr$ be a minion. Given a number $n\in \NN$, a \emph{$n$-ary predicate} over $\Mscr$ 
is a subset $P\subseteq \Mscr(n)$. We write $n=\arty(P)$. We write $2^\Mscr$ for the set of predicates over $\Mscr$ of arbitrary arity.  Given a  relational signature $\Sigma$, a \emph{$\Sigma$-interpretation $\Ical$ over $\Mscr$} consists of (1) a predicate $U^\Ical\in 2^\Mscr$, (2) a predicate $R^\Ical\in 2^\Mscr$ for each symbol $R\in \Sigma$, and (3) a map $\Pi^{\Ical}_{R,i}:[\arty(R^\Ical)]\mapsto [\arty(U^\Ical)]$ for each symbol $R\in \Sigma$ and each index $i\in [\arty(R)]$. Interpretations over minions induce two kinds of structures, global and local, defined as follows.
\begin{itemize}
    \item The \emph{global structure induced by $\Ical$}, 
denoted $\bGS=\bGS_\Ical$,
is a $\Sigma$-structure with universe $\GS=U^\Ical$, where for each symbol $R\in \Sigma$, a tuple $(f_1,\dots, f_{\arty(R)})\in 
(U^\Ical)^{\arty(R)}$ belongs to $R^{\GS}$ if there is an element
$g\in R^\Ical$ satisfying that
$f_i = g^{\pi_i}$ for each $i\in [\arty(R)]$, where $\pi_i= \Pi^\Ical_{R,i}$.  \par
\item Given an element $f\in \Mscr$, we define the set $U^{\Ical,f}\subseteq [\arty(U^\Ical)]^{[\arty(f)]}$ 
as the subset of maps $\pi$ for which $f^\pi\in U^\Ical$. 
The \emph{local structure 
induced by $\Ical$ on $f$'s minors},
denoted $\bLS= \bLS_{\Ical, f}$,
is a $\Sigma$-structure whose universe is $U^{\Ical,f}$, 
where for each $R\in \Sigma$
the relation $R^{\LS}$ consists of all the tuples of the form $(\sigma\circ \pi_1, \dots, \sigma \circ \pi_{\arty(R)})\in (U^{\Ical,f})^{\arty(R)}$, where $f^\sigma \in R^\Ical$, and $\pi_i= \Pi^\Ical_{R,i}$ for each $i\in [\arty(R)]$.
\end{itemize}

It is important to remark that the definition of local structure appearing here is not exactly the same as the one given in the proof sketch. Indeed, the local structure $\bGS_{\Ical,f}$ is \emph{not} the induced substructure of $\bGS_\Ical$ on the minors of $f$. The reason is that it may be that $f^{\pi_1}= f^{\pi_2} \in U^{\Ical}$ for some suitable maps, whereas in the local structure $\bGS_{\Ical, f}$ the maps $\pi_1, \pi_2$ represent different elements. \par

Although this is not relevant in the proofs of our algorithmic results, when we analyse cyclic and WNU polymorphisms we will deal with minions that correspond to disjoint unions of simpler minions. For instance, the minion representing the existence of a $k$-ary WNU polymorphism for each arity $k\geq 3$ is the disjoint union of the minion generated by an abstract $3$-ary WNU element, with the one generated by a $4$-ary WNU element, and so on. Hence, we also specialise some of our auxiliary results to disjoint unions of minions. \par

Given a subminion $\Nscr \subseteq \Mscr$, the \emph{restricted interpretation} $\Jcal= \Ical\vert_\Nscr$ is defined by 
$U^\Jcal= U^\Ical \cap \Nscr$,  $R^\Jcal= R^\Ical\cap \Nscr$ for each $R\in \Sigma$, and $\Pi^\Jcal_{R,i}= \Pi^\Ical_{R,i}$ for each $R\in \Sigma$, $i\in [\arty(R)]$. 

\begin{observation}
\label{obs:restricted_interpretation}
    Let $J$ be a set, and $\Ical$ be a $\Sigma$ interpretation over a disjoint union of minions $\Mscr= \bigsqcup_{j\in J} \Mscr_j$, and let 
    $\Ical_j= \Ical\vert_{\Mscr_j}$ for each $j\in J$. Then the following hold: (1) $\bGS_\Ical =\bigsqcup_{j\in J} 
    \bGS_{\Ical_j}$, and (2)
    $\bLS_{\Ical,f}= \bLS_{\Ical_j, f}$
    for all $j\in J$, $f\in \Mscr_j$. 
\end{observation}

\subsection{Manifold Minions}
\label{sec:manifold_minion}

Now we are in position to introduce manifold minions, which is arguably the main construction that enables our proofs. \par

Let $\Mscr$ be a minion, 
$\Ical$ be a $\Sigma$ interpretation over it, and $\C$ be a $\Sigma$-structure. The \emph{manifold minion}
$\Nscr = \C^{\Ical}$ is defined as follows. The elements in $\Nscr(n)$ are pairs $(f, \chi)$, where $f\in \Mscr(n)$, and $\chi: U^{\Ical,f} \rightarrow C$ is a homomorphism from 
$\bLS_{
\Ical,f}$ to $\C$, which we call a \emph{local homomorphism}. Given a map $\pi$, the minoring operation is given by $(f, \chi)^\pi = (f^\pi, \chi^\pi)$, where $\chi^\pi$ is defined by $\gamma \mapsto \chi(\gamma\circ \pi)$. Observe that if $g=f^\pi$, then $U^{\Ical,g} = \{  \sigma \, \vert \, 
\sigma \circ \pi \in U^{\Ical,f}  \}$, so the homomorphism $\chi^\pi$ is well-defined. The \emph{canonical projection} $\pj: \C^\Ical \partialmap \C$ is the partial map defined on pairs $(f, \chi)$ such that $f\in U^\Ical$, which maps each such pair to the element $\chi(\id)$, where $\id=\id_{[\arty(U^\Ical)]}$ refers to the identity map over $[\arty(U^\Ical)]$. Observe that $\chi(\id)$ is well-defined. Indeed, the fact that $f\in U^\Ical$ implies that $\id\in \bLS_{\Ical,f}$.

\subsection{Descriptions}
\label{sec:descriptions}
In this section we introduce a way to obtain quotient minions that preserve sufficient information about interpretations. Roughly, the main result here, \Cref{lem:hom_to_power_minion}, states that given a $\Sigma$-interpretation $\Ical$ over a minion $\Mscr$, a $\Sigma$-structure $\C$, and a nice quotient $\Mscr/\sim$ of $\Mscr$, then $\Mscr \rightarrow \C^{\Ical/\sim}$ if and only if the global structure $\bGS_\Ical$ maps homomorphically into $\C$. \par

A \emph{description} of $\Mscr$ is a set of predicates $\Dcal\subseteq \Mscr$. A description $\Dcal$ 
induces an equivalence relation $\sim_\Dcal$ on $\Mscr$ as follows. Let $f_1, f_2\in \Mcal(n)$ for some $n\in \NN$. Then $f_1 \sim_\Dcal f_2$ means that 
$f_1^\pi \in P$ if and only if $f_2^\pi \in P$ for every $P\in \Dcal$ and every map $\pi: [n] \rightarrow [\arty(P)]$.
We shorten $\sim_\Dcal$ to $\Dcal$ when writing quotients to keep the notation light (i.e., we write $\, \cdot \, / \Dcal$ instead of $\, \cdot \, / \sim_\Dcal$).
Given an element $f\in \Mscr$, we write $\qclass{f}$ to denote its equivalence class in $\Mscr/\Dcal$, and given a set $S\subseteq \Mscr$ we write $\qclass{S}$ for the set of equivalence classes $\qclass{f}$ of elements $f\in S$.  The quotient $\Mscr/\Dcal$ inherits a natural minion structure from $\Mscr$: for every $n\in \NN$, we define $(\Mscr/\Dcal)(n)=
\qclass{\Mscr(n)}$, and for every $\qclass{f}\in \Mscr/\Dcal$ and every suitable map $\pi$, we define the minoring operation as $(\qclass{f})^\pi = \qclass{f^\pi}$. Observe that this operation is well defined and does not depend on the chosen representative $f$. Moreover, if the description $\Dcal$ is finite, then $\Mscr/\Dcal$ is locally finite.  
\par
Given a description $\Dcal\subseteq 2^\Mscr$, and a $\Sigma$-interpretation $\Ical$, the \emph{quotient interpretation} $\Jcal=\Ical/\Dcal$ is the $\Sigma$ interpretation over $\Mscr/\Dcal$ defined by $U^{\Jcal}= \qclass{U^\Ical}$, 
$R^\Jcal= \qclass{R^\Ical}$ for each $R\in \Sigma$, and $\Pi^\Jcal_{R,i}= \Pi^\Ical_{R,i}$ for each $R\in \Sigma$, $i\in [\arty(R)]$. 
\par

We remark that by considering minion quotients corresponding to descriptions we are not loosing any generality. Indeed, if there is a surjective homomorphism $F: \Mscr \rightarrow \Nscr$ then $\Nscr$ is isomorphic to $\Mscr/\Dcal$, where $\Dcal$ is the description consisting of all the preimages $F^{-1}(g)$. So descriptions should be understood as a way to efficiently encode equivalence relations over minions. Indeed, given a description $\Dcal$ of $\Mscr$ consisting only of $k$-ary predicates, the equivalence relation $\sim_{\Dcal}$ over $\Mscr$ could have at most $k^n 2^{|\Dcal|} $
classes of $n$-ary elements. \par

\paragraph{Internal Predicates}

The next step now is to formally introduce what it means for a description $\Dcal$ of a minion $\Mscr$ to preserve the meaningful information about an interpretation $\Ical$. More precisely, what we want is that the interpretation $\Ical$ is preserved, in some way, by all homomorphisms from $\Mscr$ to the quotient $\Mscr/\Dcal$.
The key definition here is that of an \emph{internal predicate with respect to a description}. The description $\Dcal$ will define a good quotient for $\Ical$ if all the predicates that form $\Ical$ are internal with respect to $\Dcal$. \par

Let $\Mscr$ be a minion and let $\Dcal\subseteq 2^\Mscr$ be a description. A predicate $Q\in 2^\Mscr$ is called \emph{$\Dcal$-stable} if 
$f\sim_\Dcal g$ together with $f\in Q$ imply that $g\in Q$ for any $f,g\in \Mscr(\arty(Q))$. Given  $h\geq \arty(Q)$, we say that $Q$
is \emph{internal at arity $h$} with respect to $\Dcal$ if (1) $Q$ is $\Dcal$-stable, and
(2) all partial homomorphisms 
$F: \Mscr \overpartialmap{h} 
\Mscr/\Dcal$ satisfy $F(Q)\subseteq \qclass{Q}$. Similarly, a $\Sigma$-interpretation $\Ical$ over $\Mscr$ is $\Dcal$-stable if all the predicates $U^\Ical$,  
$R^\Ical$ for $R\in \Sigma$ are $\Dcal$-stable; and $\Ical$ is internal w.r.t. $\Dcal$ at arty $h$ if all the predicates $U^\Ical$, $R^\Ical$ for $R\in \Sigma$ are internal at arity $h$ w.r.t. $\Dcal$. In particular this requires that $h$ is at least as large as $\arty(U^\Ical)$ and $\arty(R^\Ical)$ for each $R\in \Sigma$.
\par

An \emph{internal reference to } $Q$ w.r.t. $\Dcal$ is a pp-formula $\Phi(x)\in \Lcal_{\MC}$ with one free variable $x$ of arity $\arty(Q)$ satisfying 
\[
\Mscr/\Dcal \models \Phi(\qclass{f}) \implies
\qclass{f} \in \qclass{Q} 
\]
for all $\qclass{f} \in \Mscr(\arty(Q))/\Dcal$. Our main tool for showing that a predicate is internal is the following result.

\begin{lemma}[Main criterion for internal predicates]
\label{le:internal_criterion}
Let $\Mscr$ be a minion $\Dcal\subseteq 2^\Mscr$ a description, $Q\in 2^\Mscr$ a predicate that is $\Dcal$-stable, and $h\in \NN$ a number. Suppose that for each $f \in Q$ there is an internal reference $\phi_f(x)$ to $Q$ w.r.t. $\Dcal$ whose maximum arity is at most $h$, and such that $\Mscr \models \phi_f(f)$. Then $Q$ is internal at arity $h$ w.r.t. $\Dcal$ 
\end{lemma}
\begin{proof}
    Let $F: \Mscr \overpartialmap{h} \Mscr/\Dcal$ be a partial homomorphism and $f\in Q$. Observe that 
    $\Mscr \models \phi_f(f)$ is witnessed by an assignment over elements of arity at most $h$, so this implies
    $\Mscr/\Dcal \models \phi_f(F(f))$. The fact that $\phi_f$ is an internal reference to $Q$ w.r.t. $\Dcal$ yields $F(f)\in \qclass{Q}$. This proves the result. 
\end{proof}

A situation in which this last criterion is especially easy to apply is that in which $\phi_f$ can be chosen to be the same for all $f\in Q$. This motivates the following notion. 
Let $\Mscr$ be a minion, $\Dcal\subseteq 2^\Mscr$ a description and $Q\in 2^\Mscr$ a predicate.
An \emph{internal definition of $Q$ with respect to $\Dcal$} is an internal reference $\Phi(x)$ to $Q$ w.r.t. $\Dcal$ that additionally satisfies 
\[
f\in Q \implies  \Mscr \models \Phi(f)
\]
for all $f\in \Mscr(\arty(Q))$.
If such definition $\Phi(x)$ exists, its arity is bounded by a number $h\in \NN$, and additionally $Q$ is $\Dcal$-stable, then $Q$ is said to be \emph{internally definable at arity $h$ w.r.t. $\Dcal$}. Observe that by~\Cref{le:internal_criterion}, this  implies that $Q$ is internal at arity $h$ w.r.t. $\Dcal$.

\paragraph{Main Results About Quotient Interpretations}

The next two lemmas show that, given a minion $\Mscr$, if $\Ical$ is an internal interpretation with respect to a description $\Dcal \subseteq 2^\Mscr$, then homomorphisms $\Mscr \rightarrow \Mscr/\Dcal$ preserve the interpretation, and homomorphisms to power minions $\Mscr \rightarrow \C^{\Ical/\Dcal}$ correspond to homomorphisms from the global induced structure $\bGS_\Ical$ to $\C$.

\begin{lemma}[Homomorphisms to good quotients preserve local structures]
\label{le:local_structure_quotient}
    Let $\Mscr$ be a minion, $\Ical$ a $\Sigma$-interpretation over $\Mscr$,
    and $\Dcal\subseteq 2^\Mscr$ a description. Define $\Jcal= \Ical/\Dcal$. The following hold.
    \begin{enumerate}[label=({\arabic*})]
        \item Suppose that $\Ical$ is $\Dcal$-stable. Then 
        $\bLS_{\Ical, f}= \bLS_{\Jcal, \qclass{f}}$ for all $f\in \Mscr$.
        \item Suppose that $\Ical$ is internal at arity $h\in \NN$ w.r.t. $\Dcal$, and that
        $F:\Mscr \overpartialmap{h} \Mscr/\Dcal$ is a partial homomorphism. Then for any
        $f\in\Mscr$ whose arity is at most $h$, the local structure $\bLS_{\Ical, f}$ is contained in $\bLS_{\Jcal, F(f)}$.
    \end{enumerate}
\end{lemma}
\begin{proof}
\emph{(1)} This is a direct consequence of the stability condition. For any $f\in \Mscr$ and any suitable map $\pi$ it holds that $f^\pi \in U^\Ical$ if and only if $\qclass{f}^\pi \in U^\Jcal$, and $f^\pi \in R^\Ical$ if and only if $\qclass{f}^\pi \in R^\Jcal$ for each $R\in \Sigma$. This proves the statement. \par
\emph{(2)} Let $\A=\bLS_{ \Ical, f}$, 
    $\B=\bLS_{\Jcal, F(f)}$, and $n_U=\arty(U^\Ical)$.
    First, we show that if a map $\sigma\in [n_U]^{[n]}$ belongs to $A$,
    then it also belongs to $B$. 
    Indeed, the first condition is equivalent to $f^\sigma\in U^\Ical$. Because $\Ical$ is internal, it follows that $F(f)^{\sigma}\in 
    \qclass{U^\Ical} = U^\Jcal$, which is equivalent to the second condition. Now, let $R\in \Sigma$, and 
    $n_R=\arty(R^\Ical)$. We show that if a tuple $\bm{\sigma}$ belongs to  $R^A$, then $\bm{\sigma}$ belongs to $R^B$ as well. The first condition means that there is a map $\gamma\in [n_R]^{[n]}$ such that $\sigma_i=\Pi_{R,i}^\Ical$ for each $i\in [\arty(R)]$, and $f^\gamma\in R^\Ical$. Because $\Ical$ is internal, it must hold that $F(f)^\gamma \in R^\Jcal$, which shows that $\bm{\sigma}$ also belongs to $R^B$.
\end{proof}

We gently remind the reader that, given a manifold minion $\C^{\Jcal}$ defined by a $\Sigma$-structure $\C$ and a $\Sigma$-interpretation $\Jcal$, we write $\pj$ for the canonical projection from $\C^{\Jcal}$ to $\C$, defined in \Cref{sec:manifold_minion}.

\begin{lemma}[Power minions on good quotients capture homomorphisms from the global structure]
\label{lem:hom_to_power_minion}
    Let $\Mscr$ be a minion, $\Ical$ a $\Sigma$-interpretation over $\Mscr$, $\Dcal\subseteq 2^\Mscr$
    a description, and $\C$ a $\Sigma$-structure. Define $\Jcal= \Ical/\Dcal$.  
    Then the following hold.
    \begin{enumerate}[label=({\arabic*})]
        \item Suppose that $\Ical$ is $\Dcal$-stable and $\bGS_\Ical \rightarrow \C$. Then 
        $\Mscr\rightarrow \C^\Jcal$. 
        \item Suppose that $\Ical$ is internal w.r.t. $\Dcal$ at arity $h\in \NN$, and 
        $F:\Mscr \overpartialmap{h} \C^\Jcal$ is a partial homomorphism. 
        Then $\pj \circ  F\vert_{U^\Ical}$ is a homomorphism from $\bGS$ to $\C$.
    \end{enumerate}    

\end{lemma}
\begin{proof}
\emph{(1)} Suppose there is a homomorphism $H: \bGS_\Ical \rightarrow \C$. Given an element $f\in \Mscr$, the local homomorphism $\chi_f: \bLS_{\Ical,f} \rightarrow \C$ is given by $\pi \mapsto H(f^\pi)$ for each $\pi\in U^{\Ical,f}$. Observe that if $g=f^\pi$, and $\sigma\in U^{\Ical,g}$, then $\chi_g(\sigma)= \chi_f(\sigma\circ\pi)$, so the local homomorphisms we have defined are compatible with minoring. Now, by item (1) of \Cref{le:local_structure_quotient}, 
$\bLS_{\Ical,f}= \bLS_{\Jcal, \qclass{f}}$ for each $f\in \Mscr$, so the map 
$f\mapsto (\qclass{f}, \chi_f)$ is a minion homomorphism from $\Mscr$ to $\C^\Jcal$. \par
\emph{(2)}  Suppose there is a partial homomorphism $F: \Mscr \overpartialmap{h} \C^\Jcal$,
    given by $f\mapsto (p_f, \chi_f)$, and $\rho$ be the canonical projection from $\C^\Jcal$ to $\C$. We show that
    $H=\pj \circ F\vert_{U^\Ical}$ is a homomorphism from 
    $\bGS_\Ical$ to $\C$. 
    The map $H$ sends every element $f\in U^\Ical$ to $\chi_f(\id)$, where $\id$ denotes the identity map over $[\arty(U^\Ical)]$. First, let us see that $H$ is a well-defined map. Observe that the map $F^\prime: \Mscr \overpartialmap{h} \Mscr/\Dcal$ given by $f\mapsto p_f$ is a partial minion homomorphism. By item (2) of \Cref{le:local_structure_quotient}, $\bLS_{\Ical,f}\subseteq \bLS_{\Jcal, p_f}$ for every $f\in \Mscr(n)$, $n\leq h$.
    In particular, if $f\in U^\Ical$, then $\id$ belongs to $\bLS_{\Jcal, p_f}$, so $H(f)= \chi_f(\id)$ is well-defined. Now let $R\in \Sigma$ and $(f_1,\dots, f_{\arty(R)})\in R^{\GS_\Ical}$. We prove that
    $(H(f_1),\dots, H(f_{\arty(R)}))\in R^C$. 
    Let $\pi_i=\Pi^\Ical_{R,i}$ for each $i\in [\arty(R)]$.
    By the definition of $\bGS_\Ical$, there must be an element $f_R\in R^\Ical$ such that $
    f_i= f_R^{\pi_i}$ for each $i\in [\arty(R)]$. 
    In particular, this means that $(\pi_1, \dots, \pi_{\arty(R)}) \in R^A$, where $\A=\bLS_{\Ical, f_R}$. Additionally,  given $i\in [\arty(R)]$,
    the following chain of identities holds:
    \[
    \chi_{f_R}(\pi_i)=
    \chi_{f_R}^{\pi_i}(\id) = \chi_{f_i}(\id)= H(f_i).
    \]  
    Hence, as $\chi_{f_R}$ is a homomorphism from
    $\bLS_{\Ical, f_R}$ to $\C$, the tuple
    $(H(f_1), \dots, H(f_{\arty(R)})=$ $
    (\chi_{f_R}(\pi_1),$ $ \dots, \chi_{f_R}(\pi_{\arty(R)}))$ belongs to $R^C$, as we wanted to prove. 
 \end{proof}

Finally, we give another version of this last lemma that handles minions $\Mscr$ that can be decomposed as disjoint unions of subminions. 

\begin{lemma}
\label{lem:hom_to_power_minion_disjoint_union}
     Let $K$ be a set, $\Mscr$ a disjoint union $\bigsqcup_{k\in K} \Mscr_k$ of subminions, $\Ical$ a $\Sigma$-interpretation over $\Mscr$, $\Dcal\subseteq 2^\Mscr$ a description, 
     and $\C$ a $\Sigma$-structure. Define
     $\Ical_k=  \Ical\vert_{\Mscr_k}$ for each $k\in K$, and 
     $\Jcal= \Ical/\Dcal$. The following hold.
     \begin{enumerate}[label=({\arabic*})]
         \item Suppose that $\Ical$ is $\Dcal$-stable, and $\bGS_{\Ical_k}\rightarrow \C$ for some $k\in K$. Then $\Mscr_k\rightarrow \C^\Jcal$.
         \item Suppose that $\Ical$ is internal at arity $h$ w.r.t. $\Dcal$, and $F:\Mscr_k \overpartialmap{h} \C^\Jcal$ is a partial homomorphism for some $k\in K$. Then
    $\pj\circ F \vert_{U^{\Ical_k}}$ is a homomorphism from $\bGS_{\Ical_k}$ to $\C$.
     \end{enumerate}
\end{lemma}
\begin{proof}
    The proof follows from the same arguments as the proof of~\Cref{lem:hom_to_power_minion}, using~\Cref{obs:restricted_interpretation}. The main insight is that, given a minion $\Nscr$, partial homomorphisms $\Mscr \overpartialmap{h} \Nscr$ correspond to families $(F_k)_{k\in K}$ of partial homomorphisms $F_k: \Mscr_k \overpartialmap{h} \Nscr$ in a one-to-one fashion. If we let $\Nscr= \Mscr/\Dcal$, the fact that $\Ical$ is internal at arity $h$ w.r.t. $\Dcal$, implies that for all $F_k: \Mscr_k \overpartialmap{h} \Nscr$ it must hold that $F_k(U^{\Ical_k})\subseteq \qclass{U^{\Ical}}$ and $F_k(R^{\Ical_k})\subseteq \qclass{R^{\Ical}}$ for all $R\in \Sigma$.
\end{proof}

\subsection{Patterns}
\label{sec:patterns}
The ideas introduced in the previous sections are (almost) enough to prove our main undecidability and non-computability results.
Indeed, consider a $\Sigma$-interpretation $\Ical$ over a minion $\Mscr$ which is internal at arity $k$ w.r.t. some description $\Dcal\subseteq 2^\Mscr$, and a $\Sigma$-structure $\bT$. Then,  \Cref{lem:hom_to_power_minion} shows that
there is a homomorphism from the global structure $\bGS_\Ical$ to $\bT$ if and only if $\Mscr \rightarrow \bT^{\Ical/\Dcal}$. Moreover, it is not difficult to show that it is possible to compute a homomorphism $F: \bGS_\Ical \rightarrow \bT$ when given oracle access to a homomorphism $H: \Mscr \rightarrow \bT^{\Ical/\Dcal}$. 
What is missing in this picture is a way to reduce, assuming that $ \bGS_\Ical \rightarrow \bT$, the problem
$\spcsp( \bGS_\Ical,\bT)$ to $\sPMC_h(\Mscr,\bT^{\Ical/\Dcal})$ for some $h\in \NN$. This reduction will be given by what we call a \emph{pattern}, which is a polynomial-time algorithm that constructs minor conditions related to input $\Sigma$-structures, in a way similar to the formulas $\Psi_{\bG}$ defined in \eqref{eq:pattern_warm-up} during our proof sketch. The formal definition is given below. \par

 Let $\Mscr$ a minion, $h\in \NN$ an integer,
 $\Dcal\subseteq 2^\Mscr$ a description, and $\Ical$ a $\Dcal$-stable $\Sigma$-interpretation over $\Mscr$.
    A \emph{pattern of $h$-ary internal references to $\Ical$ w.r.t. $\Dcal$} is a
     function $\Psi$ computable in polynomial time that sends every finite structure $\I$ satisfying
    $\I \rightarrow \bGS_{\Ical}$ 
    to a minor condition $\Psi_{\I}$ of arity at most $h$, 
    such that $\Mscr\models \Psi_\I$, and that can be written as    
    \begin{align*}
        \bigexists_{v\in I} x_{v}
\bigexists_{R\in \Sigma, r \in R^I} x_r 
\left(
        \bigwedge_{v\in I} \psi_{v}(x_v) 
\right)
\bigwedge  \left( 
\bigwedge_{R\in \Sigma, r \in R^I}        \psi_r(x_r) \bigwedge_{i\in [\arty(R)]}
        x_r^{\Pi^\Ical_{R,i}} = x_{r(i)}
        \right),
    \end{align*}
     where each $v\in I$ the formula $\psi_v(x)$ is an internal reference to $U^\Ical$ w.r.t. $\Dcal$, and for each $R\in \Sigma$ and each $r\in R^I$, the formula $\psi_r(x)$ is an internal reference to $R^\Ical$ w.r.t. $\Dcal$. To abbreviate, we will refer to the tuple $(\Ical, \Dcal, \Psi, h)$ as a $\Sigma$-\emph{pattern} over $\Mscr$.     
     \par

\begin{lemma}
    \label{le:patterns}
    Let $(\Ical, \Dcal, \Psi,h)$ be a pattern over a minion $\Mscr$, and 
    $\bT$ a finite structure satisfying $\bGS_\Ical\rightarrow \bT$. 
    Then there is a many-one reduction from 
    $\spcsp(\bGS_\Ical, \bT)$ to $\sPMC_h(\Mscr, \bT^{\Ical/\Dcal})$.
\end{lemma}
\begin{proof}
    Let $\Jcal= \Ical/\Dcal$, and $\Nscr= \bT^\Jcal$. Observe that by \Cref{lem:hom_to_power_minion}, the fact that $\Ical$ is $\Dcal$-stable implies that
    $\Mscr\rightarrow \bT^{\Jcal}$, so 
    $\sPMC_h(\Mscr, \Nscr)$ is well-defined. \par
    We give a many-one reduction from $\spcsp(\bG, \bT)$ to 
    $\sPMC_h(\Mscr, \Nscr)$. This reduction consists of a pair of 
    polynomial-time computable functions
    $(\alpha, \beta)$, where (1) $\alpha$ maps 
structures $\I$ 
satisfying $\I\rightarrow \bG$ to minor conditions $\psi$ of arity at most $h$ that satisfy $\Mscr \models \psi$, and (2)
$\beta$ maps pairs $(\I, F)$ to homomorphisms $H:\I \rightarrow \bT$, where $\I$ is a structure satisfying $\I \rightarrow \bG$ and $F$
is an assignment satisfying $\alpha(\I)$ over $\Nscr$. We simply define the function $\alpha$ as the map $\I\mapsto \Psi_\I$. 
Now consider an assignment $F$ that satisfies $\Psi_\I$ over $\Nscr$, and let $\pj$ be the canonical projection form $\Nscr$ to $\bT$. Let $H_F: \I \rightarrow \bT$ be the map that sends each $v\in I$ to  
$\pj(F(x_v))$, which is clearly computable in polynomial time. We claim that $H_F$ is a homomorphism. This way, we can define $\beta$ as $(\I, F) \mapsto H_F$.
Let us show that $H_F$ is indeed a homomorphism. Recall the structure of the minor condition $\Psi_\I$. We can write 
\begin{align*}
\Psi_\I \equiv &
\bigexists_{v\in I} x_{v}
\bigexists_{R\in \Sigma, r \in R^I} x_{r}
\left(
        \bigwedge_{v\in I} \psi_{v}(x_v) 
\right)
\bigwedge \\ &
         \left( 
\bigwedge_{R\in \Sigma, r\in R^I}        \psi_r(x_r) \bigwedge_{i\in [\arty(R)]}
        x_r^{\Pi^\Ical_{R,i}} = x_{r(i)}
        \right),
\end{align*}
where 
$\psi_v(x)$ is an internal reference (w.r.t. $\Dcal$)
to $U^\Ical$ for each $v\in I$, and $\psi_r(x)$ 
is an internal reference to $R^\Ical$ for each
$R\in \Sigma$, $r\in R^I$. \par
We suppose that $F$ maps $v\mapsto (\qclass{p_v}, \chi_v)$ for each $v\in I$, and
$r \mapsto (\qclass{p_r}, \chi_r)$ for each $R\in \Sigma$, and tuple $r \in R^I$.
First we argue that $H_F$ is a well-defined map. Let $v\in I$. Then 
$\Mscr/\Dcal\models \psi_v(\qclass{p_v})$, which means that $\qclass{p_v}\in U^\Jcal$, by the definition of internal reference. Hence $H_F(v)=\pj(F(x_v))= \chi_v(\id)$ is well defined. Now we show that $H_F$ is a homomorphism. Let $R\in \Sigma$ and $r=(v_1,\dots, v_{\arty(R)})\in R^I$.
We have that $\Mscr/\Dcal \models \psi_r(\qclass{p_r})$, so $\qclass{p_r}\in R^\Jcal$. It also holds that $F(x_r)^{\pi_i}=F(x_{v_i})$ for each $i\in [\arty(R)]$, where $\pi_i= \Pi_{R,i}^\Jcal$. This means that $\chi_r(\pi_i)= \chi_{v_i}(\id)= 
H_F(v_i)$ for each $i\in [\arty(R)]$. Observe that $(\pi_1, \dots, \pi_{\arty(R)})\in R^{\A}$, where $\A= \LS_{\Jcal, \qclass{p_r}}$. Indeed, we have established that $\qclass{p_r}\in R^\Jcal$, and it holds that $\pi_i= \pi_i \circ \id$ for each $i\in [\arty(R)]$. Using that $\chi_r$ is a homomorphism from $\bLS_{\Jcal,\qclass{p_r}}$ to $\bT$, we obtain that 
 $(\chi_r(\pi_1), \dots, \chi_r(\pi_{\arty(R)}))
 =(H_F(v_1),\dots, H_F(v_{\arty(R)})
 \in R^{T}$. This proves that $H_F$ is a homomorphism. 
\end{proof}

\section{From Minions to Templates}
\label{sec:minion_closures}

In this section, given a locally finite minion $\Mscr$ and a number $h\in \NN$, we obtain a finite template $(\A, \B)$ satisfying that $\pol(\A, \B)$ has a partial isomorphism to $\Mscr$ up to arity $h$ and that
$\Nscr\overpartialmap{h}\Mscr$ if and only if $\Nscr \rightarrow \pol(\A, \B)$ for any other minion $\Nscr$. Furthermore, we want to do so while keeping $\A$ small, controlling both the size of its universe $A$ and the size of its relations. \par
In order to state the main result of this section we need to define one more notion. Let $\Mscr$ be a minion. 
The \emph{rank} of $\Mscr$ is the smallest number $r\in \NN$ satisfying that, for any $n\in \NN$, whenever two elements $f_1, f_2\in \Mscr(n)$ have the same $r$-ary minors (i.e., $f_1^\pi= f_2^\pi$ for all $\pi \in [r]^{[n]}$), then $f_1= f_2$. We say that $\Mscr$ has infinite rank if no such $r$ exists. Our goal is to prove the following.

\begin{restatable}{theorem}{miniontotemplate}
\label{th:minion_to_template}
        Let $\Mscr$ be a minion whose rank is at most $r$, and let $h\geq r$. Then
        there is a template $(\A, \B)$ where $|A|=r$ and $|R^A|\leq h$ for all relations in $\A$ such that
        \begin{enumerate}[label=({\arabic*})]
            \item There is a partial minion isomorphism $\Mscr \overpartialmap{h} \pol(\A, \B)$ defined up to arity $h$.
            \item Given a minion $\Mscr^\prime$, it holds that $\Mscr^\prime \rightarrow \pol(\A, \B)$ if and only if 
           $\Mscr^\prime \overpartialmap{h}\Mscr$.
        \end{enumerate}
        Moreover, if $\Mscr$ is locally finite, then $\B$ is a finite structure. 
\end{restatable}

Through similar ideas to the ones used to prove this result, we are also able to characterize abstract minions which are isomorphic to the polymorphism minion of some finite template. This characterization is not essential to show the main results of this paper, and we defer it to \Cref{app:pol_characterization}. \par
In our reductions we apply \Cref{th:minion_to_template} to manifold minions built using quotient interpretations. We will need the following bound on the rank of such minions.

\begin{lemma}
    \label{le:rank_manifold_minion}
    Let $\bT$ be a $\Sigma$-structure, $\Mscr$ a minion, $\Dcal\subseteq 2^\Mscr$, 
    $\Ical$ a $\Sigma$-interpretation over $\Mscr$, and $r\in \NN$. Suppose that
    \begin{enumerate}[label=({\arabic*})]
        \item $r\geq \arty(Q)$ for all $Q\in \Dcal$, and
        \item $r\geq \arty(U^\Ical)$. 
    \end{enumerate}
    Then the rank of $\bT^{\Ical/\Dcal}$ is at most $r$.
\end{lemma}
\begin{proof}
  Let $f,g\in \Mscr(n)$ for some $n\in \NN$, and suppose that $f^\pi\sim_{\Dcal} g^\pi$ for all $\pi \in {[r]}^{[n]}$. We show this implies $f\sim_\Dcal g$. By \Cref{le:rank} this implies that the rank of $\Mscr/\Dcal$ is at most $r$. Indeed, suppose that $f \not\sim_\Dcal g$. Then, without loss of generality we may assume there is some $Q\in \Dcal$ and some $\pi \in [\arty(Q)]^{[n]}$ such that $f^\pi \in Q$ but $g^\pi \not\in Q$. Let $\alpha:[\arty(Q)] \rightarrow [r]$ 
  and $\beta:[r]\rightarrow [\arty(Q)]$ be such that $\beta \circ \alpha=\id_{[\arty(Q)]}$. Such maps exist because $\arty(Q)\leq r$.
  By assumption,
  $f^{\alpha \circ \pi} = g^{\alpha \circ \pi}$. However
  $f^{\pi}= (f^{\alpha \circ \pi})^\beta$, and 
  $g^{\pi}= (g^{\alpha \circ \pi})^\beta$, a contradiction. \par
  Now define $\Jcal=\Ical/\Dcal$, and $\Nscr=\bT^{\Ical/\Dcal}$. Let $(\qclass{f}, \chi_f), (\qclass{g}, \chi_g)\in \Nscr(n)$ for some $n\in \NN$. Suppose that 
  \[
  (\qclass{f}, \chi_f)^{\pi} = (\qclass{g}, \chi_g)^\pi \quad \text{ for all $\pi\in [r]^{[n]}$}.
  \]
  Then by the previous arguments $\qclass{f}=\qclass{g}$, so the local structure
  $\bLS_{\Jcal, \qclass{f}}$ and $\bLS_{\Jcal, \qclass{g}}$ are the same. Now, suppose, for the sake of a contradiction, that there is a map $\sigma\in U^{\Jcal, \qclass{f}}$ such that
  $\chi_f(\sigma)\neq \chi_g(\sigma)$. Let $\alpha: [\arty(U^\Ical)] \rightarrow
  [r]$ and $\beta: [r] \rightarrow [\arty(U^\Ical)]$ be such that $\beta\circ \alpha =\id_{[\arty(U^\Ical)]}$ (observe that $\arty(U^\Ical)= \arty(U^\Jcal)$). Then it must hold that
  \[
  \chi_f^{\alpha \circ \sigma}(\beta)= \chi_f(\sigma) \neq \chi_g(\sigma) = \chi_g^{\alpha \circ \sigma}(\beta). 
  \]
  However, $\alpha\circ \sigma\in [r]^{[n]}$, yielding a contradiction. This completes the proof. 
\end{proof}

Finite rank minions are precisely those which are isomorphic to a function minion. Indeed, any function minion, as defined in \cite{BG21:sicomp}, on a domain $D$ must have rank at most $|D|$. In the other direction, if $\Mscr$ has rank $r$, and $p\in \Mscr$ is a $n$-ary element, then it can be seen as a function from $[r]^{n}$ to $\Mscr(r)$ by letting $p(\pi)=p^\pi$ for each $\pi\in [r]^{[n]}$ (recall that we identify tuples and maps). The following is a useful auxiliary fact.
\par
\begin{lemma}
    \label{le:rank}
    Let $\Mscr$ be a minion, $n\in \NN$, and $f, g\in \Mscr(n)$. Let $r \leq h$ be natural numbers. Suppose that $f^\pi= g^\pi$ for every $\pi \in [h]^{[n]}$. Then 
    $f^\pi= g^\pi$ for every $\pi\in [r]^{[n]}$.
\end{lemma}
\begin{proof}
    Let $\alpha: [r]\rightarrow [h]$ and $\beta:[h] \rightarrow [r]$ be such that
    $\beta \circ \alpha =\id_{[r]}$. Then for any $\pi \in [r]^{[n]}$ it holds that
    \[
    f^\pi = f^{\beta\circ (\alpha \circ \pi)} = g^{\beta\circ (\alpha \circ \pi)} = g^\pi,
    \]
    as we wanted to show. 
\end{proof}

\paragraph{Minion Closures}

Let $A$ and $B$ be finite sets, and consider an arbitrary minion $\Mscr$ whose $n$-ary elements are functions of the form $f:A^n \rightarrow B$ where minoring is defined as usual. One of the insights of \cite{BG21:sicomp} is that a minion of this kind is a polymorphism minion of a finite template if and only if it has bounded \emph{finitised arity}. In other words, there is a number $h\in \NN$ such that for any $n\in \NN$ and any function $f: A^n\rightarrow B$ it holds that $f\in \Mscr$ if all $h$-ary minors of $f$ belong to $\Mscr$. Roughly, this means that every consistent family of $h$-ary elements in $\Mscr$ occurs as the family of $h$-ary minors of a higher arity element. Conversely, in the case where $\Mscr$ is not isomorphic to the polymorphism minion of a finite template, it must be that some consistent family of $h$-ary elements in $\Mscr$ is not represented by a higher arity element. In a way, this means that $\Mscr$ is missing some elements. To prove \Cref{th:minion_to_template}, we first define the $h$-closure of the minion $\Mscr$ which results from adding to it an element representing each consistent family of $h$-ary minors. Then, we show how to construct a finite template $(\A, \B)$ such that $\pol(\A, \B)$ is isomorphic to this closure. \par


An \emph{$m$-ary system of $k$-ary minors} over $\Mscr$ is a map $\zeta:[k]^{[m]} \rightarrow \Mscr(k)$ satisfying that for any pair of maps $\pi_1,\pi_2\in [k]^{[m]}$ and any map $\sigma\in [k]^{[k]}$ for which $\pi_1 = \sigma \circ\pi_2$ it holds that $\zeta(\pi_2)^\sigma = \zeta(\pi_1)$. 
If $\zeta$ is an $m$-ary system, and $\sigma \in [n]^{[m]}$ is a map, we denote by $\zeta^\sigma$ the $n$-ary system corresponding to the map $\pi \mapsto \zeta(\pi\circ \sigma)$. The $h$-closure of a minion $\Mscr$ is another minion $\Mscr^{(h)}$ whose $n$-ary elements are the $n$-ary systems of $h$-ary minors over $\Mscr$, and where minoring is given by the operation $\zeta\mapsto \zeta^\pi$. \par
Any element $p\in \Mscr(n)$ defines a $n$-ary system $\zeta_p\in \Mscr^{(h)}(n)$ in a natural way. That is, for any $\pi:[n]\rightarrow [h]$ we set $\zeta_p(\pi)= p^\pi$. This mapping $p \mapsto \zeta_p$ is, in fact, a minion homomorphism  $\mathrm{Cl}_h: \Mscr \rightarrow \Mscr^{(h)}$, which we call the \emph{canonical homomorphism} from $\Mscr$ to $\Mscr^{(h)}$. \par

Now we prove that the $h$-closure of a minion $\Mscr$ satisfies both properties we require from $\pol(\A, \B)$ in \Cref{th:minion_to_template}.

\begin{proposition}
\label{th:closure_equivalent}
    Let $\Mscr$ be a minion and let $h\in \NN$. Then the 
    restriction of the canonical homomorphism $\mathrm{Cl}_h: \Mscr \rightarrow \Mscr^{(h)}$ to elements of arity at most $h$ is a partial isomorphism.    
\end{proposition}
\begin{proof}
Let $F$ be the restriction of $\mathrm{Cl}_h$ to elements of arity at most $h$. We define another partial homomorphism $H: \Mscr^{(h)}\overpartialmap{h} \Mscr$ and show that $F$ and $H$ are inverses. For each $k\leq h$ we fix maps $\sigma_k: [k]\rightarrow [h]$ and $\pi_k:[h] \rightarrow [k]$ satisfying
$\pi_k \circ \sigma_k = \id_{[k]}$. Then, given a $k$-ary element $\zeta\in \Mscr^{(h)}$, we define $H(\zeta)$ as $(\zeta(\sigma_k))^{\pi_k}$.

Given $n\leq h$ and an element $f\in \Mscr(n)$, we define $F(f)$ as the $n$-ary system $\zeta_f$ of $h$-ary minors that maps each function $\pi\in [h]^{[n]}$ to $f^\pi$. The map $F$
    defined this way is clearly a $h$-partial minion homomorphism. Now let us show that $F$ is injective. Given $n\leq h$, 
    it is possible to find maps $\pi: [n]\rightarrow [h]$ and $\sigma:[h]\rightarrow [n]$ such that $\sigma \circ \pi = \id_{[n]}$.
    Hence
    $f^\pi = g^\pi$ implies $f=g$
    for all pairs $f,g\in \Mscr(n)$. In particular, this means that $\zeta_f= \zeta_g$ if and only if $f=g$. Finally, let us prove that $F$ is surjective. Let $n\leq h$
    and let $\pi, \sigma$ be the same maps as before. Consider an arbitrary $n$-ary system $\zeta$ of $h$-ary minors over $\Mscr$. We claim that $\zeta= \zeta_f$, where $f= \zeta(\pi)^\sigma$. To prove this we need to show that $\zeta(\pi^\prime)= f^{\pi^\prime}$ for all $\pi:[n]\rightarrow [h]$. Observe that $\pi^\prime = \pi^\prime  \circ \sigma \circ \pi$. Thus, by the definition of system, $\zeta(\pi^\prime)= \zeta(\pi)^{\pi^\prime \circ \sigma}$. However, 
    $\zeta(\pi)^{\pi^\prime \circ \sigma}= f^{\pi^\prime}$, proving that $\zeta=\zeta_f$. 
\end{proof}

\begin{proposition}
\label{th:minion_closure_univ_property}
    Let $\Mscr$ and $\Nscr$ be minions, and $h\in \NN$ be a number. Then $\Mscr \overpartialmap{h} \Nscr$ 
    if and only if
    $\Mscr \rightarrow\Nscr^{(h)}$. 
\end{proposition}
\begin{proof}
  Suppose there is a minion homomorphism $F: \Mscr \rightarrow \Nscr^{(h)}$. Then the restriction of $F$ to elements of arity at most $h$ 
    yields a partial homomorphism from $\Mscr$ to  $\Nscr^{(h)}$ up to arity $h$. By~\Cref{th:closure_equivalent}, this implies there is a partial homomorphism from $\Mscr$ to $\Nscr$ up to arity $h$. \par 
    Now suppose there is a partial homomorphism $F: \Mscr \overpartialmap{h} \Nscr$. We use $F$ to define a minion homomorphism $F^\prime: \Mscr \rightarrow \Nscr^{(h)}$. Let $n\in \NN$, and let $f\in \Mscr(n)$. Then we define $F^\prime(f)$ to be the  system $\zeta_f$ that sends each map $\pi\in [h]^{[n]}$ to $F(f^\pi)$. The system $\zeta_f$ is well-defined:
    if $\pi_1=\sigma \circ \pi_2$ for 
    some maps $\pi_1, \pi_2 \in [h]^{[n]}$, $\sigma\in [h]^{[h]}$, then $\zeta_f(\pi_1)= F(f^{\sigma\circ \pi_2})= F(f^{\pi_2})^\sigma= \zeta_f(\pi_2)^\sigma$.
    Finally, the map $F^\prime$ is a minion homomorphism. Indeed, if $f=g^\pi$, then $\zeta_f=\zeta_g^\pi$ following the definition of minoring for systems. \par 
    \end{proof}

\paragraph{Finite Templates}

We introduce two kinds of structures that will be used in our templates. Let $h\geq r$ be two natural numbers. Let $m= r^h$, and let $\pi_1,\dots \pi_m$ be the lexicographical ordering of $[r]^{[h]}$. The \emph{complete structure} $\bK^h_r$ is the relational structure whose signature consists of a single  $m$-ary symbol $R$, whose universe is $[r]$, and where $R^{K_r^h}$ is defined as the set of tuples $(\pi_1(i),\dots, \pi_m(i))$, where $i\in [h]$. This construction was given without a name in \cite[Lemma 6.7]{BG21:sicomp}. This is the most general structure on $r$ elements with relations of size at most $h$, in the sense that any other such structure is pp-definable on $\bK^h_r$. \par
The second kind of structures we use are the so-called \emph{free structures}, introduced in \cite{BBKO21}. Let $\Mscr$ be a minion and let $\A$ be a $\Sigma$-structure. Let $n=|A|$ and identify $A=[n]$ in some fixed way. Similarly, for each $R\in \Sigma$, let $m_R=|R^A|$, and identify $R^A$ with $m_R$ in a fixed way. The \emph{free structure of $\Mscr$ generated by $\A$} is a $\Sigma$-structure, denoted $\bF= \bF_\Mscr(\A)$ has universe $F=\Mscr(n)$, and for each symbol $R\in \Sigma$ the relation $R^F$ is given by the set of tuples $(f_1,\dots, f_{\arty(R)})$ for which there is an element $g\in \Mscr(m)$ satisfying $g^{\pi_i}=f_i$ for each $i\in [\arty(R)]$, where $\pi_i: [m]\rightarrow [n]$ is the $i$-th projection $r\mapsto r(i)$ (recall that $r\in R^A$ is seen as an element of $m$, and $r(i)\in A$ as an element of $[n]$). \par

Having defined complete structures and free structures, we are in conditions to prove the main result of the section.

\begin{proof}[Proof of \Cref{th:minion_to_template}]
The template witnessing the statement is given by $\A= \bK^h_r$ and $\B= \bF_{\Mscr}(\A)$. 
Observe that if $\Mscr$ is locally finite, then $\B$ is a finite structure. 
We prove that $\pol(\A, \B)$ is isomorphic to the closure $\Mscr^{(h)}$. Observe that this proves the theorem: the first item follows then from \Cref{th:closure_equivalent}, and the second from \Cref{th:minion_closure_univ_property}. \par

Let $\Nscr= \pol(\A, \B)$, $m= r^h$, and $\pi_1, \dots, \pi_m$ be the lexicographical ordering of $[r]^{[h]}$. We see elements $p\in \Nscr(n)$
    as maps from $[r]^{[n]}$ to $\Mscr(r)$ by identifying $[r]^n$ with $[r]^{[n]}$ as usual. Given a map $\sigma \in [h]^{[n]}$, the elements $(p(\pi_1\circ\sigma),\dots, p(\pi_m\circ \sigma))$ must belong to the relation $R^B$. In other words, there must be some element $f\in \Mscr(h)$ such that $f^{\pi}= p(\pi\circ \sigma)$ for each $\pi \in [r]^{[h]}$. Moreover, as the rank of $\Mscr$ is bounded by $r$, there can be only one element $f$ with such property. We denote $f$ as $p^{*}(\sigma)$.  \par
 Now we are able to define an isomorphism $H: \Nscr \rightarrow \Mscr^{(h)}$.  Given $n\in \NN$ and a polymorphism $p\in \Nscr(n)$, we define $H(p)$ as the system $\zeta_p$ that sends each map $\sigma \in [h]^{[n]}$ to the element 
    $p^*(\sigma)\in \Mscr(h)$. Let us show that $\zeta_p$ is a well-defined system. Let $\sigma\in [h]^{[n]}$, $\gamma \in [h]^{[h]}$, $f=p^*(\sigma)= \zeta_p(\sigma)$, and $g= p^* 
    (\gamma \circ \sigma) = \zeta_p(\gamma \circ \sigma)
    $.
    We need to show that $f^\gamma=g$. 
    Given $\pi \in [r]^{[h]}$, using the definition of $p^*$ and the fact that $\pi\circ \gamma \in [r]^{[h]}$ we obtain 
    \[
    (f^{\gamma})^\pi = f^{\pi\circ \gamma} = p( (\pi \circ \gamma) \circ \sigma) = p(\pi \circ (\gamma \circ \sigma)) = g^{\pi}.
    \]
    As the rank of $\Mscr$ is at most $r$, \Cref{le:rank} implies $f^\gamma=g$, as we wanted. \par   
    Now let us prove that $H$ is a minion homomorphism. Let $p\in \Nscr(n_1)$ be a polymorphism, and let $q=p^\gamma$ where $\gamma\in [n_2]^{[n_1]}$. We need to show that $ \zeta_q = \zeta_p^\gamma$. Let $\sigma\in [h]^{[n_2]}$,
    be an arbitrary map,
     $g= \zeta_q(\sigma)$, and
    $f=\zeta_p^\gamma(\sigma)= \zeta_p(\sigma \circ \gamma)$. It is enough to prove that $f=g$. For each $\pi \in [r]^{[h]}$ we have that $f^{\pi}=p(\pi \circ \sigma \circ \gamma)= q(\pi\circ \sigma) = g^\pi$. Using that the rank of $\Mscr$ is at most $r$, this shows that $f=h$, completing the proof.  \par
    Now let us prove that $H$ is a bijective map. This is enough because the inverse map of a bijective minion homomorphism is also a minion homomorphism.  As $r\leq h$, there are two maps $\sigma\in [h]^{[r]}$ and $\gamma \in [r]^{[h]}$ such that $\id_{[r]}=\gamma \circ \sigma$. 
    To see that $H$ is injective,
    let $p, q\in \Nscr(n)$ be two different polymorphisms.
    Using that the rank of $\Nscr$ is at most $r$, \Cref{le:rank}
    yields some $\tau\in [r]^{[n]}$ such that $p(\tau)\neq q(\tau)$. But    
    $p(\tau)= \zeta_p(\sigma \circ \tau)^\gamma$, and
    $ q(\tau)= \zeta_q(\sigma \circ \tau)^\gamma$, so $ H(p) = \zeta_p\neq H(q) = \zeta_q$. Finally, we show that $H$ is surjective. 
    Let $\zeta\in \Mscr^{(h)}(n)$ be an arbitrary system. We need to prove that there is some polymorphism $p\in \Nscr(n)$ such that $\zeta= \zeta_p$. Define $p: [r]^{[n]}\rightarrow \Mscr(r)$ as the function that maps each $\tau\in [r]^{[n]}$ to $\zeta(\sigma\circ \tau)^\gamma$. We need to show that $p$ is a $n$-ary polymorphism in $\Nscr$.
    For this it is enough to show that
    $(p(\pi_1\circ \tau^\prime), \dots, p(\pi_m\circ \tau^\prime))\in R^B$ for each map $\tau^\prime \in [h]^{[n]}$. Indeed, 
    $p(\pi\circ \tau^\prime) =\zeta(\sigma\circ \pi\circ \tau^\prime)^\gamma$ for all $\pi\in [r]^{[h]}$. However, by the definition of system, 
    this last element equals 
    $\zeta(\tau^\prime)^{\gamma \circ (\sigma \circ \pi)}= \zeta(\tau^\prime)^{\pi}$, 
    where the last equality uses the fact that $\gamma \circ \sigma = \id_{[r]}$. The fact that 
     $p(\pi\circ \tau^\prime) = \zeta(\tau^\prime)^{\pi}$ 
     for all $\pi \in [r]^{[h]}$ implies
     that
    $(p(\pi_1\circ \tau^\prime), \dots, p(\pi_m\circ \tau^\prime))\in R^B$. To see that $\zeta= \zeta_p$, observe that $p^*(\tau^\prime)$ must equal $\zeta(\tau^\prime)$ for all $\tau^\prime \in [h]^{[n]}$. 
 \end{proof}

Rather than referencing this result, it will be more convenient to use the following more concrete corollary, which simply follows from the fact that the template constructed in the last proof is precisely $(\bK^h_r, \bF_\Mscr(\bK^h_r))$.

\begin{corollary}
    \label{cor:minion_to_template}
        Let $\Mscr$ be a minion whose rank is at most $r$, and let $h\geq r$. Then
        the template $(\A, \B)= (\bK_r^h, \bF_\Mscr(\bK^h_r))$ satisfies the following.
        \begin{enumerate}[label=({\arabic*})]
            \item There is a partial minion isomorphism $\Mscr \overpartialmap{h} \pol(\A, \B)$ defined up to arity $h$.
            \item Given a minion $\Mscr^\prime$, it holds that $\Mscr^\prime \rightarrow \pol(\A, \B)$ if and only if 
           $\Mscr^\prime \overpartialmap{h}\Mscr$.
        \end{enumerate}
        Moreover, if $\Mscr$ is locally finite, then $\B$ is a finite structure. 
\end{corollary}

\section{Main Reductions}
\label{sec:reductions}

In this section we show the reductions that enable our main results. Most proofs are short, as they are just a matter of putting together the pieces we have constructed up until now. As in the overview given in \Cref{sec:warm-up}, the general strategy can be described as follows. We start with an interpretation $\Ical$ over a minion $\Mscr$ which is internal at some arity $h$ w.r.t. a description $\Dcal$. Then we are able to reduce problems related to the global structure $\bGS_\Ical$ to problems related to $\Mscr$ by considering manifold minions of the form $\bT^{\Ical/\Dcal}$ for finite structures $\bT$ similar to $\bGS_\Ical$, and then obtaining finite templates from these manifold minions applying \Cref{cor:minion_to_template}.

\paragraph{Reductions for Undecidability}

\begin{theorem}
\label{th:reduction_undecidability}

Let $\bG$ be a $\Sigma$‑structure, let $\Mscr$ be a minion, let $\Dcal\subseteq 2^\Mscr$ be a finite description, and let $\Ical$ be a $\Sigma$‑interpretation over $\Mscr$.
Let $h,r\in\mathbb{N}$ with $h\ge r$.
Suppose that 
    \begin{enumerate}[label=({\arabic*})]
        \item $\Hom(\bG, \cdot)$ is undecidable,
        \item $\Ical$ is internal at arity $h$ w.r.t. $\Dcal$,
        \item $\bGS_\Ical$ is finitely equivalent to $\bG$, and
        \item $r\geq \arty(Q)$ for every $Q\in \Dcal$ and $r\geq \arty(U^\Ical)$.
    \end{enumerate}
    Then the family of finite templates of the form $(\bK^h_r, \B)$ 
    satisfying $\Mscr \rightarrow \pol(\A, \B)$
    is undecidable.
\end{theorem}
\begin{proof}[Proof of \Cref{th:reduction_undecidability}]
We show that, given a finite $\Sigma$-structure $\C$, there is an algorithm that computes a 
finite template of the form $(\bK^h_r, \B)$ such that $\Mscr \rightarrow \pol(\bK^h_r, \B)$ if and only if $\bG \rightarrow \C$.
Let $\Nscr$ be the manifold minion $\C^{\Ical/\Dcal}$. We define $\B= \bF_\Nscr(\bK^h_r)$. Let us reassure the reader that $\B$ is constructible from $\C$, even though this may not be immediately obvious. This follows from observing that $\B$ can be constructed when given access to the restriction of $\Mscr/\Dcal$ to elements of arity at most $h$
and to the interpretation $\Ical/\Dcal$, and these are finite objects that are independent from the input $\C$, so they can be precomputed. \par

Let us see that $\Mscr \rightarrow \pol(\bK^h_r, \B)$ if and only if $\bG \rightarrow \C$. This is a consequence of the following chain of double implications:

\begin{align*}
&
\bG \rightarrow \C \quad & \xLeftrightarrow{\text{\Cref{lem:hom_to_power_minion}}} \\
&
\Mscr \overpartialmap{h} \Nscr \quad & \xLeftrightarrow{\text{\Cref{cor:minion_to_template}}} \\
&
\Mscr \rightarrow \pol(\bK^h_r, \B). & 
\end{align*}

Let us spell the arguments out. As $\bGS_\Ical$ is finitely equivalent to $\bG$, 
$\bG \rightarrow \C$ holds if and only if $\bGS_\Ical \rightarrow \C$. 
By~\Cref{lem:hom_to_power_minion}, there is a homomorphism from $\bGS_\Ical$ to $\C$ if and only if there is a partial homomorphism $F:\Mscr \overpartialmap{h} \C^{\Ical/\Dcal}=\Nscr$ 
By \Cref{le:rank_manifold_minion}, the rank of $\Nscr$ is at most $r$, so 
we can apply~\Cref{cor:minion_to_template} to the minion $\Nscr$ and the template $(\bK^h_r, \B)=(\bK^h_r, \bF_\Nscr(\bK^h_r))$. In particular, $\Mscr \overpartialmap{h} \Nscr$ if and only if $\Mscr\rightarrow \pol(\bK^h_r, \B)$. This completes the proof.
\end{proof}

We also give an analogue of last result tailored for minions that can be expressed as disjoint unions of subminions. In this case, the global structure induced by an interpretation is a disjoint union of the global structures induced on each of the disjoint subminions (\Cref{obs:restricted_interpretation}), and we can apply the strategy from last proof to each of the parts. 

\begin{theorem}
\label{th:reduction_undecidability_disjoint_union}
    Let $\mathbb G$ be a family of $\Sigma$-structures. Let $\Mscr= \bigsqcup_{\bG \in \mathbb G} \Mscr_{\bG}$ be a minion, let $\Dcal\subseteq 2^\Mscr$ be a finite description, and
    let $\Ical$ be a $\Sigma$-interpretation over $\Mscr$. Let $h,r\in \NN$ with $h\geq r$. Suppose that 
    \begin{enumerate}[label=({\arabic*})]
        \item $\Ical$ is internal at arity $h$ w.r.t. $\Dcal$,
        \item $\bGS_{\Ical_{\bG}}$ is finitely equivalent to $\bG$ for each $\bG\in \mathbb{G}$,  
        where $\Ical_{\bG}= \Ical\vert_{\Mscr_{\bG}}$, and
        \item $r\geq \arty(Q)$ for every $Q\in \Dcal$ and $Q=U^\Ical$,  
    \end{enumerate}
    Then there is an algorithm that, given a finite $\Sigma$-structure $\C$, yields a 
    finite template $(\bK^h_r, \B)$ such that $\bG \rightarrow \bT$ if and only if 
    $\Mscr_{\bG} \rightarrow \pol(\bK^h_r, \B)$ for each $\bG\in \mathbb{G}$. In particular, the following hold:
    \begin{enumerate}[label=(\roman*)]
        \item The set $\Hom(\mathbb{G}, \cdot)$ is Turing-reducible to the family of finite templates $(\bK^h_r, \B)$ 
    satisfying $\Mscr_{\bG}\rightarrow \pol(\bK_r^h, \B)$ for some $\bG\in \mathbb{G}$.
    \item The set $\Hom_{\eg}(\mathbb{G}, \cdot)$ is Turing-reducible to the family of finite templates $(\bK^h_r, \B)$ 
    satisfying $\Mscr_{\bG}\rightarrow \pol(\bK_r^h, \B)$ for all but finitely many $\bG\in \mathbb{G}$.
    \item The set $\Hom_{\io}(\mathbb{G}, \cdot)$ is Turing-reducible to the family of finite templates $(\bK^h_r, \B)$ 
    satisfying $\Mscr_{\bG}\rightarrow \pol(\bK_r^h, \B)$ for infinitely many $\bG\in \mathbb{G}$.
    \end{enumerate}
\end{theorem}
\begin{proof}[Proof of~\Cref{th:reduction_undecidability_disjoint_union}]
Given a finite $\Sigma$-structure $\C$, define $\Nscr=\C^{\Ical/\Dcal}$ and $\B=\bF_\Nscr(\bK^h_r)$.
As in the proof of~\Cref{th:reduction_undecidability}, there is an algorithm that
computes the template $(\bK^h_r, \B)$ in response to the input $\C$.
In this case, \Cref{lem:hom_to_power_minion_disjoint_union} tells us that $\Mscr_{\bG} \overpartialmap{h} \Nscr$ if and only if $\bG \rightarrow \C$ for a given $\bG\in \mathbb{G}$. By the same arguments as in the proof of~\Cref{th:reduction_undecidability}, this is the case if and only if $\Mscr_{\bG} \rightarrow \pol(\bK^h_r,\B)$, as we wanted to show. \par
Observe that the last three items in the statement of the theorem are a direct consequence of the facts that $\bG \rightarrow \bT$ if and only if $\Mscr_{\bG} \rightarrow \pol(\bK^h_r, \B)$ for each $\bG\in \mathbb{G}$, and that 
$(\bK^h_r, \B)$ can be computed in response to $\C$.
\end{proof}

\paragraph{Reduction for Non-Computability}

\begin{theorem}
\label{th:reduction_non_computability}
 Let $\bG$ be a $\Sigma$-structure, $\Mscr$ a minion, $\Dcal\subseteq 2^\Mscr$ a finite description, 
    $\Ical$ a $\Sigma$-interpretation over $\Mscr$, and $h\geq r$ natural numbers.
    Suppose that 
    \begin{enumerate}[label=({\arabic*})]
        \item there is $\bT\in \Hom(\bG, \cdot)$ for which no homomorphism $H:\bG \rightarrow \bT$
        is computable,
        \item $\Ical$ is internal at arity $h$ w.r.t. $\Dcal$,
        \item $\bGS_\Ical$ is finitely equivalent to $\bG$, 
        \item there is a computable homomorphism $F:\bG \rightarrow \bGS_\Ical$, and
        \item $r\geq \arty(Q)$ for every $Q\in \Dcal$ and $Q=U^\Ical$. 
    \end{enumerate}
    Then there exists a finite template of the form $(\bK_r^h, \B)$
    satisfying that $\Mscr \rightarrow \pol(\bK^h_r, \B)$, but there is no computable minion homomorphism $H: \Mscr \rightarrow \pol(\bK^h_r, \B)$.
    \end{theorem}
\begin{proof}
Define $\Jcal=\Ical/\Dcal$, and $\Nscr= \bT^{\Jcal}$. We claim that the finite template satisfying the theorem's statement is $(\bK_r^h,\B)$, 
where $\B=\bF_\Nscr(\bK_r^h)$. By \Cref{le:rank_manifold_minion}, we have that $\Nscr$ has rank at most $r$, so $\Nscr$ and the template $(\bK_r^h, \B)$ witness \Cref{cor:minion_to_template}.
Hence $\bG\rightarrow \bT$ implies $\Mscr \rightarrow \pol(\bK_r^h,\B)$. We show that there is no computable minion homomorphism from
$\Mscr$ to $\pol(\bK_r^h, \B)$. \par
We proceed by contradiction. Suppose there is a computable minion homomorphism $H: \Mscr \rightarrow \pol(\bK_r^h,\B)$. We give a composition of computable partial maps that yields a homomorphism from $\bG$ to $\bT$. By \Cref{cor:minion_to_template}, there is a partial homomorphism $H^\prime: \pol(\bK_r^h, \B) \overpartialmap{h} \Nscr$. Observe that $H^\prime$ is given by a finite table (i.e, it is defined on a finite set, and its co-domain, consisting of the elements $f\in \Nscr$ of arity bounded by $h$, is finite). Hence, $H^\prime$ is computable. 
This way, $H^\prime \circ H$ yields a computable partial homomorphism from $\Mscr$ to $\Nscr$ defined up to arity $h$. Let $\pj: \Nscr \partialmap \bT$ be the canonical projection (recall the definition from \Cref{sec:manifold_minion}). This, again, is a computable partial map. By~\Cref{lem:hom_to_power_minion}, the map $\pj \circ H^\prime \circ H$ restricted to $U^\Ical$ is a homomorphism from $\bGS_\Ical$ to $\bT$ (we do not require the restriction to be computable; we just use that $\pj \circ H^\prime \circ H$ is computable). By hypothesis, there exists a computable homomorphism $F: \bG \rightarrow \bGS_\Ical$. Then $\pj \circ H^\prime \circ H \circ F$ is a computable homomorphism from $\bG$ to $\bT$, yielding a contradiction. Thus, there cannot be a computable homomorphism from $\Mscr$ to $\pol(\bK_r^h,\B)$, as we wanted to prove.
\end{proof}

\paragraph{Reduction for Hardness} ~\\ 

\vspace{-0.1em}
\noindent
Let us take another look at the algorithms $\Qcal \in \{\AIP, \BLP, \BLP+\AIP\}$. Unrolling the definitions, it turns turns out that for a given instance $\I$ the following are equivalent: (1) $\I$ is accepted by $\Qcal$, and (2) $\I\rightarrow \bF_{\Mscr_\Qcal}(\A)$. This is shown in \cite{BBKO21} for $\BLP$ and $\AIP$, and further discussed in \cite{BGWZ20} for the case of $\BLP+\AIP$. From this we obtain the following alternative formulation of rounding problems as left-infinite PCSPs.
\begin{fact}
    \label{fact:rounding_to_pcsp}
   Let $\Qcal \in \{\AIP, \BLP, \BLP+\AIP\}$, and  $(\A, \B)$ be a finite template. 
   Suppose that $\Qcal$ solves $\pcsp(\A, \B)$. Then the problems $\spcsp_\Qcal(\A, \B)$ and
   $\spcsp(\bF_{\Mscr_\Qcal}(\A), \B)$ are the same.
\end{fact}

The \emph{projection minion} $\mathscr{P}$ is defined as the minion where $\mathscr{P}(n)= [n]$ and where 
$f^\pi= \pi(f)$ for each $f\in \mathscr{P}(n)$ and each $\pi\in [m]^{[n]}$. If P$\neq$ NP, the minion $\mathscr{P}$ is homomorphically equivalent to $\pol(\A)=\pol(\A,\A)$ for each finite structure $\A$
for which $\csp(\A)$ is NP-hard.  It is also not difficult to show that $\mathscr{P}\rightarrow \Mscr$ for every non-empty minion $\Mscr$.  \par
One of the cornerstones of the algebraic approach to PCSPs is the result \cite[Theorem 3.12]{BBKO21} that, for a finite template $(\A, \B)$ and $h$ at least as large as $|A|$ and each relation of $\A$, the problems $\spcsp(\A, \B)$
and $\sPMC_h(\mathscr{P},\pol(\A, \B))$ are log-space equivalent, and similarly for their decision variants. The same proof actually shows the following more general result \footnote{To see that this indeed generalizes \cite[Theorem 3.12]{BBKO21}, choose $\Mscr$ to be the projection minion $\mathscr{P}$, and use the fact that $\bF_{\mathscr{P}}(\A)$ is isomorphic to $\A$.}. 

\begin{theorem}
\label{th:spcsp_to_spmc}
    Let $\Mscr$ be a minion, $N\in \NN$ a number, and $(\A,\B)$ a finite template satisfying that 
    $\Mscr \rightarrow \pol(\A, \B)$. Then there is a log-space reduction from $\sPMC_N(\Mscr, \pol(\A, \B))$ 
    to $\spcsp(\bF_\Mscr(\A), \B)$. Further suppose that $N$ is at least as large as $|A|$ and $|R^A|$ for any relation $R$ in the signature of $\A$. Then there is a log-space reduction from $\spcsp(\bF_\Mscr(\A), \B)$ to  $\sPMC_N(\Mscr, \pol(\A, \B))$.
\end{theorem}
\begin{proof}
    The reductions are essentially the ones given in the proof of \cite[Theorem 3.12]{BBKO21}. We sketch them below. 
    For the rest of the proof we identify $A$ with $[|A|]$, and $R^A$ with $[|R^A|]$ for each relation symbol $R$ in some fixed way. Given a relation symbol $R$ and an index $i\in [\arty(R)]$ we write $\Pi_{R,i}: R^A \rightarrow A$ for the map $\bm a\mapsto a_i$.\par
    \emph{Reducing from $\sPMC$ to $\spcsp$.} Let $\Phi$ be a minor condition of arity at most $N$. Without loss of generality, we may assume that $\Phi$ is of the form 
    $\exists x_1, \dots, x_k \varphi(x_1,\dots, x_k)$, where $\varphi$ is a conjunction of formulas of the form $x_i = x_j^\pi$. For each $i\in [k]$ we let $n_i$ be the arity of $x_i$. The first part of the reduction is the construction of an structure $\I_\Phi$ satisfying $\I_\Phi \rightarrow \bF_\Mscr(\A)$ if $\Mscr \models \Phi$. First, define the set $S=\{ x_i(\pi) \, \vert i\in [k], \pi \in A^{[n_i]} \}$. Now, we identify two elements $x_i(\pi)$ and $x_j(\sigma)$ if there is an identity in $\varphi$ of
    the form $x_i = x_j^\gamma$, where $\sigma = \pi\circ \gamma$. We define the universe $I_\Phi$ as
    the result of performing all these identifications in $S$, and write $\qclass{x_i(\pi)}$ to denote the equivalence class of $x_i(\pi)$ in $I_\Phi$. Given a relation symbol $R$ of arity $m$, we define $R^{I_\Phi}$ as the set of tuples of the form $(\qclass{x_{i}(\pi_1)},\dots,\qclass{x_{i}(\pi_m)})$ for which there is a map $\sigma\in [R^A]^{[n_i]}$ satisfying 
    $\Pi_{R,j}\circ \sigma = \pi_j$ for each $j\in [m]$. The structure $\I_\Phi$ can be built in log-space. 
    Now, if $x_i \mapsto p_i$ is a satisfying assignment of $\Phi$ in $\Mscr$, then the map
    $\qclass{x_i(\pi)}\mapsto p_i^\pi$ is a well-defined homomorphism from $\I_\Phi$ to $\bF_\Mscr(\A)$. Now, let $F: \I_\Phi \rightarrow \B$ be a homomorphism. The second part of the reduction computes an assignment of $\Phi$ over $\pol(\A, \B)$ from $F$. For each $i\in [k]$ we define $f_i: \A^{n_i} \rightarrow \B$ as the map 
    $\bm{a} \mapsto F(\qclass{x_i(\bm a)}$, where we recall that $\bm a$ can be seen as a map from $[n_i]$ to $A$. It is routine to verify that $f_i$ is indeed a homomorphism and that the map $x_i \mapsto f_i$ is a satisfying assignment of $\Phi$ in $\pol(\A, \B)$ that can be obtained in log-space. \par
    \emph{Reducing from $\spcsp$ to $\sPMC$.} Let $\I$ be an instance of $\spcsp(\bF_\Mscr(\A), \B)$. The first map of our reduction constructs in log-space a minor condition $\Phi_\I$ of arity at most $N$ that satisfies $\Mscr\models \Phi_\I$ if $\I\rightarrow \bF_\Mscr(\A)$. Let $\Sigma$ be the relational signature of all structures under consideration. Then the condition $\Phi_\I$ is defined as 
    \[
    \bigexists_{v\in I}^{|A|} x_v \bigexists_{R \in \Sigma, r \in R^I}^{|R^A|} y_{r} 
    \left(
    \bigwedge_{R \in \Sigma,  r \in R^I, i\in [\arty(R)]} 
    x_{r(i)} = y_{r}^{\Pi_{R,i}}
    \right).
    \]
    The arity of this minor condition is the maximum of $|A|$ and $|R^A|$ for all $R\in \Sigma$, so it is bounded by $N$ by assumption. 
    Now suppose that $F:\I \rightarrow \bF_\Mscr(\A)$ is a homomorphism. 
    Given $R\in \Sigma$, and $r\in R^I$
    we write $\widehat{F}(r)$ for the element $p\in \Mscr(R^A)$
    witnessing that $(F(r(1)),\dots, F(r(\arty(R)))\in R^{F_\Mscr(A)}$. Then the map $x_v\mapsto F(v)$ 
    together with $y_{r} \mapsto \widehat{F}(r)$ is a satisfying assignment of $\Phi_\I$ in $\Mscr$. 
    Now, let $H$ be a satisfying assignment of $\Phi_\I$ in $\pol(\A, \B)$. The second map of the reduction needs to construct a homomorphism $F:\I\rightarrow \B$ in log-space using $H$. Given $v\in I$, we define $F(v)$ as
    $H(x_v)(\bm a)$, where $\bm a\in A^{|A|}$ is our fixed enumeration of $A$. Now it is routine to check that 
    $F$ is indeed a homomorphism.
\end{proof}

By this point showing our main hardness reduction is just a matter of composing the reductions we have shown up until now.

\begin{theorem}\label{th:reduction_hardness}

 Let $\bG$ be a $\Sigma$-structure, $\Mscr$ a minion, 
 $(\Ical, \Dcal, \Psi, h)$ a pattern over $\Mscr$
 and $r\leq h$ a natural number. Suppose that 
    \begin{enumerate}[label=({\arabic*})]
        \item $\bGS_\Ical$ is finitely equivalent to $\bG$, and
        \item $r\geq \arty(Q)$ for every $Q\in \Dcal$ and $Q=U^\Ical$.
    \end{enumerate}
Then for each $\C\in \Hom(\bG, \cdot)$ there is a finite template $(\bK_r^h, \B)$
such that $\spcsp(\bG, C)$ is many-one reducible to 
$\spcsp(\bm{F}_\Mscr(\bK^h_r), \B)$.
\end{theorem}
\begin{proof}
Let $\C \in \Hom(\bG, \cdot)$. Let $\Nscr= \C^{\Ical/\Dcal}$. By \Cref{le:patterns} there is a many-one reduction from $\spcsp(\bG, \bT)$ to
$\sPMC_h(\Mscr,\Nscr)$. 
By  \Cref{le:rank_manifold_minion} the rank of $\Nscr$ is at most $r$, 
so the minion $\Nscr$ and the template $(\bK_r^h, \B)= (\bK_r^h, \bF_{\Nscr}(\bK_r^h))$
satisfy \Cref{cor:minion_to_template}. Hence, there is a partial isomorphism $F: \pol(\bK_r^h, \B) \overpartialmap{h}  \Nscr$. This yields a many-one reduction from 
$\sPMC_h(\Mscr, \Nscr)$ to $\sPMC_h(\Mscr, \pol(\bK_r^h, $ $\B))$. 
Finally, by \Cref{th:spcsp_to_spmc},
the problem $\sPMC_h(\Mscr, \pol(\bK_r^h, \B))$ has a many-one reduction to
$\spcsp(\bF_\Mscr(\bK_r^h), \B)$. 
\end{proof}

\section{Proof of the Main Results}
\label{sec:main_proofs}

In this section we present the proofs of \Cref{th:AIP_main,th:BLP_main,th:BLP+AIP_main}, and 
\Cref{th:minor_identities_main}. Each of these proofs consists of reducing one of the source problems described in \Cref{sec:sources} to our target problems through the reductions shown in \Cref{sec:reductions}. The lengthy lists of requirements for our reductions indicate that applying them requires some work: even after coming up with a suitable interpretation $\Ical$ and a suitable description $\Dcal$, proving that $\Ical$ is internal w.r.t. $\Dcal$ requires multiple applications of 
\Cref{le:internal_criterion}, which means constructing multiple logical formulas, and analysing their satisfying assignments on the original minion $\Mscr$ and the quotient $\Mscr/\Dcal$.

To make matters worse, this section makes painfully apparent that there is, in principle, no way of transferring our results between the minions $\Mscr_{\AIP}, \Mscr_{\BLP}, \Mscr_{\BLP+\AIP}$ despite their similarities. This means that we need to write very similar proofs (with small, but significant changes) over and over. We remark that minion homomorphisms do not yield reductions for any of our results. Indeed, suppose that $\Mscr \rightarrow \Nscr$. Then, in general there is no relation between the decidability of the finite templates $(\A, \B)$ such that $\Mscr \rightarrow \pol(\A, \B)$, and the decidability of those satisfying $\Nscr \rightarrow \pol(\A, \B)$. Similarly, 
in the absence of additional restrictions, it is possible that all homomorphisms $F: \Nscr \rightarrow \pol(\A, \B)$ are non-computable, but there is a computable homomorphism $H: \Mscr \rightarrow \pol(\A, \B)$, or vice versa. As for hardness, it is true that
$\spcsp(\bF_{\Mscr}(\A), \B)$ has a straight-forward reduction to $\spcsp(\bF_\Nscr(\A), \B)$, so it may seem that the family of rounding problems arising from $\Mscr$ is easier than the one resulting from $\Nscr$. However, this is not true in general: There may be templates $(\A, \B)$ for which $\Mscr\rightarrow \pol(\A, \B)$ but $\Nscr \not \rightarrow \pol(\A, \B)$, implying that there are \emph{more} rounding problems for $\Mscr$ than for $\Nscr$, and these may be hard. \par

\paragraph{How to read this section} ~\\
\par
\vspace{-0.1em}
\noindent
We try to avoid being meticulous and redundant to the point of exhausting the reader, and being so careless as to skip relevant details. The bulk of the following subsections is spent proving that several interpretations are internal with respect to corresponding descriptions. This is done through several claims. We believe that after getting acquainted with the typical arguments, the proofs of many of these claims are routine. \par

\Cref{sec:AIP,sec:BLP,sec:BLP+AIP} are conveniently ordered from less to more difficult, and they all employ different versions of the arguments in \Cref{sec:AIP}. These subsections are further divided to show interpretations of $\bm{\Gamma}^+$ and interpretations of $\bm\Gamma$. Our interpretations of $\bm{\Gamma}^+$ are relatively simpler than the ones of $\bm\Gamma$, because the interpretations in this second group are meant to produce Boolean templates. Hence, a good overview of the following subsections is the following.
\begin{itemize}
    \item \textbf{Interpretations of the super-grid $\bm{\Gamma}^+$} (ordered from simpler to more complex):
    \begin{center}
        \Cref{sec:3-dim_grid_AIP} (\AIP), \, \, \, \,
        \Cref{sec:3-dim_grid_BLP} (\BLP), \, \, \, \, and \, \,
        \Cref{sec:3-grid_BLP+AIP} (\BLP+\AIP).
    \end{center}
        \item \textbf{Interpretations of the grid $\bm{\Gamma}$} (ordered from simpler to more complex):
    \begin{center}
        \Cref{sec:grid_AIP} (\AIP), \, \, \, \, and \, \,
        \Cref{sec:grid_BLP} (\BLP).
    \end{center}
\end{itemize}

The proofs in \Cref{sec:WNU} and \Cref{sec:cyclic} handle WNU and cyclic polymorphisms respectively, and stand on their own. In both these cases we define minions characterizing the presence of those polymorphisms, and then find ways to interpret useful structures (e.g., grids) on them. However, these minions lack the arithmetical structure of $\Mscr_{\AIP}, \Mscr_{\BLP}$, and $\Mscr_{\BLP+\AIP}$, so we need to use different arguments. There is ample difference between \Cref{sec:WNU} and \Cref{sec:cyclic}, but they handle minions constructed in a similar way, and some arguments are analogous. They are, again, ordered from simpler to more complex, so we recommend the reader to tackle 
\Cref{sec:WNU} first. 

\subsection{The AIP Algorithm}
\label{sec:AIP}

We prove \Cref{th:AIP_main} in this section. 
To prove item (1) we give an interpretation of the super-grid $\bm{\Gamma}^+$ over $\Mscr_\AIP$, shown in \Cref{sec:3-dim_grid_AIP}. 
Items (2),(3), and (4) are proven similarly, by showing in \Cref{sec:grid_AIP} a suitable interpretation of the grid $\bm{\Gamma}$ over $\Mscr_\AIP$.

\subsubsection{AIP: Interpreting the Super-Grid}
\label{sec:3-dim_grid_AIP}
The following interpretation $\Ical$ 
induces a global structure $\bGS_\Ical$ that is finitely equivalent to $\bm{\Gamma}^+$. \par

\begin{interpret} The $\Sigma_{\Gamma^+}$-interpretation $\Ical$ over $\Mscr_\AIP$ is given by
    \begin{flalign*} 
 U^\Ical = 
\Bigl\{ (m_1, m_2, m_3, n) & \in \Mscr_\AIP(4) \, \Big\vert \, m_i> 0 \text{ for all } i=1,2,3\Bigr\}, &&
\end{flalign*}

\begin{flalign*}
& O^\Ical   = \Big\{ (1,1,1,-4) \Big\}, \quad \text{ and } \quad \Pi^\Ical_{O,1} =\id,  &&
\end{flalign*}
\begin{flalign*}
& E_i^\Ical  =
\Big\{
(m_1, m_2, m_3, 1, n) \in \Mscr_\AIP(5) \Big\}, \quad \text{ and }   && \\
&
\Pi^\Ical_{E_i, j} = 
\begin{cases}
    (1,2,3,4,4), \text{ when $j=1$},\\
    (1,2,3,i,4) \text{ when $j=2$},
\end{cases}
\text{ for each $i\in [3]$,}  &&
\end{flalign*}

\begin{flalign*}
& \Ebb_i^\Ical =
\Big\{ 
(m_1, m_2, m_3, m_i, n) \in \Mscr_\AIP(5) \Big\}, \quad \text{ and }  && \\ & \Pi^\Ical_{\Ebb_i,j} =
 \begin{cases}
     (1,2,3,4,4) \text{ for $j=1$}, \\
     (1,2,3,i,4) \text{ for $j=2$}
 \end{cases} \text{ for each $i\in [3]$.}  &&
\end{flalign*}
\end{interpret}

\vspace{1em}
\noindent 
This way, $\bm{\Gamma}^+$ is easily seen to be isomorphic to $\bGS_\Ical$ via the homomorphism $(m_1, m_2, m_3)\mapsto (m_1, m_2, m_3, 1- m_1 - m_2- m_3)$. Now we define a description $\Dcal\subseteq 2^{\Mscr_\AIP}$ so that $\Ical$ is internal at arity $5$ w.r.t. $\Dcal$. 
\begin{desc} The description $\Dcal \subseteq 2^{\Mscr_{\AIP}}$ consists of the predicates
\begin{flalign*}
&D_<= \Big\{ (m_1, m_2, m_3)\in \Mscr_\AIP \, \Big\vert \, m_1 < m_2 \Big\}, &&  \\
&D_1=\Big\{ (1,0) \Big\}. &&
\end{flalign*}
\end{desc}
\vspace{1em}

The interpretation $\Ical$ and the description $\Dcal$ are the simplest in \Cref{sec:main_proofs}. The fact that we are not aiming to obtain Boolean templates means that we have virtually no restrictions on $\Ical$ and $\Dcal$, and the minion $\Mscr_{\AIP}$ already has the arithmetical structure required to express unit increments and multiplication times two.  The following claims show that $\Ical$ is internal at arity $5$ (w.r.t. $\Dcal$): \\~\\
\noindent
\textbf{Claim 1: $R^\Ical$ is $\Dcal$-stable for each $R\in \Sigma_{\Gamma^+}$.} This follows by a direct application of the definitions. \\~\\
\noindent
    \textbf{Claim 2: $D_1$ is internal at arity $5$.}
    Indeed, the minor condition $\phi_1(x)\equiv
    x= x^{(1,1)}$ is an internal definition of $D_1$.\\~\\
    \noindent
    \textbf{Claim 3: } \textbf{$O^\Ical$ is internal at arity $5$.} The following is an internal definition of $O^\Ical$:
    \[
    \phi_O(x)\equiv  
    \phi_1(x^{(1,2,2,2)}) \wedge
    \phi_1(x^{(2,1,2,2)}) \wedge
    \phi_1(x^{(2,2,1,2)}).
    \] 
    ~\\
    \noindent
    \textbf{Claim 4:} \textbf{$E_i^\Ical$ is internal at arity $5$ for all $i\in [3]$.} The following is an internal definition of $E_i^\Ical$:
    \[
    \phi_{E_i}(x) \equiv \phi_1(x^{(2,2,2,1,2)}).
    \] ~\\
    \noindent
    \textbf{Claim 5: }  \textbf{$\Ebb_i^\Ical$ is internal at arity $5$ for all $i\in [3]$.}
    The following is an internal definition of $\Ebb_i^\Ical$:
    \[
    \phi_{E_i}(x) \equiv x^{\sigma}= x^{\tau},
    \]
    where $\sigma\in [3]^{[5]}$ is the map sending $i$ to $1$, $4$ to $2$, and the other elements to $3$,
    and $\tau\in [3]^{[5]}$ sends $i$ to $2$, $4$ to $1$, and the other elements to $3$. The key insight is that if $f^\sigma \sim_\Dcal f^\tau$ for some element $f\in \Mscr_\AIP(5)$, then it cannot be that $f(i)< f(4)$ or $f(i)> f(4)$, so it must be that $f(i) = f(4)$. \\~\\ \noindent
    \textbf{Claim 6: } \textbf{Let $i\in [3]$. Suppose that $\phi(x)$ is an internal reference to $U^\Ical$. Then the following formula is also an internal reference to $U^\Ical:$}
    \begin{align*}
   \phi^\prime(x) \, \equiv \, &
    \exists y \exists z \Bigl(
    \phi_{E_i}(z) \wedge \phi(y) \Bigr. \Bigl.
    \wedge \, 
    y = z^{(1,2,3,4,4)} \, \wedge \,
    x = z^{(1,2,3,i,4)}
    \Bigr).
    \end{align*}
    \textbf{Moreover, if $f$ satisfies $\phi$ on $\Mscr_\AIP$, then
    $g$ satisfies $\phi^\prime$ on $\Mscr_\AIP$,
    where $g(i)= f(i)+1$, $g(4)= f(4) - 1$, and
    $g(j)=f(j)$ for $j\neq i,4$.}
    Let us show that $\phi(x)$ is indeed an internal reference to $U^\Ical$.
    Suppose that $\Mscr_\AIP/\Dcal \models \phi(\qclass{f_x})$  with
    $\qclass{f_y}, 
    \qclass{f_z}$ as existential witnesses for $y,z$. The formula $\phi$ is an internal reference to $U^\Ical$, so
    $f_y\in U^\Ical$.
    By a similar reasoning we also obtain that $f_z\in E_i^\Ical$. It must hold that $f_z^{(1,2,3,4,4)}\sim_\Dcal f_y$, so $f_z^{(1,2,3,4,4)}\in U^\Ical$ (by Claim 1).
    Let $f_z=(m_1, m_2, m_3, 1, m_4)$. The fact that $f^{(1,2,3,4,4)}\in U^\Ical$ means that $m_1,m_2, m_3>0$. Hence, the fist three elements in $f^{(1,2,3,i,4)}$ must also be positive, so this tuple belongs to $U^\Ical$ as well. Using the fact that $f_z^{(1,2,3,i,4)}\sim_\Dcal f_x$ we conclude that $f_x\in U^\Ical$, proving that
    $\phi$ is an internal reference to $U^\Ical$. \par
    Finally, suppose that $\Mscr_\AIP\models\phi(f)$ for some element $f$, and let $g$ be defined as in the statement. Then $g$ satisfies $\phi^\prime(x)$
    on $\Mscr_\AIP$ with $f$ as an existential witness for $y$ and the tuple
    $(f(1), f(2), f(3), 1, f(4)-1)$ as 
    an existential witness for $z$. \\~\\
    \noindent
   \textbf{Claim 7: } \textbf{Let $i\in [3]$. Suppose that $\phi(x)$ is an internal reference to $U^\Ical$. Then the following formula is also an internal reference to $U^\Ical:$}
    \begin{align*}
    \phi^\prime(x) \equiv 
    \exists y \exists z \Bigl( 
    \phi_{\Ebb_i}(z) \wedge \phi(y)  \wedge 
    y = z^{(1,2,3,4,4)} \wedge 
    x = z^{(1,2,3,i,4)}
    \Bigr).
    \end{align*}
    \textbf{Moreover, if $f$ satisfies $\phi$ on $\Mscr_\AIP$, then
    $g$ satisfies $\phi^\prime$ on $\Mscr_\AIP$,
    where $g(i)= 2f(i)$, $g(4)= f(4) - f(i)$, and
    $g(j)=f(j)$ for $j\neq i,4$.} This follows analogously to the previous claim. \\~\\
    \noindent
    \textbf{Claim 8: } \textbf{$U^\Ical$ is internal at arity $5$.} We prove this claim using \Cref{le:internal_criterion}. We define an internal reference $\phi_{n_1,n_2,n_3}(x)$ to $U^\Ical$ inductively for each $(n_1,n_2,n_3)\in \NN^3$ following the lexicographical order, satisfying that $\Mscr_\AIP \models \phi_{n_1,n_2,n_3}( (n_1,n_2,n_3, 1- n_1 - n_2 - n_3))$. We define $\phi_{1,1,1}(x) \equiv \phi_O(x)$. Now let $(n_1, n_2, n_3)\in \NN^3$ be different from $(1,1,1)$, $i\in [3]$ be the largest index for which $n_i>1$. We have two cases. Suppose that $n_i$ is odd. Then we let $m_i=n_i-1$, and $m_j=n_j$ for $j\in [3]$, $j\neq i$, and define
    \begin{align*}
    \phi_{n_1, n_2, n_3}(x) \equiv 
    \exists y \exists z \Bigl(
    \phi_{m_1, m_2, m_3}(y) \wedge
    \phi_{E_i}(z) \wedge 
    z^{(1,2,3,4,4)}= y \wedge 
    z^{(1,2,3,i,4)}= x
    \Bigr).
    \end{align*}
    Otherwise, if $n_i$ is even, we 
    let $m_i=n_i/2$, and $m_j=n_j$ for $j\in [3]$, $j\neq i$, and define
    \begin{align*}
    \phi_{n_1, n_2, n_3}(x) \equiv 
    \exists y \exists z \Bigl( 
    \phi_{m_1, m_2, m_3}(y) \wedge
    \phi_{\Ebb_i}(z) 
     \wedge 
    z^{(1,2,3,4,4)}= y \wedge 
    z^{(1,2,3,i,4)}= x
    \Bigr).
    \end{align*}
    Now Claims 6 and 7 prove the statement. \\~\\

\begin{proof}[Proof of item (1) in \Cref{th:AIP_main}]
The claims in this section show that $\Ical$ is internal w.r.t. $\Dcal$ at arity $5$. We also have that $\bGS_\Ical$ is isomorphic to $\bm\Gamma^+$. Hence, the result will follow from applying \Cref{th:reduction_hardness}
after defining a $5$-ary pattern $\Psi$ of internal references to $\Ical$ w.r.t. $\Dcal$.
Given a structure $\I$ satisfying $\I\rightarrow \bm{
\Gamma}^+$, we find a homomorphism $v\mapsto (m_v, n_v, o_v)$ from $\I$ to  $\bm{
\Gamma}^+$ in polynomial time. Moreover, we can assume that $\max_v \log_2(m_v n_v o_v) \leq |I|$.
Observe that the map 
$v\mapsto (m_v, n_v, o_v, 
 1- m_v - n_v - o_v)$ is a homomorphism from $\I$ to $\bGS_\Ical$. We define
\begin{align*}
\Psi_\I\equiv 
\bigexists_{v\in I}
x_v 
\bigexists_{R\in \Sigma_{\Gamma^+
}, r\in R^I} x_r 
\left(
\bigwedge_{v\in I}
\phi_{m_v,n_v,o_v}(x_v) \right) \bigwedge 
\left(
\bigwedge_{R\in \Sigma_\Gamma, r\in R^I}
\phi_R(x_r)  \bigwedge_{i\in [\arty(R)]}
x_r^{\Pi^\Ical_{R,i}} = x_{r(i)}
\right).
\end{align*}
Here, the minor conditions $\phi_{m,n,o}(x)$, and $\phi_R(x)$ for $R\in \Sigma_{\Gamma^+}$, are the internal references defined in the previous claims.
To see that $\Psi_\I$ can be computed in polynomial time, observe that the formula $\phi_{m, n, o}(x_v)$ can be constructed inductively in time $O(\log_2(mno))$, and
the maximum of $\log_2(m_v,n_v, o_v)$  is at most $|I|$ for $v\in I$. Now we only need to show that
$\Mscr_\AIP \models \Psi_\I$ in order to prove that
$(\Ical, \Dcal, \Psi)$ is a valid pattern.
We give existential witnesses for the variables in $\Psi_\I$ to show that $\Mscr_\AIP\models \Psi_\I$.
For each $v\in I$, we choose \[
f_v=(m_v, n_v, o_v, 1 - m_v - n_v - o_v)
\] 
as the existential witness for $x_v$. 
By Claim 6, we know that $f_v$ satisfies $\phi_{m_v, n_v, o_v}(x)$ on $\Mscr_\AIP$. The fact that the map $v\mapsto f_v$ is a homomorphism from $\I$ to $\bGS_\Ical$ means that for each $R\in \Sigma_{\Gamma^+}, r\in R^I$ there is some $f_r\in R^\Ical$ satisfying that 
\[f_r^{\Pi_{R,i}^\Ical}= f_{r(i)} \quad \quad \text{for each $i\in \arty(R)$}. \] 
The element $f_r$ must also satisfy $\phi_R(x)$ on $\Mscr_\AIP$, because $\phi_R(x)$ is an internal definition of $R$. Hence, $f_r$ is a valid existential witness for $x_r$. This shows that $\Mscr_\AIP\models \Psi_\I$, and completes the proof that $(\Ical, \Dcal, \Psi, 5)$ is a pattern over $\Mscr_\AIP$. \par
Now \Cref{th:AIP_main}-(1) follows from \Cref{th:reduction_hardness} together with \Cref{prop:3-grid}-(1). Observe that $3\geq \arty(U^\Ical)$ and $3\geq \arty(P)$ for all $P\in \Dcal$, so it is enough to consider templates of the form $(\bK_3^5, \B)$ in 
\Cref{th:AIP_main}-(1).
\end{proof}

\subsubsection{AIP: Interpreting the Grid}
\label{sec:grid_AIP}

The following interpretation $\Ical$
induces a global structure which is finitely equivalent to $\bm{\Gamma}$. 

\begin{interpret} The $\Sigma_\Gamma$-interpretation $\Ical$ over $\Mscr_{\AIP}$ is given by
\begin{flalign*}
U^\Ical=  \Bigl\{
\left( 2^m3^n, y\right) \in  \Mscr_\AIP(2) \, \Big\vert \Bigr.
\Bigl.
\, m, n \text{ non-negative integers }\Bigr\} &&
\end{flalign*}
\begin{flalign*}
   O^\Ical= \Bigl\{(1,0)\Bigr\}, \quad \text{and} \quad \Pi^\Ical_{O,1}=\id, &&
\end{flalign*}
\begin{flalign*}
E_1^\Ical= \Bigl\{ 
(m, m, n, o) \in \Mscr_\AIP(4)
\, \Big\vert \, m+n=1 \Bigr\}, \quad \text{and} \quad 
\Pi^\Ical_{E_1,j} =
\begin{cases}
   (1,2,2,2) \text{ for $j=1$, and } \\
   (1,1,2,2) \text{ for  $j=2$.} 
\end{cases}
&&
\end{flalign*}

\begin{flalign*}
E_2^\Ical = \Bigr\{
(m, m, m, n, o) \in \Mscr_\AIP(5) \, \Big\vert
\, m + n= 1\Bigl\},  \quad \text{and} \quad 
\Pi^\Ical_{E_2,j}=
\begin{cases}
    (1,2,2,2,2) \text{ for $j=1$, and } \\
    (1,1,1,2,2) \text{ for $j=2$.}
\end{cases} &&
\end{flalign*}

\end{interpret}

\noindent 
Observe the grid structure 
$\bm{\Gamma}$ is isomorphic to 
$\bGS_\Ical$ via the bijection 
\[
(m,n) \mapsto (2^{m-1}3^{n-1}, 1 - 2^{m-1}3^{n-1} ). \]
Moreover, this map is computable under the plain encoding for $\bm{\Gamma}$ and $\Mscr_\AIP$. Next, we define a description $\Dcal\subseteq 2^{\Mscr_\AIP}$
so that
$\Ical$ is internal w.r.t. $\Dcal$ at arity $5$. 
\begin{desc}
The description $\Dcal\subseteq 2^{
\Mscr_{\AIP}}$ consists of the predicates $U^\Ical, O^\Ical$ defined in the previous interpretation.
\end{desc}

We have defined the interpretation $\Ical$ and description $\Dcal$ with the goal to obtain templates $(\A, \B)$ where $\A$ is Boolean after applying \Cref{cor:minion_to_template}. By \Cref{le:rank_manifold_minion}, this requires $\arty(U^\Ical)\leq 2$ and $\arty(D)\leq 2$ for all predicates $D\in \Dcal$. This is the reason behind the more awkward encoding of $\Gamma$ compared to the previous subsection. An added difficulty is that because $\Dcal$ cannot contain ternary predicates, we cannot speak about equality between coordinates in a straight-forward way. In the previous subsection, $\Dcal$
contained the predicate $D_<= \{ (m_1, m_2, m_3)\in \Mscr_\AIP \, \vert \, m_1 < m_2 \}$, so given an element $f=(n_1, n_2, n_3, n_4)\in \Mscr_\AIP(4)$ we could tell whether its first two coordinates were equal by checking whether $\qclass{f}=\qclass{f}^{(2,1,3,4)}$ in the quotient $\Mscr_{\AIP}/\Dcal$
\footnote{We do this when proving Claim 5 in \Cref{sec:3-dim_grid_AIP}.}. In this section we cannot compare coordinates in this direct way. But instead, we can still verify that the first two coordinates of $f$ are equal as long as we have, for example, that $n_1+n_3=1$ and $n_2 + n_3=1$. This motivates the definition of the predicates $E_1^\Ical$ and $E_2^\Ical$ in this subsection. \par

Let us show that $\Ical$ is indeed internal at arity $5$ (w.r.t. $\Dcal$). This is a consequence of the following claims. \\~\\
\noindent
\textbf{Claim 1: The predicates $U^\Ical$, 
$O^\Ical$, $E_1^\Ical$, and $E_2^\Ical$ are all $\Dcal$-stable.} This follows directly from the definitions.
\\~\\ \noindent
\textbf{Claim 2: The predicate $O^\Ical$ is internal at arity $5$.} Indeed, the minor condition $\phi_O(x)\equiv x = x^{(1,1)}$ is an internal definition of $O^\Ical$. \\~\\
\noindent
 \textbf{Claim 3: The predicates $E_1^\Ical$ and $E_2^\Ical$ are internal at arity $5$.} 
 Consider the following internal definitions for $E_1^\Ical$ and $E_2^\Ical$:
\[
\phi_{E_1}(x) \equiv \phi_{O}(x^{(1,2,1,2)}) \wedge
\phi_{O}(x^{(2,1,1,2)})
, 
\]
\begin{align*} 
\phi_{E_2}(x) \equiv \phi_{O}(x^{(1,2,2,1,2)}) \wedge  
\phi_{O}(x^{(2,1,2,1,2)}) \wedge \phi_{O}(x^{(2,2,1,1,2)}).   
\end{align*}
~\\
\noindent
\textbf{Claim 4: The following implications hold.}
\begin{align*}
&
f^{(1,2,2,2)}\in U^\Ical \implies
f^{(1,1,2,2)} \in U^\Ical \textbf{ for all } f\in E_1^\Ical \\
&
f^{(1,2,2,2,2)}\in U^\Ical \implies
f^{(1,1,1,2,2)} \in U^\Ical 
\textbf{ for all } f\in E_2^\Ical.
\end{align*}
We prove the statement for $f\in E_1^\Ical$. The case where $f\in E_2^\Ical$ follows analogously. Let $(m,n)= f^{(1,2,2,2)}$. By the definition of $E_1^\Ical$, it must hold that
$f^{(1,1,2,2)}= (2m, n-m)$. Now, observe that $(m,n)\in U^\Ical$ implies that $(2m, m - n)\in U^\Ical$ as well, following the definition of $U^\Ical$. \\~\\
\noindent
\textbf{Claim 5: } \textbf{Let $\phi(x)$ be an internal reference to $U^\Ical$ then the following formulas are also internal references to $U^\Ical$:}
\begin{align*}
\phi_1(x) =  \exists y \exists z \Bigl( 
\phi(y) \wedge \phi_{E_1}(z) \wedge 
y = z^{(1,2,2,2)} \wedge x = z^{(1,1,2,2)}
\Bigr),
\end{align*}
\textbf{ and }
\begin{align*}
\phi_2(x) = \exists y \exists z \Bigl( 
\phi(y) \wedge \phi_{E_2}(z) \wedge
y = z^{(1,2,2,2,2)} \wedge x = z^{(1,1,1,2,2)}
\Bigr).
\end{align*}

Indeed, let us argue the statement for $\phi_1$. The case of $\phi_2$ follows similarly. Suppose that 
$\Mscr_\AIP/\Dcal \models \phi_1(\qclass{f_x})$ for some $f\in \Qcal(2)$, and let $\qclass{f_y}$ and $\qclass{f_z}$ be existential witnesses for $y$ and $z$ respectively. In particular, 
\[
\Mscr_\AIP/\Dcal \models \phi(\qclass{f_y}), \quad \text{ and }
\Mscr_\AIP/\Dcal \models \phi_{E_1}(\qclass{f_z}).\]
The fact that $U^\Ical$, and $E_1^\Ical$ are $\Dcal$-stable and $\phi(x)$, $\phi_{E_1}(x)$ are internal references to those predicates implies that $f_y\in U^\Ical$ and $f_z\in E_1^\Ical$. It must also hold that $f_y \sim_\Dcal f_z^{(1,2,2,2)}$, so we can conclude that $f_z^{(1,2,2,2)}\in U^\Ical$. By Claim 5 this implies that $f_z^{(1,1,2,2)}\in U^\Ical$ as well.
Using that $f_z^{(1,1,2,2)}\sim_\Dcal f_x$ we finally obtain that $f_x\in U^\Ical$, proving that $\phi_1(x)$ is an internal reference to $U^\Ical$. \\~\\
\noindent
\textbf{Claim 7: The predicate $U^\Ical$ is internal at arity $5$.} We use~\Cref{le:internal_criterion} to prove the claim. 
We define an internal reference  $\phi_{m,n}(x)$ to $U^\Ical$ inductively for each $m,n\in \NN$. Additionally, we keep the invariant that 
$\Mscr_\AIP\models \phi_{m,n}(f)$, 
for each $m,n \in \NN$
where $f= ({2^{m-1}2^{n-1}}, 1-{2^{m-1}3^{n-1}})$. 
We define $\phi_{1,1}$ as the minor condition $\phi_O$. Now, given $m>1$, we define

\begin{align*}
\phi_{m,1}(x) \equiv
\exists y \exists z \Bigl( 
\phi_{m-1,1}(y) \wedge \phi_{E_1}(z) \wedge
y = z^{(1,2,2,2)} \wedge x = z^{(1,1,2,2)}
\Bigr).
\end{align*}

By the previous Claim, the fact that 
$\phi_{m-1,1}$ is an internal reference to $U^\Ical$, means that so is $\phi_{m,1}$ as well. Moreover, the fact that $({2^{m-2}}, 1-{2^{m-2}})$ satisfies $\phi_{m-1,1}(x)$ on $\Mscr_\AIP$, means that
$({2^{m-1}}, 1-{2^{m-1}})$ satisfies $\phi_{m,1}(x)$ on $\Mscr_\AIP$ by taking
$({2^{m-2}}, 1-{2^{m-2}})$ as an existential witness for $y$
and 
\[
\left({2^{m-2}}, 
{2^{m-2}},
1 - {2^{m-2}},
 -{2^{m-2}}\right)
\]
as an existential witness for $z$. Arguing in a similar way, if $n>1$, we define 
\begin{align*}
\phi_{m,n}(x) = \exists y \exists z \Bigl( 
\phi_{m,n-1}(y) \wedge \phi_{E_2}(z) \wedge 
y = z^{(1,2,2,2,2)} \wedge x = z^{(1,1,1,2,2)}
\Bigr). 
\end{align*}
Again, using the previous claim we obtain that $\phi_{m,n}$ is an internal reference to $U^\Ical$. One can also see that
$({2^{m-1}3^{n-1}}, 1-{2^{m-1}3^{n-1}})$
satisfies $\phi_{m,n}(x)$ on $\Mscr_\AIP$. This shows the claim.

\begin{proof}[Proof of items (2), (3), (4) of \Cref{th:AIP_main}]
In this section we have shown that
$\bGS_\Ical$ is isomorphic to $\bm \Gamma$ through a computable homomorphism, and that $\Ical$ is
internal at arity $5$ w.r.t. $\Dcal$. Additionally, $\arty(U^\Ical)=2$, and $\arty(Q)=2$ for all $Q\in \Dcal$. Then items (3), (4) of \Cref{th:AIP_main} follow from
 \Cref{prop:grid} together via \Cref{th:reduction_undecidability} and \Cref{th:reduction_non_computability}. \par
Finally, to prove \Cref{th:AIP_main}-(2) we define a $5$-ary 
pattern $\Psi$ of internal references to $\Ical$ w.r.t. $\Dcal$.
We map each $\I\in \Hom(\bm \Gamma, \cdot)$ to a minor condition $\Psi_\I$ defined as follows. First, we find in polynomial time a homomorphism $v\mapsto (m_v, n_v)$ from $\I$ to $\bm{\Gamma}$. Moreover,
this can be done in such a way that $\max_{v\in I} m_v + n_v \leq |I| + 1$.
Then
\begin{align*}
\Psi_\I \equiv 
\bigexists_{v\in I} x_v
\bigexists_{R\in \Sigma_\Gamma, r\in R^{I}} x_r  \left(
\bigwedge_{v\in I}
\phi_{m_v,n_v}(x_v) \right)  \bigwedge
\left(
\bigwedge_{R\in \Sigma_\Gamma, r\in R^I}
\phi_R(x_r)  \bigwedge_{i\in [\arty(R)]}
x_r^{\Pi^\Ical_{R,i}} = x_{r(i)}
\right).
\end{align*}
Here, for each $R\in \Sigma_R$, $\phi_R(x)$ is the internal definition of $R^\Ical$ given in the previous claims, and $\phi_{m,n}(x)$ is the internal reference to $U^\Ical$ defined in Claim 7. To see that $\Psi_\I$ can be constructed in polynomial time, observe that $\phi_{m,n}$ takes $O(n+m)$ time to construct inductively, and we have that $\max_{v\in I} n_v + m_v \leq |I| +1$.
To prove that $(\Ical, \Dcal, \Psi, 5)$ is a valid pattern we need to show that $\Psi_\I$ is satisfiable over $\Mscr_\AIP$. Recall that $\Mscr_\AIP \models \phi_{m,n}(f_{m,n})$, where $f_{m,n}=({2^{m-1} 3^{n-1}}, 1 - {2^{m-1} 3^{n-1}})$, and observe that the map $v \mapsto f_{m_v, n_v}$ is a homomorphism from $\I$ to $\bGS_\Ical$. Using this fact, $\Mscr_\AIP \models \Psi_\I$ can be proven following the same reasoning as in the proof of \Cref{th:AIP_main}-(1), in \Cref{sec:3-dim_grid_AIP}.
\end{proof}

\subsection{The BLP Algorithm}
\label{sec:BLP}

We prove \Cref{th:BLP_main} in this section. 
To prove item (1) we give an interpretation of the super-grid $\bm{\Gamma}^+$ over $\Mscr_\BLP$, shown in \Cref{sec:3-dim_grid_BLP}. 
Items (2),(3), and (4) are proven similarly, by showing in  \Cref{sec:grid_BLP} a suitable interpretation of the grid $\bm{\Gamma}$ over $\Mscr_\BLP$. 

\subsubsection{BLP: Interpreting the Super-Grid}
\label{sec:3-dim_grid_BLP}

The following interpretation $\Ical$ induces a global structure which is finitely equivalent to $\bm{\Gamma}^+$.

\begin{interpret} The $\Sigma_{\Gamma^+}$-interpretation $\Ical$ over $\Mscr_{\BLP}$ is given by
  \begin{flalign*}
U^\Ical =
\Bigl\{ (m s, n s, o s, s, t) \in \Mscr_\BLP(5) \, \Big\vert 
\, s= \frac{1}{2^j} \text{ for some } j\in \NN, \, 
\text{and } m,n,o\in \NN
\Bigr\}, &&
  \end{flalign*}
\begin{flalign*}
O^\Ical = \Bigl\{ (s,s,s,s,t) \in \Mscr_\BLP(5) \Big\vert 
s= 2^{-i} \text{ for some } i\in\NN \Bigr\}, \quad \text{and} \quad \Pi^\Ical_{O,1} = \id, &&
\end{flalign*}
\begin{flalign*}
E_i = \Bigl\{
(s_1, s_2, s_3, s, s, t) \in \Mscr_\BLP(6) \Bigr\} \quad \text{and} \quad
\Pi^\Ical_{E_i, j} = \begin{cases}
    (1,2,3,4,5,5) \text{ if $j=1$, } \\
    (1,2,3,i,4,5) \text{ if $j=2$ }
\end{cases} \text{ for all $i\in [3]$,} &&
\end{flalign*} 

\begin{flalign*}
& \Ebb_i^\Ical =
\{
(s_1, s_2, s_3, s_i, s, t) \in \Mscr_\BLP(6) \}  \quad \text{and} \quad
\Pi^\Ical_{\Ebb_i,j} =
 \begin{cases}
     (1,2,3,4,5,5) \text{ if $j=1$,}\\
     (1,2,3,i,4,5) \text{ if $j=2$}
 \end{cases}\text{ for each $i\in [3]$. }   &&
\end{flalign*}
\end{interpret}

\noindent
Here it is worth pointing out the main difference with respect to the interpretation of $\bm \Gamma^+$ over $\Mscr_\AIP$ given in \Cref{sec:3-dim_grid_AIP}. Addition and, by extension, multiplication times two are easy to perform in the minion $\Mscr_{\BLP}$, but we are missing a natural notion of ``unit increment''. To remediate this, we represent a triple $(m,n,o)\in \NN^3$ in the $3$-dimensional grid with elements of the form $(ms, ns, os, s, t)$ in $\Mscr_{\BLP}$, where the fourth coordinate represents the increment that we have chosen in our encoding. Of course, for a given increment $s$, the resulting grid in $\Mscr_{\BLP}$ is finite, but we consider the disjoint union of these grids when $s$ ranges over all values of the form $2^{-\ell}$ for $\ell\in \NN$. The resulting global structure is not isomorphic to $\bm \Gamma^+$, but it is finitely equivalent to it, as indicated below. \\~\\

\noindent
\textbf{Claim 1: The structure $\bGS=\bGS_\Ical$ induced by $\Ical$ is finitely equivalent to $\bm{\Gamma}^+$.}
    Let us describe  $\bGS$. 
    For each integer $i\geq 2$ we define $\bGS$ as the connected component of $\bGS_i$ containing the element 
    \[
\left(\frac{1}{2^i},\frac{1}{2^i},
\frac{1}{2^i},\frac{1}{2^i}
1 - \frac{1}{2^{i-2}}\right) \in U_\Ical.
    \]
    Observe that $\bGS$ is the disjoint union $\bigsqcup_{i\geq 2} \bGS_i$.
Given $i \geq 2$, we define $\bm{\Gamma}^+_i$
as the substructure of 
$\bm{\Gamma}^+$ induced on the set of elements $(m,n,o)\in \NN^3$ with $m+n+o+1\leq 2^i$. The structure $\bm{\Gamma}_i^+$ is isomorphic to $\bGS_i$
through the bijection
\[
(m,n,o) \mapsto \left(
\frac{m}{2^i}, \frac{n}{2^i},
\frac{o}{2^i}, \frac{1}{2^i},
1- \frac{m+n+o+1}{2^i}
\right).
\]
Finally, observe that $\bm{\Gamma}^+$ is finitely equivalent to the disjoint union $\bigsqcup_{i>2} 
\bm{\Gamma}_i^+$. Indeed, if $\I\rightarrow \bm{\Gamma}^+$,
for some finite $\I$, then it must be that $\I\rightarrow \bm{\Gamma}^+_i$ for some $i>2$, and we also have that $\bm{\Gamma}^+_i \rightarrow \bm{\Gamma}^+$ for all $i$.  This proves the claim.  \par

Now let us define a description $\Dcal\subseteq 2^{\Mscr_\BLP}$ so that $\Ical$ is internal at arity $6$ w.r.t. $\Dcal$. 
\begin{desc} 
The description $\Dcal\subseteq 2^{\Mscr_\BLP}$ consists of the predicate $U_\Ical$ defined in the previous interpretation, and the predicates 
\begin{flalign*}
& D_p=\Bigl\{ (\frac{1}{2^i}, 1 - \frac{1}{2^i})
\, \vert  \, i\geq 0 \Bigr\}, \text{ and }  && \\
& D_< = \Bigr\{ (s_1, s_2, s_3)\in \Mscr_\BLP \, \vert \, s_1 < s_2 \Bigl\}. &&
\end{flalign*}
\end{desc}

\noindent The following claims establish that $\Ical$ is internal at arity $6$ (w.r.t. $\Dcal$). \\~\\
\noindent
\textbf{Claim 2: The predicates $U^\Ical$, and $R^\Ical$
    for $R\in \Sigma_{\Gamma^+}$ are all $\Dcal$-stable.} This follows from the definitions. \\~\\
    \noindent
    \textbf{Claim 3: The predicates $E_i^\Ical$ are internal at arity $6$.} The following is an internal definition of $E_i^\Ical$:
    \[
    \phi_{E_i}(x) \equiv
    x^{(3,3,3,1,2,3)}= 
    x^{(3,3,3,2,1,3)}.
    \] 
    ~\\
    \noindent
    \textbf{Claim 4: The predicates $\Ebb_i^\Ical$ are internal at arity $6$.} The following is an internal definition of $E_i^\Ical$:
    \[
    \phi_{\Ebb_i}(x) \equiv
    x^{\sigma}= 
    x^{\tau},
    \]
    where $\sigma: [6]\rightarrow [3]$ maps $i$ to $1$, 
    $4$ to $2$ and the other elements to $3$, and $\tau: [6] \rightarrow [3]$ maps $i$ to $2$, $4$ to $1$ and
    the other elements to $3$. \\~\\
    \noindent
    \textbf{ Claim 5: If $\phi(x)$ is an internal reference to $D_p$, then the following formula is also an internal reference to $D_p$:}
    \begin{align*}
    \phi^\prime(x) = 
    \exists y \exists^3 z 
    \Bigl(
    \phi(y)  \wedge 
    z = z^{(2,1,3)} \wedge
    y = z^{(1,1,2)}
    \wedge 
    x = z^{(1,2,2)}
    \Bigr) .
    \end{align*}
    \textbf{Moreover, if $(1/2^i, - 1/2^i)$ satisfies $\phi(x)$ (over $\Mscr_\BLP$), then
    $(1/2^{i+1}, 1 - 1/2^{i+1})$ 
    satisfies $\phi^\prime(x)$.} Let us begin with the second part of the statement. To see that $(1/2^{i+1}, 1- 1/2^{i+1})$ satisfies $\phi^\prime(x)$, observe that $(1/2^i, 1 - 1/2^i)$ and
    $(1/2^{i+1}, 1/2^{i+1}, 1- 1/2^i)$
    are valid existential witnesses for $y$ and $z$.\par
    Now  let us show that $\phi^\prime$ is an internal reference to $\Dcal_{1/2}$. Suppose $\Mscr_\BLP/\Dcal
    \models \phi^\prime(\qclass{f_x})$
    and $\qclass{f_y},\qclass{f_z}$
    are existential witnesses for $y$ and $z$. As $\phi$ is an internal reference to $D_p$, we have $f_y\in D_p$. Hence, the fact that
    $f_y \sim_\Dcal f_z^{(1,1,2)}$
    means that $f_z^{(1,1,2)}\in D_p$,
    so $f_z(1) + f_z(2) = 1/2^i$ for some $i\geq 0$. 
    Also, observe that the fact that
    $f_z \sim_\Dcal f_z^{(2,1,3)}$
    implies $f_z(1) = f_z(2)$. Indeed, otherwise exactly one of $f_z$
    or $f_z^{(2,3,1)}$ would belong to 
    $D_<$. Hence $f_z(1)= 1/2^{i+1}$, and $f_z^{(1,2,2)}\in D_p$. Finally, because $f_x \sim_\Dcal f_z^{(1,2,2)}$, we must have $f_x\in D_p$, proving the claim. \\~\\
\noindent
\textbf{ Claim 6:  The predicate $D_p$ is internal at arity $6$.} 
We use \Cref{le:internal_criterion} to prove the claim. We define an internal reference $\phi_i(x)$ to $D_p$ inductively for each $i\geq 0$ in such a way that $(1/2^i, 1- 1/2^i)$ satisfies $\phi_i(x)$ on $\Mscr_\BLP$.
We define
\[
\phi_0(x) \equiv x = x^{(1,1)}.
\]
Given $i>0$, we define
\begin{align*}
   \phi_i(x) \equiv \exists y \exists^3 z
 \Bigl( 
    \phi_{i-1}(y)  \wedge 
    z = z^{(2,1,3)} 
    \wedge
    y = z^{(1,1,2)}
    \wedge 
    x = z^{(1,2,2)}
    \Bigr). 
\end{align*}

Now the previous claim proves that
$\phi_i(x)$ is an internal reference to $D_p$ and $(1/2^i,1- 1/2^i)$ satisfies it. This proves the statement. \\~\\
\noindent
\textbf{ Claim 7: The predicate $O^\Ical$ is internal at arity $6$.}
Let $i\geq 2$. Then, by last claim, the following formula is an internal reference to $O^\Ical$:
\begin{align*}
\phi_{O,i}(x) \equiv
\phi_i(x^{(1,2,2,2,2)}) &\wedge 
x^{(1,2,3,3,3)} = x^{(2,1,3,3,3)} \\ &\wedge
x^{(1,3,2,3,3)} = x^{(2,3,1,3,3)} \\
&\wedge
x^{(1,3,3,2,3)} = x^{(2,3,3,1,3)}.
\end{align*}
Additionally, $\phi_{O,i}(x)$ is satisfied on $\Mscr_\BLP$ by the element
\[
\left(
\frac{1}{2^i},
\frac{1}{2^i},
\frac{1}{2^i},
\frac{1}{2^i},
1- \frac{1}{2^{i-2}}
\right).
\]
Now the claim follows from \Cref{le:internal_criterion}. \\~\\
\noindent
\textbf{Claim 8: Let $i\in [3]$. Suppose $\phi(x)$ is an internal reference to $U^\Ical$. Then the following formula is also an internal reference to $U^\Ical$: }
\begin{align*}
\phi^\prime(x) \equiv
\exists y \exists^6 z
\Bigl( 
\phi(y)
\wedge \phi_{E_i}(z) 
\wedge y = z^{(1,2,3,4,5,5)}
\wedge x = z^{(1,2,3,4,i,5)}
\Bigr).
\end{align*}
\textbf{Moreover, if $f$ satisfies $\phi(x)$ (on $\Mscr_\BLP$), then the tuple $g$
satisfies $\phi^\prime(x)$, where $g$ is
defined by
$g(i)= f(i)+ f(4)$, 
$g(5)= f(5) - f(4)$ and $g(j)=f(j)$,
and $g(j) = f(j)$ for $j\neq i, 5$.
}
We begin with the second part of the statement. To see that $g$ satisfies $\phi^\prime(x)$, observe that $f$ and $f_z=(f(1), f(2), f(3), f(4), f(4), f(5) - f(4))$ are existential witnesses for $y$ and $z$. \par
Now let us show that $\phi^\prime$ is an internal reference to $U^\Ical$.
Suppose $\Mscr_\BLP/\Dcal \models \phi^\prime(\qclass{f_x})$ with
$\qclass{f_y}, \qclass{f_z}$ as existential witnesses for $y,z$. As $\phi$ is an internal reference to $U^\Ical$, it must hold that $f_y\in U^\Ical$. Moreover, $f_y\sim_\Dcal f_z^{(1,2,3,4,5,5)}$, so 
$f_z^{(1,2,3,4,5,5)}\in U^\Ical$, meaning that $f_z(4)= \frac{1}{2^i}$ for some $i\geq 2$, and $f_z(j)= \frac{m_j}{2^i}$ for some $m_j\in \NN$ for each $j\in [3]$. Now, the fact that
$f_z\in E_i^\Ical$
means that
$f_z(4)= f_z(5)$. Hence, we conclude that $f_z(5)= \frac{1}{2^i}$ as well. This implies that $f_z^{(1,2,3,4,i,5)}$ belongs to $U^\Ical$. Finally, using that $f_x\sim_\Dcal f_z^{(1,2,3,4,i,5)}$, we obtain $f_x\in U^\Ical$, as we wanted. \\~\\
\noindent 
\textbf{ Claim 9: Let $i\in [3]$. Suppose $\phi(x)$ is an internal reference to $U^\Ical$. Then the following formula is also an internal reference to $U^\Ical$: }
\begin{align*}
\phi^\prime(x) \equiv 
\exists y \exists^6 z
\Bigl( 
\phi(y) \wedge  \phi_{\Ebb_i}(z)
\wedge y = z^{(1,2,3,4,5,5)}
\wedge x = z^{(1,2,3,4,i,5)}
\Bigr)
\end{align*}
\textbf{Moreover, if $f$ satisfies $\phi(x)$ (on $\Mscr_\BLP$), then the tuple $g$
satisfies $\phi^\prime(x)$, where $g$ is
defined by
$g(i)= 2f(i)$, 
$g(5)= f(5) - f(i)$ and $g(j)=f(j)$,
and $g(j) = f(j)$ for $j\neq i, 5$.
} This can be shown analogously to the previous claim. \\~\\
\noindent
\textbf{Claim 10: The predicate $U^\Ical$ is internal at arity $6$.}
We prove the claim using~\Cref{le:internal_criterion}, as usual. We define an internal reference $\phi_{m_1,m_2,m_3,m_4}(x)$ to $U^\Ical$ inductively, following the lexicographical order, for each $(m_1,m_2,m_3,m_4)\in \NN^4$
such that $m_2+m_3+m_4+1 \leq 2^{m_1}$, in such a way that 
\[
\left(
\frac{m_2}{2^{m_1}},
\frac{m_3}{2^{m_1}},
\frac{m_4}{2^{m_1}},
\frac{1}{2^{m_1}},
1- \frac{m_2+m_3+m_4+1}{2^{m_1}}
\right)
\]
satisfies $\phi_{m_1,m_2,m_3,m_4}(x)$ on $\Mscr_\BLP$. For each $m\geq 2$ we define
\[
\phi_{m,1,1,1}(x)\equiv \phi_{O,m}(x),
\]
where $\phi_{O,m}(x)$ is the internal reference to $O^\Ical$ defined in 
Claim 7. Now suppose that $i>1$
is the maximum index such that $m_i>1$. We have two cases. Suppose that $m_i$ is odd. Then we let
$n_i= m_i-1$, $n_j= m_j$ for all $j\neq i$, and define $\phi_{m_1,m_2,m_3,m_4}(x)$
as 
\begin{align*}
\exists y \exists^6 z
\Bigl( & 
\phi_{n_1,n_2, n_3, n_4}(y)
\wedge  
 \phi_{E_i}(z) 
\wedge  y = z^{(1,2,3,4,5,5)}
\wedge x = z^{(1,2,3,4,i,5)}
\Bigr).
\end{align*}
Otherwise, suppose $m_i$ is even. Then
 Then we let
$n_i= m_i/2$, $n_j= m_j$ for all $j\neq i$, and define $\phi_{m_1,m_2,m_3,m_4}(x)$
as
\begin{align*}
\exists y \exists^6 z
\Bigl( 
\phi_{n_1,n_2, n_3, n_4}(y)
\wedge 
  \phi_{\Ebb_i}(z) 
\wedge y = z^{(1,2,3,4,5,5)}
\wedge x = z^{(1,2,3,4,i,5)}
\Bigr). 
\end{align*}
Now the statement follows from Claims 8 and 9.

\begin{proof}[Proof of item (1) of \Cref{th:AIP_main}.]
We define a $6$-ary pattern $\Psi$ of internal references to $\Ical$ w.r.t. $\Dcal$. Then the result will frollow from \Cref{th:reduction_hardness}. 
Given $\I\in \Hom(\bm \Gamma^+, \cdot)$, we first compute a homomorphism 
$v\mapsto (m_v, n_v, o_v)$ from $\I$ to
$\bm{\Gamma}^+$ in polynomial time in such a way that $\max_{v\in I} \log_2( m_v n_v o_v) \leq |I|$. Let $M= \max_{v\in I} m_v + n_v + o_v + 1 $, and let $j= \lceil \log_2 M \rceil$. Then $F$ is actually a homomorphism from $\I$ to $\bm{\Gamma}_j^+$, where $\bm{\Gamma}_j^+\subseteq \bm{\Gamma}^+$ is the substructure defined in Claim 1. Then the minor condition $\Psi_\I$ is defined as
\begin{align*}
\bigexists_{v\in I}
x_v
\bigexists_{
R \in \Sigma_{\Gamma^+},
r\in R^I} x_r
\left(
\bigwedge_{v\in I}
\phi_{j,m_v, n_v, o_v}(x_v)
\right) 
&
\bigwedge
\left(
\bigwedge_{r\in O^I}
\phi_{O,j}(x_r) \wedge
x_r = x_{r(1)}
\right) \\ 
&\bigwedge 
\left(
\bigwedge_{
R\in \Sigma_{\Gamma^+}, R\neq O,
r\in R^I}
\phi_{R}(x_r) \bigwedge_{i\in \arty(R)}
x_r^{\Pi^\Ical_{R,i}}
= x_{r(i)}
\right).
\end{align*}
To see that $\Psi_\I$ can be computed in polynomial time, observe that
$\phi_{j,m_v, n_v, o_v}$ can be constructed inductively in time $O(j + \log_2(m_v n_v o_v))$,
$j\leq 1 + \log_2 4|I|$, and 
$\max_{v\in I} \log_2(m_v n_v o_v) \leq |I|$.
In order to prove that $(\Ical,\Dcal, \Psi, 6)$ is a valid pattern, we only need to show that $\Mscr_\BLP \models \Psi_\I$. To do this, observe that the map
\[
v\mapsto
f_v = \left(
\frac{m_v}{2^j},
\frac{n_v}{2^j},
\frac{o_v}{2^j},
\frac{1}{2^j},
1- \frac{m_v + n_v + o_v +1}{2^j}
\right),
\]
is a homomorphism from $\I$ to the connected component of
\[\left(
\frac{1}{2^j},
\frac{1}{2^j},
\frac{1}{2^j},
\frac{1}{2^j},
1- \frac{1}{2^{j-2}}
\right) \]
in $\bGS_\Ical$ (recall the isomorphisms from Claim 1). Now $\Mscr\models \Psi_\I$ can be proven analogously to \Cref{th:AIP_main}-(1) in 
\Cref{sec:3-dim_grid_AIP}.
Observe that $5\geq \arty(U^\Ical)$ and $5\geq \arty(P)$ for all $P\in \Dcal$. Hence, we only need to consider templates of the form $(\bK_5^6, \B)$ to achieve the TFNP-hardness result. 
\end{proof}

\subsubsection{BLP: Interpreting the Grid}
\label{sec:grid_BLP}

The following interpretation $\Ical$ induces a global structure that is finitely equivalent to $\bm{\Gamma}$. 

\begin{interpret} The $\Sigma_\Gamma$-interpretation $\Ical$ over $\Mscr_\BLP$ is given by
\begin{flalign*}
&
U^\Ical= \Bigl\{ 
\left(\frac{1}{2^m3^n},  y\right)\in \Mscr_\BLP(2) \, \vert \, m, n \text{ non-negative integers}\Bigr\}, &&
\end{flalign*}
\begin{flalign*}
   O^\Ical= \Bigl\{(1,0)\Bigr\}, \quad \text{and} \quad \Pi^\Ical_{O,1}=\id, &&
\end{flalign*}
\begin{flalign*}
& E_1^\Ical= \Bigl\{ 
(x, x, y, y) \in \Mscr_\BLP(4)
\Bigl\}, \quad \text{and} \quad \Pi^\Ical_{E_1,i}=
\begin{cases}
    (1,1,2,2) \text{ for $i=1$, } \\
    (1,2,2,2) \text{ for $i=2$, }
\end{cases} &&
\end{flalign*}
\begin{flalign*}
& E_2^\Ical = \Bigl\{
(x,x,x, y, z) \in \Mscr_\BLP(4) \, \vert
\, x+ y= 1/3\Bigr\},  \quad \text{and} \quad 
\Pi^\Ical_{E_2,i}=
\begin{cases}
    (1,1,1,2,2) \text{ for $i=1$,} \\
    (1,2,2,2,2) \text{ for $i=2$.}
\end{cases} &&
\end{flalign*}
\end{interpret}

\noindent
This way, the grid structure 
$\bm{\Gamma}$ is isomorphic to 
$\bS$ via the bijection 
\[
(m,n) \mapsto \left(\frac{1}{2^{m-1}3^{n-1}}, 1 -\frac{1}{2^{m-1}3^{n-1}}\right). \] 
Moreover, this map is computable under the plain encoding for $\bm{\Gamma}$ and $\Mscr_\BLP$. Next, we define a description $\Dcal\subseteq 2^{\Mscr_\BLP}$
so that
$\Ical$ is internal w.r.t. $\Dcal$ at arity $5$. 
\begin{desc}
    The description $\Dcal\subseteq 2^{\Mscr_\BLP}$ consists of the binary predicates $U^\Ical$ and $\quad O^\Ical$ defined in the previous interpretation $\Ical$, and the new predicates 
    \begin{flalign*}
    & D_{1/2}=\Bigl \{ (1/2,1/2) \Bigl\}, \quad \text{ and } \quad D_{1/3}=\Bigl\{(1/3, 2/3) \Bigr\}. &&
    \end{flalign*}
\end{desc}
It is worth to compare the interpretation $\Ical$ and description $\Dcal$ with the ones from \Cref{sec:grid_AIP}, used to encode $\bm \Gamma$ on $\Mscr_\AIP$.
Both constructions are completely analogous, with the main difference that in that subsection we considered positive powers $2^m3^n$ and here we consider negative ones $2^{-m}3^{-n}$. Let us show that $\Ical$ is indeed internal at arity $5$ w.r.t. $\Dcal$. This is a consequence of the following claims.
\\~\\
\noindent
    \textbf{Claim 1: The predicates $U^\Ical$, 
$O^\Ical$, $E_1^\Ical$, and $E_2^\Ical$ are all $\Dcal$-stable.} This is a routine check. \\~\\
\noindent
    \textbf{Claim 2: The predicate $O^\Ical$ is internal at arity $5$ (w.r.t. $\Dcal$).}
    Indeed, the minor condition $\phi_O(x)\equiv x = x^{(1,1)}$ is an internal definition of $O^\Ical$. \\~\\
    \noindent
    \textbf{Claim 3: The predicates $D_{1/2}$ and $D_{1/3}$ are internal at arity $5$.}
    The following are internal definitions for $D_{1/2}$ and $D_{1/3}$ w.r.t. $\Dcal$:
\begin{align*}
& \phi_{1/2}(x) \equiv x = x^{(2,1)}, \\
& \phi_{1/3}(x) \equiv 
\exists^3y \Bigl(
y
= y^{(3,1,2)} \wedge y = y^{(2,3,1)} \wedge x = y^{(1,2,2)} \Bigr). 
\end{align*} ~\\
\noindent
\textbf{Claim 4: The predicates $E_1^\Ical$ and $E_2^\Ical$ are internal at arity $5$.}
 Consider the following internal definitions for $E_1^\Ical$ and $E_2^\Ical$:
\begin{align*}
& \phi_{E_1}(x) \equiv \phi_{1/2}(x^{(1,2,1,2)}) \wedge
\phi_{1/2}(x^{(2,1,1,2)}), \quad \text{and} \\
&
\phi_{E_2}(x) \equiv \phi_{1/3}(x^{(1,2,2,1,2)}) \wedge 
\phi_{1/3}(x^{(2,1,2,1,2)})
\wedge \phi_{1/3}(x^{(2,2,1,1,2)}).
\end{align*}
~\\
\noindent
\textbf{ Claim 5: The following implications hold. }
\begin{align*}
&
f^{(1,1,2,2)}\in U^\Ical \implies
f^{(1,2,2,2)} \in U^\Ical 
\textbf{ for all } f\in E_1^\Ical,  \quad \textbf{and} \\
&
f^{(1,1,1,2,2)}\in U^\Ical \implies
f^{(1,2,2,2,2)} \in U^\Ical
\textbf{ for all } f\in E_2^\Ical.
\end{align*}
We prove the statement for $f\in E_1^\Ical$. The case where $f\in E_2^\Ical$ follows analogously. Let $(m,n)= f^{(1,2,2,2)}$. By the definition of $E_1^\Ical$, it must hold that
$f^{(1,1,2,2)}= (2m, n - m)$. Now, observe that $(m, n)\in U^\Ical$ implies that $(2m, n - m)\in U^\Ical$ as well, following the definition of $U^\Ical$. \\~\\
\noindent
\textbf{ Claim 6: Let $\phi(x)$ be an internal reference to $U^\Ical$ then the following formulas are also internal references to $U^\Ical$:}
\begin{align*}
&
\phi_1(x) =  \exists y \exists z \Bigl( 
\phi(y) \wedge \phi_{E_1}(z) \wedge \Bigl.
y = z^{(1,1,2,2)} \wedge x = z^{(1,2,2,2)}
\Bigr), \quad \textbf{and}  \\
&
\phi_2(x) = \exists y \exists z \Bigl( 
\phi(y) \wedge \phi_{E_2}(z) \wedge
y = z^{(1,1,1,2,2)} \wedge x = z^{(1,2,2,2,2)}
\Bigr).
\end{align*}
Indeed, let us argue the statement for $\phi_1$. The case of $\phi_2$ follows similarly. Suppose that 
$\Mscr_\BLP/\Dcal \models \phi_1(\qclass{f_x})$ for some $f\in \Qcal(2)$, and let $\qclass{f_y}$ and $\qclass{f_z}$ be existential witnesses for $y$ and $z$ respectively. In particular, 
\[
\Mscr_\BLP/\Dcal \models \phi(\qclass{f_y}), \quad \text{ and }
\Mscr_\BLP/\Dcal \models \phi_{E_1}(\qclass{f_z}).\]
The fact that $U^\Ical$, and $E_1^\Ical$ are $\Dcal$-stable and $\phi(x)$, $\phi_{E_1}(x)$ are internal references to those predicates implies that $f_y\in U^\Ical$ and $f_z\in E_1^\Ical$. It must also hold that $f_y \sim_\Dcal f_z^{(1,1,2,2)}$, so we can conclude that $f_z^{(1,1,2,2)}\in U^\Ical$. By Claim 5 this implies that $f_z^{(1,2,2,2)}\in U^\Ical$ as well.
Using that $f_z^{(1,2,2,2)}\sim_\Dcal f_x$ we finally obtain that $f_x\in U^\Ical$, proving that $\phi_1(x)$ is an internal reference to $U^\Ical$. \\~\\
\noindent 
\textbf{Claim 7: The predicate $U^\Ical$ is internal at arity $5$.} We use~\Cref{le:internal_criterion} to prove the claim. 
We define an internal reference  $\phi_{m,n}(x)$ to $U^\Ical$ inductively for each $m,n\in \NN$. Additionally, we keep the invariant that 
$\Mscr_\BLP\models \phi_{m,n}(f)$, 
for each $m,n \in \NN$
where $f= (\frac{1}{2^{m-1}2^{n-1}}, 1-\frac{1}{2^{m-1}3^{n-1}})$. 
We define $\phi_{1,1}$ as the minor condition $\phi_O$. Now, given $m>1$, we define as
\begin{align*}
\phi_{m,1}(x) \equiv
\exists y \exists z \Bigl( 
\phi_{m-1,1}(y) \wedge \phi_{E_1}(z) \wedge 
y = z^{(1,1,2,2)} \wedge x = z^{(1,2,2,2)}
\Bigr).   
\end{align*}
By the previous Claim, the fact that 
$\phi_{m-1,1}$ is an internal reference to $U^\Ical$, means that so is $\phi_{m,1}$ as well. Moreover, the fact that $(\frac{1}{2^{m-2}}, 1-\frac{1}{2^{m-2}})$ satisfies $\phi_{m-1,1}(x)$ on $\Mscr_\BLP$, means that
$(\frac{1}{2^{m-1}}, 1-\frac{1}{2^{m-1}})$ satisfies $\phi_{m,1}(x)$ on $\Mscr_\BLP$ by taking
$(\frac{1}{2^{m-2}}, 1-\frac{1}{2^{m-2}})$ as an existential witness for $y$
and 
\[
\left(\frac{1}{2^{m-2}}, 
\frac{1}{2^{m-2}},
\frac{1}{2} - \frac{1}{2^{m-2}},
\frac{1}{2} -\frac{1}{2^{m-2}}\right)
\]
as an existential witness for $z$. Arguing in a similar way, if $n>1$, we define 
\begin{align*}
\phi_{m,n}(x) = \exists y \exists z \Bigl( &
\phi_{m,n-1}(y) \wedge \phi_{E_2}(z) \wedge \Bigr. \\ & \Bigl.
y = z^{(1,1,1,2,2)} \wedge x = z^{(1,2,2,2,2)}
\Bigr).   
\end{align*}
Again, using the previous claim we obtain that $\phi_{m,n}$ is an internal reference to $U^\Ical$. One can also see that
$(\frac{1}{2^{m-1}3^{n-1}}, 1-\frac{1}{2^{m-1}3^{n-1}})$
satisfies $\phi_{m,n}(x)$ on $\Mscr_\BLP$. This shows the claim.

\begin{proof}[Proof of 
items (2),(3), and (4) of 
\Cref{th:BLP_main}]
Observe that $\arty(U^\Ical)=2$, and $\arty(Q)=2$ for all $Q\in \Dcal$. 
Thus items (3) and (4) of \Cref{th:BLP_main} follow from
\Cref{prop:grid} together with \Cref{th:reduction_undecidability} and \Cref{th:reduction_non_computability} using the interpretation $\Ical$ and the description $\Dcal$ in this section. \par
Finally, to prove item (2) of \Cref{th:BLP_main} we define a $5$-ary pattern
$\Psi$ of internal references to $\Ical$ w.r.t. $\Dcal$.
Given  $\I\in \Hom(\bm \Gamma, \cdot)$ we construct the minor condition $\Psi_\I$ as follows. First, we find in polynomial time a homomorphism $v\mapsto (m_v, n_v)$ from $\I$ to $\bm{\Gamma}$ such that
$\max_{v\in I} m_v + n_v \leq |I| +1 $. Then
\begin{align*}
\Psi_\I \equiv 
\bigexists_{v\in I} x_v
\bigexists_{R\in \Sigma_\Gamma, r\in R^{I}} x_r  \left(
\bigwedge_{v\in I}
\phi_{m_v,n_v}(x_v) \right) \bigwedge 
\left(
\bigwedge_{R\in \Sigma_\Gamma, r\in R^I}
\phi_R(x_r)  \bigwedge_{i\in [\arty(R)]}
x_r^{\Pi^\Ical_{R,i}} = x_{r(i)}
\right).
\end{align*}
Here, for each $R\in \Sigma_R$, $\phi_R(x)$ is the internal definition of $R^\Ical$ given in the previous claims, and $\phi_{m,n}(x)$ is the internal reference to $U^\Ical$ defined in Claim 7. To see that $\Psi_\I$ can be computed in polynomial time, observe that $\phi_{m,n}(x)$ takes $O(m+n)$ time to be constructed inductively, and $\max_{v\in I} m_v + n_v \leq |I| +1$. 
The fact that $\Mscr_\BLP \models \Psi_\I$
can be proven similarly to \Cref{th:AIP_main}-(1) in \Cref{sec:3-dim_grid_AIP}, using that the map $v \mapsto (\frac{1}{2^{m-1} 3^{n-1}}, 1 - \frac{1}{2^{m-1} 3^{n-1}})$ is a homomorphism from $\Ical$ to $\bm{\Gamma}$.
\end{proof}

\subsection{The BLP + AIP Algorithm}
\label{sec:BLP+AIP}

We prove \Cref{th:BLP+AIP_main} in this section. This result follows from a suitable interpretation of the super-grid $\bm{\Gamma}^+$ over $\Mscr_{\BLP+\AIP}$, shown in \Cref{sec:3-grid_BLP+AIP}.

\subsubsection{BLP + AIP: Interpreting the Super-Grid}
\label{sec:3-grid_BLP+AIP}

Let $\Mscr=\Mscr_{\BLP+\AIP}$ for the rest of this section. Recall that elements in $\Mscr(n)$ consists of pairs $(f,g)$ where $f\in [0,1]^n$ and $g\in \ZZ^n$. We consider the lexicographical order on pairs $(s,m)\in [0,1]\times \ZZ$.
We perform arithmetic on tuples coordinate-wise. We write $m \Div n$ for $n,m\in \ZZ$
to denote that there exists an integer $o\in \ZZ$ such that $o n = m$, and 
$m\nDiv n$ for the negation of this statement \footnote{The more standard notation $n \mid m$ would quickly lead to readability issues.}. Similarly, given $(s,m), (t,n)\in [0,1]\times \ZZ$,
we write $(s,m) \Div (t,n)$ if there is some integer $o\geq 0$ such that
$o(t,n)=(s,m)$.
\par

The constructions and arguments in this section are very similar to those in \Cref{sec:3-dim_grid_BLP}, so we will omit some details. The following interpretation $\Ical$ induces a global structure
that is finitely equivalent to $\bm{\Gamma}^+$. 

\begin{interpret} The $\Sigma_{\Gamma^+}$-interpretation $\Ical$ over $\Mscr_{\BLP+\AIP}$ is given by
\begin{flalign*}
& U^\Ical= 
\Bigl\{ 
(f,g)\in \Mscr(5) \, \Big\vert \, f(4) < 1/3,
 g(4) \nDiv 2, \text{ and } 
(f(i),g(i)) \Div (f(4), g(4)) \text{ for each $i\in [3]$} \,
\Bigr\} &&
\end{flalign*}
\begin{flalign*}
& O^\Ical = \Bigl\{ 
(f,g)\in \Mscr(5) \, \Big\vert
\, f(1)< 1/3, \,  
g(1) \nDiv 2, \text{ and }
 (f(i), g(i))=(f(4), g(4)) \, \text{
for all $i\in [3]$} \Bigr\}, \text{ and } &&  \\
& \Pi^\Ical_{O,1} =\id. &&
\end{flalign*}
\noindent
And, for all $i\in [3]$
\begin{flalign*}
& E_i^\Ical  =
\Bigl \{
(f,g)\in \Mscr(6) \, \Big\vert
\, (f(4), g(4))=(f(5), g(5))
\Bigr \}, \quad \text{ and } \quad \Pi^\Ical_{E_i, j} =
\begin{cases}
    (1,2,3,4,5,5) \text{ for $j=1$,}\\
    (1,2,3,i,4,5) \text{ for $j=2$,}
\end{cases} &&
\end{flalign*}
\begin{flalign*}
    & \Ebb_i^\Ical =
\Bigl\{
(f,g)\in \Mscr(6) \, \Big\vert 
\, (f(4), g(4))=(f(i), g(i))
\Bigr\},  \quad \text{and} \quad
 \Pi^\Ical_{\Ebb_i,j} =
 \begin{cases}
     (1,2,3,4,5,5) \text{ for $j=1$, }\\
     (1,2,3,i,4,5) \text{ for $j=2$.}
 \end{cases} &&
\end{flalign*}

\end{interpret}
~\\
\noindent
As in \Cref{sec:3-dim_grid_BLP}, we follow the idea of interpreting increasingly large grids and keeping track of the increments in each of these grids. Similarly to before, if we represent a pair $(f,g)\in U^\Ical$ as $(f,g)=((f(1),g(1)), \dots, (f(5),g(5)))$,
it holds that the first three coordinates of this vector (which are two-dimensional vectors themselves) are integer multiples of 
$(f(4), g(4))$, which takes the role of the unit increment. Again, the fact that we consider tuples $(f,g)$ where $f(4)$ can be arbitrarily small but positive means that the global structure given by $\Ical$ contains arbitrarily large grids. This yields the following claim. \\~\\
\noindent
    \textbf{ Claim 1: The structure $\bGS=\bGS_\Ical$ induced by $\Ical$ is finitely equivalent to $\bm{\Gamma}^+$.} Given $j\in \NN$, we define $\bm{\Gamma}^+_j$ as the substructure of $\bm{\Gamma}^+$ induced on the elements $(m,n,o)\in \NN^3$ satisfying $m+n+o< j$. Consider an element $(f,g)\in U^\Ical$. Let $j= \lceil \frac{1}{f(4)} \rceil$. Then
    $\bm{\Gamma}_j^+$ is isomorphic to
    the connected component of $(f,g)$ in $\bGS$ via the bijection $(m_1, m_2, m_3) \mapsto ( f^\prime, g^\prime)$, where
    $(f^\prime(i), g^\prime(i))= m_i(f(4), g(4))$
    for $i\in [3]$,
    $(f^\prime(4), g^\prime(4))=(f(4), g(4))$,
    and $(f^\prime(5), g^\prime(5))$ is defined so that both the elements in $f^\prime$
    and the elements in $g^\prime$
    add up to $1$. We have that $\bm{\Gamma}^+$
    is finitely equivalent to the disjoint union 
    $\bigsqcup_{j\in \NN} \bm{\Gamma}_j^+$, so this proves the result.
\\~\par
Let us define a description $\Dcal$ 
so that $\Ical$ is internal at arity $6$ w.r.t. $\Dcal$. 
\begin{desc}
The description $\Dcal \subseteq 2^{\Mscr_{\BLP +\AIP}}$ is given by the predicate
$U^\Ical$ defined in the previous interpretation, and
\begin{flalign*}
& D_<
=\Bigl\{ 
(f,g)\in \Mscr(3) \, \Big\vert
\,
(f(1), g(1)) < (f(2), g(2))
\Bigr \}, && \\
& D_{\nDiv 2}= 
\Bigl\{
(f,g) \in \Mscr(2) \, \Big\vert \, f(1) < 1/3, \,
 g(1) \nDiv 2 
\Bigr\}. &&
\end{flalign*}    
\end{desc}

Let us show that $\Ical$ is internal at arity $6$ (w.r.t. $\Dcal$). This is a consequence of the following claims. \\~\\
\noindent
\textbf{Claim 2: The predicates $U^\Ical, O^\Ical$, and $E_i^\Ical$, $\Ebb_i^\Ical$
for $i\in [3]$ are all $\Dcal$-stable.} \\~\\ \noindent
\textbf{Claim 3: The predicate $E_i^\Ical$ is internal at arity $6$ for each $i\in [3]$}.
    This is shown in a similar way to Claim 3 in \Cref{sec:3-dim_grid_BLP}. The following is an internal definition of $E_i^\Ical$:
    \[
    \phi_{E_i}(x) \equiv x^{(3,3,3,1,2,3)} = x^{(3,3,3,2,1,3)}.
    \] ~\\
\noindent
\textbf{Claim 4: The predicate $\Ebb_i^\Ical$ is internal at arity $6$ for each $i\in [3]$}.
    This is shown in a similar way to Claim 4 in \Cref{sec:3-dim_grid_BLP}. The following is an internal definition of $E_i^\Ical$:
    \[
    \phi_{E_i}(x) \equiv x^\sigma = x^\tau,
    \] 
    where $\sigma: [6]\rightarrow [3]$ 
    sends $i$ to $1$, $4$ to $2$ and the other elements to $3$, and 
    $\tau:[6] \rightarrow [3]$
    sends $i$ to $2$, $4$ to $1$ and the other elements to $3$. \\~\\
    \noindent
\textbf{Claim 5: The predicate $D_{\nDiv 2}$
    is internal at arity $6$.} Indeed, the following is an internal definition of $D_{\nDiv 2}$:
    \begin{align*}
    \phi_{\nDiv 2}(x) \equiv
    \exists^5 y \Big(
    y^{(1,2,3,3,3)} = y^{(2,1,3,3,3)} \wedge   y^{(1,3,2,3,3)} = y^{(2,3,1,3,3)} \wedge 
    y^{(3,3,3,1,2)} = y^{(3,3,3,2,1)} \wedge  x = y^{(1,2,2,2,2)}
    \Big).
    \end{align*}
    We need to show both that
    $(f,g)$ satisfy $\phi_{\nDiv 2}(x)$ on $\Mscr$
    for all $(f,g)\in D_{\nDiv 2}$ and that
    $\phi_{\nDiv 2}$ is an internal reference to
    $D_{\nDiv 2}$. Let us begin with the first statement. Let $(f,g)\in D_{\nDiv 2}$. Then let $(f_y, g_y)\in \Mscr(5)$ be defined by
    $(f_y(i),g_y(i)) = (f(1),g(1))$, for all $i\in [3]$,
    and $(f_y(i),g_y(i))=\frac{1}{2}(1-f(1), 1-g(1))$
    for $i=4,5$. The fact that $(f_y, g_y)\in \Mscr$
    follows from the definition of $D_{\nDiv 2}$. Observe that $(f,g)$ satisfies $\phi_{\nDiv 2}(x)$ on $\Mscr$ with
    $(f_y, g_y)$ as an existential witness for $y$. \par
    Now let us show that $\phi_{\nDiv 2}$ is an internal reference to $D_{\nDiv 2}$. Suppose that
    $\Mscr/\Dcal \models \phi_{\nDiv 2}(\qclass{f,g})$
    with $\qclass{f_y,g_y}$ as an existential witness for $y$. We claim that
    $(f_y(i),g_y(i))= (f_y(1), g_y(1))$ for $i\in [3]$,
    and $(f_y(4),g_y(4))=(f_y(5),g_y(5))$. Indeed, 
    suppose that $(f_y(1),g_y(1))= (f_y(2), g_y(2))$
    Then exactly one of 
    $(f_y,g_y)^{(1,2,3,3,3)}$ or 
    $(f_y, g_y)^{(2,1,3,3,3)}$ would belong to $D_<$, contradicting $ (f_y,g_y)^{(1,2,3,3,3)} \sim_\Dcal (f_y, g_y)^{(2,1,3,3,3)}$. The other identities can be argued similarly. Hence, we have that
    $3g_y(1) + 2g_y(4) = 1$, where $g_y(1), g_y(4)\in \ZZ$. This forces  $g_y(1)\nDiv 2$ and $g_y(4) \nDiv 3$. In particular $g_y(1), g_y(4)\neq 0$, 
    so $f_y(1), f_y(4) > 0$ by the definition of $\Mscr$. We also have that $3 f_y(1) + 2f_y(4)$, so
    we obtain $f_y(1)< 1/3$ using that $f_y(4)>0$. All of this together implies that
    $(f_y,g_y)^{(1,2,2,2,2)}\in D_{\nDiv 2}$. Finally, the fact that $(f,g)\sim_\Dcal (f_y,g_y)^{(1,2,2,2,2)}$ shows that 
    $(f,g)\in D_{\nDiv 2}$, as we wanted to prove. \\~\\
    \noindent
    \textbf{Claim 6: The predicate $O^\Ical$ is internal at arity $6$.} The formula $\phi_O(x)$
    defined below is an internal definition of $O^\Ical$. 
    \begin{align*}
    & \exists y \Big( 
    \phi_{\nDiv 2}(y) \wedge
    y = x^{(1,2,2,2,2)} \wedge 
    x^{(1,2,3,3,3)} = x^{(2,1,3,3,3)}  \\ & \wedge
    x^{(1,3,2,3,3)} = x^{(2,3,1,3,3)} \wedge
    x^{(1,3,3,2,3)} = x^{(2,3,3,1,3)}
    \Big).
    \end{align*} 
    ~\\
    \noindent
    \textbf{Claim 7: Let $i\in [3]$. Suppose $\phi(x)$ is an internal reference to $U^\Ical$. Then the following formula is also an internal reference to $U^\Ical$: }
\begin{align*}
\phi^\prime(x) \equiv
\exists y \exists^6 z
\Bigl( 
\phi(y)
\wedge \phi_{E_i}(z) 
\wedge  y = z^{(1,2,3,4,5,5)}
\wedge x = z^{(1,2,3,4,i,5)}
\Bigr).
\end{align*}
\textbf{Moreover, if $(f,g)$ satisfies $\phi(x)$ (on $\Mscr_\BLP$), then $(f^\prime,g^\prime)$
satisfies $\phi^\prime(x)$, where 
\begin{itemize}
    \item $(f^\prime(i), g^\prime(i))= (f(i),g(i))+ (f(4),g(4))$
    \item $(f^\prime(5), g^\prime(5))= (f(5),g(5)) - (f(4),g(4))$,  and
    \item $(f^\prime(j), g^\prime(j))=(f(j), g(j))$,   for $j\neq i, 5$. 
\end{itemize}} 
This is shown exactly as Claim 8 in \Cref{sec:3-dim_grid_BLP}. \\~\\
\noindent
    \textbf{Claim 8: Let $i\in [3]$. Suppose $\phi(x)$ is an internal reference to $U^\Ical$. Then the following formula is also an internal reference to $U^\Ical$: }
\begin{align*}
\phi^\prime(x) \equiv
\exists y \exists^6 z
\Bigl( 
\phi(y)
\wedge \phi_{\Ebb_i}(z)
\wedge y = z^{(1,2,3,4,5,5)}
\wedge x = z^{(1,2,3,4,i,5)}
\Bigr).
\end{align*}
\textbf{Moreover, if $(f,g)$ satisfies $\phi(x)$ (on $\Mscr_\BLP$), then $(f^\prime,g^\prime)$
satisfies $\phi^\prime(x)$, where
\begin{flalign*}
& (f^\prime(i), g^\prime(i))= (2f(i),2g(i)), &\\&
(f^\prime(5), g^\prime(5))= (f(5),g(5)) - (f(i),g(i)), \text{ and } &\\&
 (f^\prime(j), g^\prime(j))=(f(j), g(j)), \text{ for $j\neq i, 5$.} &
\end{flalign*}
}
Again, this follows similarly to Claim 8 in \Cref{sec:3-dim_grid_BLP}.\\~\\
\noindent
\textbf{Claim 9: The predicate $U^\Ical$ is internal at arity $6$.} We prove the claim using~\Cref{le:internal_criterion}. 
For each 
$(m_1,m_2,m_3)\in \NN^3$ we define an internal reference $\phi_{m_1, m_2,m_3}(x)$ to $U^\Ical$ inductively following the lexicographical order on $\NN^3$ in such a way that $\Mscr \models \phi_{m_1, m_2, m_3}((f,g))$ for all
elements $(f,g) \in U^\Ical$ satisfying
that $(f(i),g(i))=m_i(f(4),g(4))$ for each $i\in [3]$.
We define
\[
\phi_{1,1,1}(x) \equiv 
\phi_O(x).
\]
Now, let $(m_1, m_2, m_3)\in \NN^3$ and
suppose $i\in [3]$ is the greatest index for which $m_i>1$. We have two cases. Suppose $m_i$ is odd. Then we let $n_j=m_j$ for $j\neq i$, and 
$n_i= m_i-1$, and define
\begin{align*}
\phi_{m_1,m_2, m_3}(x) \equiv 
\exists y \exists^6 z
\Bigg(  
\phi_{n_1,n_2,n_3}(y)
\wedge \phi_{E_i}(z) 
\wedge y = z^{(1,2,3,4,5,5)}
\wedge x = z^{(1,2,3,4,i,5)}
\Bigg).
\end{align*}
Otherwise, suppose $m_i$ is even. Then we let $n_j=m_j$ for $j\neq i$, and 
$n_i= m_i/2$, and define
\begin{align*}
\phi_{m_1,m_2, m_3}(x) \equiv 
\exists y \exists^6 z
\Bigg( 
\phi_{n_1,n_2,n_3}(y)
\wedge \phi_{\Ebb_i}(z) 
\wedge y = z^{(1,2,3,4,5,5)}
\wedge x = z^{(1,2,3,4,i,5)}
\Bigg).
\end{align*}
Now the statement follows from Claims 7 and 8.
\\~\\

\begin{proof}[Proof of \Cref{th:BLP+AIP_main}]
  Observe that $5\geq \arty(U^\Ical)$ and $5\geq \arty(P)$ for all $P\in \Dcal$. 
    \Cref{th:BLP+AIP_main}-(2) follows from \Cref{th:reduction_non_computability} and
    \Cref{prop:3-grid} using
    that $\bGS_\Ical$ is finitely equivalent to $\bm{\Gamma}^+$, and
    $\Ical$ is internal at arity $6$ w.r.t. $\Dcal$. \par
    We prove \Cref{th:BLP+AIP_main}-(1) using
    \Cref{th:reduction_hardness} and \Cref{prop:3-grid}. 
    To do this it is enough to define a $6$-ary pattern $\Psi$ 
    of internal references to $\Ical$ w.r.t. $\Dcal$.
    We map each finite structure $\I\in \Hom(\bm{\Gamma}^+, \cdot)$ to the minor condition
    $\Psi_\I$ defined as follows. First, we compute a homomorphism $v\mapsto (m_v, n_v, o_v)$
    from $\I$ to $\bm{\Gamma}^+$ in polynomial time, in such a way that
    $\max_{v\in I} \log_2(m_v n_v o_v) \leq |I|$. Then we set
    \begin{align*}
    \Psi_I \equiv 
    \bigexists_{v\in I}
    x_v \bigexists_{R\in \Sigma_{\Gamma^+}, r\in R^I}
    x_r 
    \left(
    \bigwedge_{v\in I}
    \phi_{m_v, n_v, o_v}(x_v)
    \right) \bigwedge
    \left(
    \bigwedge_{R\in \Sigma_{\Gamma^+}, r\in R^I} 
    \phi_R(x_r) 
    \bigwedge_{i\in \arty(R)}
    x_r^{\Pi^\Ical_{R,i}} = x_{r(i)}
    \right).
    \end{align*}
    
    To see that $\Psi_\I$ can be computed in polynomial time, observe that $\phi_{m,n,o}$
    takes $O(\log_2(mno))$ time to be constructed inductively, and $\max_{v\in I} \log_2(m_v n_v o_v) \leq |I|$.
    In order to prove that $(\Ical, \Dcal, \Psi, 6)$ is a valid pattern we just need to show that
    $\Mscr \models \Psi_\I$.
    We give an explicit homomorphism from $\I$ to $\bGS_\Ical$. This map is defined by
    $v \mapsto (f_v, g_v)$, where
     \[
    f_v = \left( 
    \frac{m_v}{j}, \frac{n_v}{j},
    \frac{o_v}{j}, \frac{1}{j}, 
    1- \frac{m_v + n_v + o_v + 1}{j}
    \right), \]
    where $j= \max_{v\in I} m_v + n_v + o_v +1$, and    
    \[
    g_v = \left( 
    m_v, n_v,
    o_v, 1, 
    1- (m_v + n_v + o_v + 1)
    \right).
    \]
    Now $\Mscr \models \Psi_\I$ can be shown the same way as in the proof of \Cref{th:AIP_main}-(1) in \Cref{sec:3-dim_grid_AIP}.    
\end{proof}

\subsection{Weak Near-Unanimity Polymorphisms}
\label{sec:WNU}

In this section we prove items (2)(i-iv) of \Cref{th:minor_identities_main}. We begin by introducing a minion $\overline \Wscr$ that characterizes the existence of WNUs suitably. Then, item (1)(i) will follow from interpreting the grid $\bm \Gamma$ over $\overline \Wscr$ (\Cref{sec:all_wnus}), and items (1)(ii-iv) from interpreting growing triangular slices $\bm{\nabla}_n$ over $\overline \Wscr$ (\Cref{sec:wnu_exists}). \par

Given $k\geq 2$, we introduce a minion $\Wscr_k$ that characterizes the existence of a $n$-ary w.n.u. We define an auxiliary minion first. Let $\Wscr_k$ be the minion whose $n$-ary elements are the ordered partitions of $[2]$. That is,
\begin{align*}
\Wscr_k(n) = 
\Biggl\{ 
\gamma \in \left( 2^{[k]} \right)^n \, \Bigg\vert 
\, \bigcup_{i\in [n]} \gamma(i) = [k], \, 
\text{ and }  \gamma(i)\cap \gamma(j) = \emptyset \, \text{for all } i\neq j
\Biggr\}. 
\end{align*}
Given a map $\pi:[n]\rightarrow [m]$ and elements $\omega\in \Wscr_k(n)$ and
$\gamma\in \Wscr_k(m)$, the identity $\gamma = \omega^\pi$ holds if
$\gamma(i) = \bigcup_{j\in \pi^{-1}(i)} \omega(j)$. We write $\gamma\sim_k \omega$ for two elements $\gamma,\omega\in  \Wscr_k(n)$ if
for any $\pi: [n] \rightarrow [2]$
we have that $\gamma^\pi = \omega^\pi$, that $|\gamma^\pi(1)|=|\omega^\pi(1)|= 1$, or that
$|\gamma^\pi(1)|=|\omega^\pi(1)|= k-1$. Clearly the equivalence relation $\sim_k$ is compatible with minoring (i.e. $\gamma \sim_k \omega$ implies $\gamma^\pi \sim_k \omega^\pi$ for all suitable maps $\pi$), so we can define $\overline{\Wscr}_k= \Wscr_k/\sim_k$. Given $\gamma \in \Wscr_k$, we write $\overline{\gamma}$ to denote its $\sim_k$-class.

\begin{lemma}
\label{le:wnu}
    Let $\Mscr$ be a minion and $p\in \NN$ be a prime number. Then $\Mscr$ contains a w.n.u. element of arity $k$ if and only if $\overline{\Wscr}_k \rightarrow\Mscr$.
\end{lemma}
\begin{proof}
    Let $\omega\in \Wscr_k(k)$ be the element defined by $\omega(i)=\{i\}$ for each $i\in [k]$. Suppose there is a homomorphism $\alpha: \overline{\Wscr_k} \rightarrow \Mscr$. Then $\alpha(\overline{\omega})$ must be a $k$-ary w.n.u. \par
    In the other direction, suppose that $f\in \Mscr(k)$ is a w.n.u. Then we define a homomorphism $\alpha: \overline{\Wscr_k}  \rightarrow \Mscr$ by setting 
    $\alpha(\overline{\omega})= f$. This defines the homomorphism completely. Indeed, for any $\gamma\in \Wscr_k(n)$ we have that $\gamma = \omega^{\pi_\gamma}$, where for each $i\in [k]$ we have $\pi_\gamma(i)=j$ if and only if $i\in \gamma(j)$.
    Hence, we define $\alpha(\overline{\gamma})= f^{\pi_\gamma}$.
    In order to prove that $\alpha$ is well-defined we need to prove that $f^{\pi_{\gamma_1}}=f^{\pi_{\gamma_2}}$ for any $\gamma_1, \gamma_2\in \Wscr_k$ satisfying $\gamma_1\sim \gamma_2$. Suppose that $\gamma_1,\gamma_2\in \Wscr_k(n)$ are elements satisfying    
    $\gamma_1\sim \gamma_2$ but $\gamma_1\neq \gamma_2$. Then there must be indices $i,j\in [n]$ satisfying $|\gamma_1(i)|=|\gamma_2(i)|=1$ and
    $|\gamma_1(j)|=|\gamma_2(j)|=k-1$. Let $\tau:[2] \rightarrow [n]$ be the map $1\mapsto i$, $2\mapsto j$. Then for $s=1,2$ we have $\pi^{\gamma_s}= \tau \circ \sigma_s$, where $\sigma_s:[k]\rightarrow  [2]$ satisfies $|\sigma_s^{-1}(1)|=1$ and $|\sigma_s^{-1}(2)|=k-1$. The fact that $f$ is a w.n.u. implies that $f^{\sigma_1}=f^{\sigma_2}$, so $f^{\pi_{\gamma_1}}=f^{\pi_{\gamma_2}}$, as we wanted to prove. 
\end{proof}

We define $\overline{\Wscr}$ as the disjoint union
$\bigsqcup_{k \geq 3} \overline{\Wscr}_k$. A straight-forward corollary of last lemma is the following.

\begin{corollary}
    \label{le:minion_all_wnus}
      Let $\Mscr$ be a minion. Then $\Mscr$ contains a w.n.u. of each arity $k\geq 3$ if and only if $\overline{\Wscr} \rightarrow\Mscr$.
\end{corollary}

\subsubsection{WNUs: Interpreting the Grid}
\label{sec:all_wnus}

The following interpretation $\Ical$ induces a global structure that is finitely equivalent to $\bm{\Gamma}$

\begin{interpret}
\label{interpret:WNUs}
    The $\Sigma_\Gamma$-interpretation $\Ical$ over $\overline{\Wscr}$ is given by
\begin{flalign*}
   U^\Ical= \overline{\Wscr}(3), && 
\end{flalign*}
\begin{flalign*}
    & O^\Ical = \Bigl\{
\overline{\omega} \in \overline{\Wscr}(3) \, \Big\vert 
\, \omega(1)= \omega(2) = \emptyset
\Bigr\}, \quad \text{and} \quad \Pi^\Ical_{O,1}=\id, &&
\end{flalign*}
and, for each $i\in [2]$,
\begin{flalign*}
  & E_i^\Ical = \Bigl\{
\overline{\omega} \in \overline{\Wscr}(4) \, \Big\vert 
\, |\omega(3)|= 1
\Bigr\}, \quad \text{and} \quad 
\Pi^\Ical_{E_i,j} = 
\begin{cases}
    (1,2,3,3) \text{ for $j=1$,}\\
    (1,2,i,3) \text{ for $j=2$}.
\end{cases} &&
\end{flalign*}
\end{interpret}
\noindent
Given an integer $k\geq 3$ we define $\Ical_k$ as the restriction $\Ical\vert_{\overline{\Wscr}_k}$.
Then we have that $\bGS_\Ical= \bigsqcup_{k\geq 3} \bGS_{\Ical_k}$. \par
The idea behind this interpretation relatively simple (the difficult part will be to show that its predicates are internal w.r.t. some suitable description). We represent grid elements $(n,m)\in \NN$ as elements $\overline{\omega}\in \overline{\Wscr}(3)$ where
the partition $\omega$ satisfies $|\omega(1)|=n-1$ and $|\omega(2)|=m-1$. This is correspondence is far from a bijection, but we only need to show finite equivalence. \\~\\ \noindent
    \textbf{ Claim 1: The structure 
    $\bGS_\Ical$ induced by $\Ical$ is finitely equivalent to $\bm{\Gamma}$.}
    Given a number $m\in \NN$, we define $\bm{\Gamma}_m$ to be the substructure of $\bm{\Gamma}$ induced on the elements $(n,o)\in \NN^2$ satisfying $n+o \leq m$.
    Observe that $\bm{\Gamma}$ is finitely equivalent to the disjoint union $\bigsqcup_{i\in \NN} \bm{\Gamma}_i$, and that $\bm{\Gamma}_i\rightarrow \bm{\Gamma}_j$ for each pair $i\leq j$. We prove that
    $\bm{\Gamma}_{k+2}$ is finitely equivalent to $\bGS_{\Ical_k}$
    for each $k\geq 3$. 
    Observe that this proves the claim. We define suitable homomorphisms. Let $F: \bGS_{\Ical_k}
    \rightarrow \bm{\Gamma}_{k+2}$ be the map 
    $\overline{\omega}\mapsto (|\omega(1)|+1,|\omega(2)|+1)$. To see that $F$ is well defined, observe that the relation $\sim_k$ preserves the size of sets. That is,
    if $\omega_1 \sim_k \omega_2$ for some
    $\omega_1, \omega_2\in \Wscr_k$, then
    $|\omega_1(i)|=|\omega_2(i)|$ for all $i$.
    The fact that $F$ is indeed a homomorphism follows from the definition of $\bGS_{\Ical_i}$. \par
    Now let $H: \bm{\Gamma}_{k+2}\rightarrow 
    \bGS_{\Ical_k}$ be the map given by
    $(m,n)\mapsto \overline{\omega_{m,n}}$,
    where $\omega_{m,n}\in \Cscr_p(3)$ is defined as $( X_m, Y_n, [k] \setminus (X_m \cup Y_n))$, where $X_m= [m-1]$ and
    $Y_n= \{k-n + 2, \dots
    k \}$. Observe that for $m=n=1$ we have $X_m=Y_n= \emptyset$. Hence,
    $H(1,1)\in O^{\GS_{\Ical_k}}$.    
    To see that $H$ is a homomorphism we need to prove that 
    $H$ preserves $E_1$ and $E_2$. We show the statement for $E_1$, the other case is analogous. In other words, we need to prove that
    $(\overline{\omega_{m,n}},
    \overline{\omega_{m+1,n}})\in E_1^{\GS_{\Ical_k}}$ for all $((m,n),(m+1,n))\in E_1^{\Gamma_{k+2}}$. Consider the element $\omega= ( X_m, Y_n, \{m\}, [k]\setminus (X_m \cup Y_n \cup \{m\}) )$. Then we have that $\overline{\omega}\in E_1^\Ical$, and 
    \[
    \overline{\omega}_{m,n}=
    \overline{\omega}^{\Pi^\Ical_{E_1,1}},\quad
    \overline{\omega}_{m+1,n}=
    \overline{\omega}^{\Pi^\Ical_{E_1,2}},
    \]
    as we wanted to prove. This completes the proof of the claim.
    \\~\par

  We define a description $\Dcal$ so that $\Ical$ is internal at arity $4$ (w.r.t. $\Dcal$). 
    \begin{desc}
    \label{desc:WNUs}
    The description $\Dcal\subseteq 2^{\overline{\Wscr}}$ consists of the predicates
    \begin{flalign*}
    & D_0= \{
    \overline{\omega}\in \overline{\Cscr}(2) \, \vert \, \omega(1)= \emptyset
    \}, && \\&
    D_1= \{
    \overline{\omega}\in \overline{\Cscr}(2) \, \vert \, |\omega(1)|= 1
    \},  && \\
    &
    D_{*}= \{
    \overline{\omega}\in \overline{\Cscr}(2) \, \vert \, |\omega(1)|, |\omega(2)|\neq 1, \text{ and } 1 \in \omega(1) \}. &&
    \end{flalign*}
    \end{desc}
  
Given $\square\in \{0,1,*\}$ we also define
the auxiliary binary predicate
$C_\square \in 2^{\Wscr}$ which contains all elements $\overline{\omega}$ such that $\overline{\omega}^{(2,1)}\in D_\square$. This way,
\[ 
    \left(\bigsqcup_{\square\in \{0,1,*\} } D_\square \right)\bigsqcup 
    \left(\bigsqcup_{\square\in \{0,1,*\} } C_\square \right)\]
    is a partition of $\Wscr(2)$. \par
We warn the reader that we deal with two nested equivalence relations from now on: An element $\qclass{\overline{\omega}}\in \overline{\Wscr}/\Dcal$ is a $\sim_\Dcal$-class of some $\overline \omega \in \overline \Wscr$, which is in turn a $\sim_k$-class of an element $\omega \in \Wscr_k$ for some integer $k\geq 3$. \par

Now the task at hand is showing that $\Ical$ is internal with respect to $\Dcal$. The crux is proving that the predicate $D_1$ is internal w.r.t.  $\Dcal$, as that predicate is what enables us to define unit increments. This seems intuitive: after all, the elements $\overline{\omega}\in \overline{\Wscr}(2)$ with $|\omega(1)|=1$ are the ones that explain the existence of WNU elements in $\overline{\Wscr}$. The intuition behind the internal reference to $D_1$ shown in Claim 4 is that, starting from an element of the form $\overline{\omega}=\overline{([k], \emptyset`)}$, the elements of $D_1$ are the only ones which allow us to transfer all the weight of $\overline{\omega}$ from its first coordinate to its second coordinate by taking $\Dcal$-equivalent ``scoops''. We show that $\Ical$ is internal at arity $4$ (w.r.t. $\Dcal$) through the following claims. \\~\\

\noindent
\textbf{Claim 2: The predicates $U^\Ical$, $O^\Ical$,
        $E_1^\Ical$ and $E_2^\Ical$ are all $\Dcal$-stable.} This follows from the definitions. \\~\\
\noindent
\textbf{Claim 3: The predicate $O^\Ical$ is internal at arity $4$.} Indeed, the following formula is an internal definition of $O^\Ical$:
        \[
        \phi_O(x) \equiv x = x^{(3,3,3)}.
        \]
        ~\\
        \noindent
\textbf{Claim 4: Let $m\geq 1$ be an integer. Then the following is an internal reference to $D_1$:}
        \begin{align*}
    \phi_{1,m}(x^2) \equiv & \bigexists_{i\in [m+2]}^2 y_i
    \bigexists_{j \in [m]}^4 z_j \Biggl( 
    y_1=y_1^{(1,1)} \,  \wedge  \, y_{m+1}= y_{m+1}^{(2,2)} 
    \Biggr)
    \bigwedge \\&
    \Biggl( \bigwedge_{i\in [m]}  x =z_i^{(2,1,2,2)} 
    \, \wedge x = z_i^{(2,2,1,2)}  \,  
    \wedge \, y_i =  z_i^{(1,1,1,2)} \,  \wedge   \,  y_{i+1}= z_i^{(1,1,2,2)}
    \, \wedge \,  y_{i+2}=z_i^{(1,2,2,2)}
    \Biggr).
    \end{align*}
    We recall the intuition given at the start of this series of claims. Consider an
    element $\overline{\omega_x}\in \overline{\Wscr}$ such that
    $\overline{\Wscr}/ \Dcal \models \phi_{1,m}(\qclass{\overline{\omega_x}})$ and consider a satisfying assignment witnessing this, given by $y_i \mapsto \qclass{\overline{\omega_{y_i}}}$ and $z_i\mapsto \qclass{\overline{\omega_{z_i}}}$ for all suitable indices $i$. What this formula indicates is that all the mass of $\omega_{y_1}$ is in its first coordinate, and all the mass of $\omega_{y_{m+1}}$ is in its second one. The sequence
    $\omega_{y_1},\dots, \omega_{y_{m+1}}$ is obtained by successively taking $\Dcal$-equivalent scoops out of the first coordinate of $\omega_{y_i}$ and placing them into its second coordinate, obtaining $\omega_{y_{i+1}}$. The key idea is that this is only possible if the scoops are of size one. Otherwise, only one of them would contain the element $1$ and would not be $\Dcal$-equivalent to the rest. This motivates the predicate $D_*\in \Dcal$. Now we give the formal argument. \par    
   Suppose that $\overline{\Cscr}/\Dcal \models 
    \phi_{1,m}(\qclass{\overline{\omega}})$,
    with $\qclass{\overline{\omega_{y_i}}}$ and $\qclass{\overline{\omega_{z_i}}}$ as witnesses for each variable $y_i$, $z_i$. We prove that $\overline{\omega}\in D_1$. 
    Suppose that $\overline{\omega}\notin D_1$ for the sake of contradiction. Then $\overline{\omega}$
    must belong to either $D_0,D_*, C_0, C_1, C_*$. We rule out each of the possibilities by case analysis. \par
    \textbf{Suppose that $\overline{\omega}\in D_0$.} Then, $\overline{\omega}_{z_i}^{(2,1,2,2)},
    \overline{\omega}_{z_i}^{(2,2,1,2)}\in D_0$ for all $i\in [m]$. This means that for all $i\in [m]$ we have 
    $\omega_{z_i}(2)=\omega_{z_i}(3)=\emptyset$, and $\overline{\omega}_{z_i}^{(1,1,1,2)}=\overline{\omega}_{z_i}^{(1,1,2,2)}= \overline{\omega}_{z_i}^{(1,2,2,2)}$. This yields that
    $\overline{\omega}_{y_i}\sim_\Dcal \overline{\omega}_{y_{i+1}}\sim_\Dcal \overline{\omega}_{y_{i+2}}$ for all $i\in [m]$, so in particular $\overline{\omega}_{y_{1}} \sim_\Dcal  \overline{\omega}_{y_{m+2}}$. This implies $\overline{\omega}_{y_1}\in D_0 \cap C_0$, a contradiction. \par
    \textbf{Suppose that $\overline{\omega}\in D_*$.} This implies that
     $\overline{\omega_{z_i}}^{(2,1,2,2)}, \overline{\omega_{z_i}}^{(2,2,1,2)} \in D_*$, for all $i\in [m]$. However, this means that
     $1\in \omega_{z_i}(2)$, and $1\in \omega_{z_i}(3)$, contradicting the fact that $\omega_{z_i}$ is an ordered partition. \par 
    \textbf{Suppose that $\overline{\omega}\in C_1$.}
      Then, $\overline{\omega}_{z_i}^{(2,1,2,2)}, \overline{\omega}_{z_i}^{(2,2,1,2)}\in D_*$ for all $i\in [m]$.
    This means that 
    \begin{align*}
    |\omega_{z_i}(1)| + |\omega_{z_i}(3)| + |\omega_{z_i}(4)| = |\omega_{z_i}(1)| + |\omega_{z_i}(2)| + |\omega_{z_i}(4)| = 1,
    \end{align*}
    This implies that $|\omega_{z_i}(2)|= |\omega_{z_i}(3)|\leq 1$. However, these identities together with
    \[
    |\omega_{z_i}(1)| + |\omega_{z_i}(2)| +  |\omega_{z_i}(3)| +  |\omega_{z_i}(4)|  \geq 3
    \]
    also imply that $|\omega_{z_i}(2)|,|\omega_{z_i}(3)|\geq 2$. This yields a contradiction. \par

    \textbf{Suppose that $\overline{\omega}\in C_0$.} Then, $\overline{\omega_{z_i}}^{(2,1,2,2)},
    \overline{\omega_{z_i}}^{(2,2,1,2)}
    \in C_0$
    for all $i\in [m]$. The first inclusion means that
    $\omega_{z_i}(1) =\omega_{z_i}(3) = \omega_{z_i}(4) = \emptyset$, and the second means that
    $\omega_{z_i}(1) =\omega_{z_i}(2) = \omega_{z_i}(4) = \emptyset$. This means that
    all entries of $\omega_{z_i}$ contain the empty set, contradicting the fact that
    $\omega_{z_i}$ is an ordered partition of $[k]$ for some $k\geq 3$. \par    
    \textbf{Suppose that $\overline{\omega}\in C_*$.} 
    This is the hardest case. We have that $\overline{\omega}_{z_i}^{(2,1,2,2)},
    \overline{\omega}_{z_i}^{(2,2,1,2)}
    \in C_*$
    for all $i\in [m]$.
    Using the facts that $\overline{\omega}_{y_1}\in C_0$ and
    $\overline{\omega}_{y_1}\sim_\Dcal \overline{\omega}_{z_1}^{(1,1,1,2)}$ we obtain that
    $\omega_{z_1}(4)=\emptyset$. This way, 
    $\overline{\omega}_{y_2} \sim_\Dcal \overline{\omega}_{z_1}^{(1,1,2,2)}= \overline{\omega}_{z_1}^{(1,1,2,1)}$, and in particular
    $\overline{\omega}_{y_2}\in D_*$. Let $2 < j \leq m+2$ be the smallest index satisfying $\overline{\omega_{y_j}}\notin D_*$. Such index must exist because 
    $D_0$ and $D_*$ are disjoint, and
    $\overline{\omega}_{y_{m+2}}\in D_0$. We prove that $\overline{\omega{y_j}}\in D_1$.
    Indeed, we have both 
    \[
    \overline{\omega}_{y_{j-1}} \sim_\Dcal \overline{\omega}_{z_{j-2}}^{(1,1,2,2)}, \text{ and }
        \overline{\omega}_{y_j} \sim_\Dcal \overline{\omega}_{z_{j-2}}^{(1,2,2,2)}.
    \]
    By our choice of $j$, it must be that
     $\overline{\omega}_{y_{j-1}},  \overline{\omega}_{z_{j-2}}^{(1,1,2,2)} \in D_*$,
     meaning that $1\notin \omega_{z_{j-1}}(3) \sqcup \notin \omega_{z_{j-2}}(4)$. Additionally, the fact that 
     $\overline{\omega}_{z_{j-2}}^{(2,1,2,2)}\in C_*$ implies that
     $1\notin \omega_{z_{j-2}}(2)$ as well. Hence, $1 \in 
     \omega_{z_{j-2}}(1)$. This implies that $\overline{\omega}_{z_{j-2}}^{(1,2,2,2)}$ belongs to either $D_1$ or $D_*$. By our choice of $j$ the first case must hold, yielding $\overline{\omega}_{y_j}\in D_1$. Thus $j <  m+2$, because $\overline{\omega}_{y_{m+2}}\in D_0$.
     Now consider the element $\omega_{z_{j-1}}$. The fact that
     $\overline{\omega}_{y_{j}} \sim_{\Dcal} \overline{\omega}_{z_{j-1}}^{(1,1,2,2)}$ implies that
     $|\omega_{z_{j-1}}(1)| + |\omega_{z_{j-1}}(2)| = 1$. However, the fact that
     $\overline{\omega}_{z_{j-1}}^{(2,1,2,2)}\in D_*$ implies that 
     $|\omega_{z_{j-1}}(2)| > 1$, a contradiction. \\~\\
     \noindent
     \textbf{Claim 5: Let $m\geq 1$ be an integer and let $\overline{\omega} \in \overline \Wscr_{m+2}(2) \cap D_1$. Then $\overline \Wscr \models \phi_{1,m}(\overline{\omega})$.} We find witnesses for each variable in $\phi_{1,m}$. For each $i\in [m+2]$, let $\omega_{y_i}\in \Wscr_{m+2}(2)$ be defined as
    \[
    \omega_{y_i}(1)= [m+2-i], \]
     and
    \[ \omega_{y_i}(2) \{m+3-i,\dots , m+2\}.
    \]
    Similarly, given $i\in [m]$, we define 
    $\omega_{z_i}\in \Wscr_{m+2}(4)$ as
    \begin{align*}
    & \omega_{z_i}(1)= [m-i], \quad 
    \omega_{z_i}(2)= \{ m + 1 -i \}, \\ & 
    \omega_{z_i}(3)= \{ m+ 2 -i \}, \quad \text{ and } \\ &
    \omega_{z_i}(4)= \{m+3-i,\dots , m+2\}.
    \end{align*}
    Now it is routine to check that $\overline \Wscr \models \phi_{1,m}(\overline{\omega})$ with 
    $\overline{\omega}_{y_i}$ as a witness for $y_i$ for each $i\in [m+2]$, and 
    $\overline{\omega}_{z_i}$ as a witness for $z_i$ for each $i\in [m]$. \\~\\ \noindent
    \textbf{Claim 6: Let $i\in [2]$.
    Then $E^\Ical_i$ is internal at arity $4$.}
    By Claim 4, given an integer $m\geq 1$, the following is an internal reference to $E^\Ical_i$: 
    \[
    \phi_{E,m}(x)\equiv \phi_{1,m}(x^{(2,2,1,2)}).    
    \]
    Moreover, by Claim 5, if
    $\overline{\omega}\in \overline\Wscr_{m+2} \cap E^\Ical_i$, then 
    $\overline{\Wscr} \models \phi_{E,m}(\overline\omega)$. Hence the claim follows from \Cref{le:internal_criterion}. \\~\\

    \begin{proof}[Proof of item (2)(i) of \Cref{th:minor_identities_main}]
    Observe that $3\geq$ $\arty(U^\Ical)$ and $3\geq \arty(P)$ for all $P\in \Dcal$.
    The claims in this section show that $\bGS_\Ical$ is finitely equivalent to $\bm{\Gamma}$, and 
    $\Ical$ is internal at arity $4$ w.r.t. $\Dcal$. Hence, the result follows from 
    \Cref{th:reduction_undecidability} together with
    \Cref{prop:grid}. It is enough to consider templates of the form $(\bK^4_3, \B)$ in the statement. 
    \end{proof}

\subsubsection{WNUs: Interpreting Triangular Slices}
\label{sec:wnu_exists}
The following interpretation $\Ical$ over
$\overline{\Wscr}$ induces structures $\bGS_{\Ical_m}$
that are homomorphically equivalent to $\bm{\nabla}_m$ for each integer $m\geq 3$, where $\Ical_m= \Ical\vert_{\overline{\Wscr}_m}$.

\begin{interpret} The $\Sigma_\nabla$-interpretation $\Ical$ over $\overline{\Wscr}$ is defined as in
\Cref{interpret:WNUs}, by adding 
\begin{flalign*}
& W^\Ical= \{ \overline{\omega}\in \overline{\Wscr}(3) \, \vert \, \omega(3) = \emptyset \}, \quad 
\text{and} \quad \Pi^\Ical_{W,1}= \id. &&
\end{flalign*}
\end{interpret}
~\\
\noindent

\textbf{Claim 1: For each integer $k$, the structures $\bm{\nabla}_k$ and $\bGS_{\Ical_k}$ are homomorphically equivalent.} The maps defined in Claim 1 in \Cref{sec:all_wnus} are homomorphisms in both directions between $\bm{\nabla}_k$ and $\bGS_{\Ical_k}$. \\~\par

We define the description $\Dcal\subseteq 2^{\overline{\Wscr}}$ in the same way as in \Cref{desc:WNUs}. We claim $\Ical$ is internal at arity $4$ (w.r.t. $\Dcal$). The following, together with the claims from \Cref{sec:all_wnus}, proves the statement. \\ ~ \\
\noindent
    \textbf{Claim 2: The predicate $W^\Ical$ is internal at arity $4$.} Clearly $W^\Ical$ is $\Dcal$-stable. Additionally, the following is an internal definition of $W^\Ical$:
    \[
    \phi_W(x) \equiv x^{(1,1,2)} = x^{(1,1,1)}.
    \]
~\\

\begin{proof}[Proof of items (2)(ii-iv) of \Cref{th:minor_identities_main}]
Observe that $3\geq \arty(U^\Ical)$ and $3\geq \arty(P)$ for all $P\in \Dcal$.
Let $\Ical$ and $\Dcal$ be the $\Sigma_{\nabla}$-interpretation over
$\overline{\Cscr}$ and the description over $\overline{\Wscr}$ given in this section. 
The claims in this section and the previous one show that for each integer $k\geq 3$, the structure $\bGS_{\Ical_k}$ is homomorphically equivalent to $\bm{\nabla}_{k+2}$, where $\Ical_k= \Ical\vert_{\overline{\Wscr}_k}$, and $\Ical$ is internal at arity $4$ w.r.t. the description $\Dcal$. Then the result follows from \Cref{prop:triangles_undecidability} together with \Cref{th:reduction_undecidability_disjoint_union}. It is enough to consider templates of the form $(\bK^4_3, \B)$ in the statement. 
\end{proof}

\subsection{Cyclic Polymorphisms}
\label{sec:cyclic}

In this section we prove items (1)(i-iv) of \Cref{th:minor_identities_main}. We begin by introducing a minion $\overline{\Cscr}$ that characterizes the existence of cyclic polymorphisms suitably. Then, item (1)(i) will follow from interpreting the grid $\bm \Gamma$ over $\overline{\Cscr}$ (\Cref{sec:all_cyclic}), and items (1)(ii-iv) from interpreting growing triangular slices $\bm{\nabla}_n$ over $\overline \Cscr$ (\Cref{sec:cyclic_exists}). \par
Given prime arity. Given a prime number $p\in \NN$, we define the minion $\mathscr{C}_p$
as follows. We let
\begin{align*}
\Cscr(n) = 
\Biggl\{ 
\gamma \in \left( 2^{\ZZ_p} \right)^n \, \Bigg \vert
\, \bigcup_{i\in [n]} \gamma(i) = \ZZ_p, \, 
\text{ and } \gamma(i)\cap \gamma(j) = \emptyset \, \text{for all } i\neq j
\Biggr\}.    
\end{align*}

The minoring operation is defined as follows. Let $\gamma\in \Cscr(n)$, $\pi:[n]\rightarrow [m]$. Then $\gamma^\pi = \omega$, where $\omega(j)=
\bigcup_{i\in \pi^{-1}(j)} \gamma(i)$ for each $j\in [m]$, and empty unions yield the empty set.\par

Given a set $S\subseteq \ZZ_p$ and an element $m \in \ZZ_p$ we write $S + m$ for the set
$\{ n+m \, \vert \, n\in S
\}$. Given two elements $\gamma,\omega \in \Cscr_p(n)$, we write 
$\gamma \sim_p \omega$ if there is some element $m\in \ZZ_p$ such that $\gamma(i) = \omega(i) + m$ for all $i\in [n]$. Observe that $\sim_p$ is an equivalence relation and it is compatible with minoring, in the sense that
$\gamma \sim_p \omega$ implies $\gamma^\pi \sim_p \omega^\pi$. We write $\overline{\gamma}$ to denote the $\sim_p$-class of an element $\gamma$, 
and define $\overline{\Cscr_p}$ as the quotient minion $\Cscr_p / \sim_p$.

\begin{lemma}
\label{le:cyclic}
    Let $\Mscr$ be a minion and $p\in \NN$ be a prime number. Then $\Mscr$ contains a cyclic element of arity $p$ if and only if $\overline{\Cscr}_p \rightarrow\Mscr$.
\end{lemma}
\begin{proof}
    We show both directions. Suppose there is a homomorphism $\alpha: \overline{\Cscr}_p \rightarrow \Mscr$. Let $\gamma=( \{0\}, \{1\}, \dots, \{p-1\})\in \Cscr_p(p)$. Then the element $\overline{\gamma}\in \overline{\Cscr}_p(p)$ is cyclic, so $\alpha(\overline{\gamma})$ must be cyclic as well. \par
    In the other direction, suppose that $f\in \Mscr(p)$ is a cyclic element. Then we define a homomorphism $\alpha: \overline{\Cscr}_p \rightarrow \Mscr$
    by setting $\alpha(\overline{\gamma}) = f$. This defines the image of any element $\overline{\omega}\in \Cscr_p(n)$. Indeed,
    we have that $\omega=\gamma^{\pi_\omega}$,
    where $\pi_{\omega}: [p] \rightarrow [n]$ is the map that sends $i\in [p]$ to $j\in [n]$ if the element $i-1\in \ZZ_p$
    belongs to $\omega(j)$. Hence, we can define $\alpha(\overline{\omega}) = f^{\pi_\omega}$.
    To see that this is well defined, we need to prove that whenever $\omega_1 \sim_p \omega_2$
    then $f^{\pi_{\omega_1}}=f^{\pi_{\omega_2}}$.
    However, if $\omega_1 \sim_p \omega_2$, then there is some $m\in \ZZ_p$ such that
    $\omega_1(i) = \omega_2(i) + m$ for all $i$.
    This means that $\pi_{\omega_2}= 
    \pi_{\omega_1} \circ \sigma^{m}$, where
    $\sigma = (p-1, 1,\dots, p-2)$ is the cyclic shift. Hence, as $f$ is cyclic,
    we obtain that $f^{\pi_{\omega_2}}=
    (f^{\sigma^m})^{\pi_{\omega_1}}= f^{\omega_1}$, as we wanted.
    We have shown that $\alpha$ is a well-defined map. The fact that $\alpha$ is a minion homomorphism now follows from the fact that
    if $\omega_1 = \omega_2^{\pi}$,
    for some elements $\omega_1, \omega_2\in \Cscr_p$
    then 
    $\pi_{\omega_1} = \pi \circ \pi_{\omega_2}$.
\end{proof}

We define $\overline{\Cscr}$ as the disjoint union
$\bigsqcup_{p \text{ prime }} \overline{\Cscr}_p$. A straight-forward corollary of last lemma is the following.

\begin{corollary}
    \label{le:minion_all_cyclic}
      Let $\Mscr$ be a minion. Then $\Mscr$ contains a cyclic element of each prime arity $p$ if and only if $\overline{\Cscr} \rightarrow\Mscr$.
\end{corollary}

\subsubsection{Cyclic Polymorphisms: Interpreting the Grid}
\label{sec:all_cyclic}

The following interpretation $\Ical$ that induces a global structure which is finitely equivalent to $\bm{\Gamma}$. We note that this interpretation is completely analogous to \Cref{interpret:cyclic}, defined over $\overline{\Cscr}$. 

\begin{interpret}
\label{interpret:cyclic}
The  $\Sigma_\Gamma$-interpretation $\Ical$ over $\overline{\Cscr}$ is given by
\begin{flalign*}
& U^\Ical= \overline{\Cscr}(3), &&    \\
    & O^\Ical = \Bigl\{
\overline{\omega} \in \overline{\Cscr}(3) \, \Big\vert 
\, \omega(1)= \omega(2) = \emptyset
\Bigr\}, \quad \text{ and }  \quad \Pi^\Ical_{O,1}=\id, &&
\end{flalign*}
and, for each $i\in [2]$,
\begin{flalign*}
& E_i^\Ical = \Bigl\{
\overline{\omega} \in \overline{\Cscr}(4) \, \Big\vert 
\, |\omega(3)|= 1
\Bigr\}, \quad \text{ and } \quad \Pi^\Ical_{E_i,j} = \begin{cases}
   (1,2,3,3) \text{ for $j=1$,} \\
   (1,2,i,3) \text{ for $j=2$.} 
\end{cases} &&
\end{flalign*}
\end{interpret}

\noindent
Given a prime number $p$ we define $\Ical_p$ as the restriction $\Ical\vert_{\overline{\Cscr}_p}$.
Then we have that $\bGS_\Ical= \bigsqcup_{p \text{ prime}} \bGS_{\Ical_p}$. \\~\\ \noindent
\textbf{Claim 1: The structure 
    $\bGS_\Ical$ induced by $\Ical$ is finitely equivalent to $\bm{\Gamma}$.}
    This follows similarly to Claim 1 in \Cref{sec:all_wnus}.
    Given a number $m\in \NN$, we define $\bm{\Gamma}_m$ to be the substructure of $\bm{\Gamma}$ induced on the elements $(n,o)\in \NN^2$ satisfying $n+o \leq m$.
    Observe that $\bm{\Gamma}$ is finitely equivalent to the disjoint union $\bigsqcup_{i\in \NN} \bm{\Gamma}_i$, and that $\bm{\Gamma}_i\rightarrow \bm{\Gamma}_j$ for each pair $i\leq j$. We prove that
    $\bm{\Gamma}_{p+2}$ is finitely equivalent to $\bGS_{\Ical_p}$
    for each prime number $p$. 
    Observe that this proves the claim. We define suitable homomorphisms. Let $F: \bGS_{\Ical_p}
    \rightarrow \bm{\Gamma}_{p+2}$ be the map 
    $\overline{\omega}\mapsto (|\omega(1)|+1,|\omega(2)|+1)$. To see that $F$ is well defined, observe that the relation $\sim_p$ preserves the size of sets. That is,
    if $\omega_1 \sim_p \omega_2$ for some
    $\omega_1, \omega_2\in \Cscr_p$, then
    $|\omega_1(i)|=|\omega_2(i)|$ for all $i$.
    The fact that $F$ is indeed a homomorphism follows from the definition of $\bGS_{\Ical_i}$. \par
    Now let $H: \bm{\Gamma}_{p+2}\rightarrow 
    \bGS_{\Ical_p}$ be the map given by
    $(m,n)\mapsto \overline{\omega_{m,n}}$,
    where $\omega_{m,n}\in \Cscr_p(3)$ is defined as $( X_m, Y_n, \ZZ_p \setminus (X_m \cup Y_n))$, where $X_m= \{ 0,\dots, m-1\}$ and
    $Y_n= \{p-n + 1, \dots
    p -1 \}$. Observe that for $m=n=1$ we have $X_m=Y_n= \emptyset$. Hence,
    $H(1,1)\in O^{\GS_{\Ical_p}}$.    
    To see that $H$ is a homomorphism we need to prove that 
    $H$ preserves $E_1$ and $E_2$. We show the statement for $E_1$, the other case is analogous. In other words, we need to prove that
    $(\overline{\omega_{m,n}},
    \overline{\omega_{m+1,n}})\in E_1^{\GS_{\Ical_p}}$ for all $((m,n),(m+1,n))\in E_1^{\Gamma_{p+2}}$. Consider the element $\omega= ( X_m, Y_n, \{m\}, \ZZ_p\setminus (X_m \cup Y_n \cup \{m\}) )$. Then we have that $\overline{\omega}\in E_1^\Ical$, and 
    \[
    \overline{\omega}_{m,n}=
    \overline{\omega}^{\Pi^\Ical_{E_1,1}},\quad
    \overline{\omega}_{m+1,n}=
    \overline{\omega}^{\Pi^\Ical_{E_1,2}},
    \]
    as we wanted to prove. This completes the proof of the claim.  \\~\par
\noindent We define a description $\Dcal$ so that $\Ical$ is internal at arity $5$ w.r.t. $\Dcal$. 
    \begin{desc}
    \label{desc:cyclic}
        The description $\Dcal\subseteq 2^{\overline{\Cscr}}$ consists of the predicates
        \begin{flalign*}
    & D_0= \Bigl\{
     \overline{\omega}\in \overline{\Cscr}(2) \, \Big\vert \, \omega(1)= \emptyset
     \Bigr\}, && \\
    & D_1= \Bigl\{
    \overline{\omega}\in \overline{\Cscr}(2) \, \Big\vert \, |\omega(1)|= 1
    \Bigr\},  && \\
    & D_<= \Bigl\{
    \overline{\omega}\in \overline{\Cscr}(3) \, \Big\vert \, |\omega(1)| < |\omega(2)|
    \Bigr\}, && \\
    & D_\div= \Bigl\{ 
    \overline{\omega}\in \overline{\Cscr}(3) \, \Big\vert \,   |\omega(2)| = n|\omega(1)| \text{ for some integer } n\geq 0
    \Bigr\}.
        \end{flalign*}
    \end{desc}
\noindent
We also consider the auxiliary predicate
\begin{flalign*}
    & D_=  = \Bigl\{
    \overline{\omega}\in \overline{\Cscr}(3) \, \Big\vert \, |\omega(1)| = |\omega(2)|
    \Bigr\}. &&
\end{flalign*}
    
As in \Cref{sec:WNU}, we warn the reader again that we deal with two nested equivalence relations from now on: An element $\qclass{\overline{\omega}}\in \overline{\Cscr}/\Dcal$ is a $\sim_\Dcal$-class of some $\overline \omega \in \overline \Cscr$, which is in turn a $\sim_p$-class of an element $\omega \in \Cscr_p$ for some prime $p$. \\~\\
\noindent
As in the previous section about WNU polymorphisms, here the more involved part is showing that $\Ical$ is internal with respect to $\Dcal$. Here the important step is showing that the predicate $D_1$ is internal, because it allows us to define unit increments. 
We can show that $D_1$ is internal if $D_{\div}$ is internal, because the elements in $D_1$ are precisely the classes of tuples $\omega$ where $\omega(1)$ is a set whose size divides the size of $\omega(2)$ (here it is crucial that we are dealing with cyclic elements of prime arity). To speak about division, we need to speak about equality, so the most important part is proving that $D_=$ is internal, as shown in Claim 4. The intuition is that, if we represent $\ZZ_p$ in a circle in the usual way, an element
$\omega$ with $\overline{\omega}\in D_=$ selects two points in the circle, given by $\omega(1)=\{a\}$ and 
$\omega(2)= \{ b \}$, and we can think of the pair $(a,b)$ as a ``step''. Then the idea is that in $\overline{\Cscr}$ one can return to $a$ from $b$ taking $\Dcal$-equivalent steps. We show that $\Ical$ is internal at arity $5$ (w.r.t. $\Dcal$) through the following claims. \\~\\
\noindent
\textbf{Claim 2: The predicates $U^\Ical$, $O^\Ical$,
        $E_1^\Ical$ and $E_2^\Ical$ are all $\Dcal$-stable.} This is a routine check. \\~\\
        \noindent
\textbf{Claim 3: The predicate $O^\Ical$ is internal at arity $5$.} Indeed, the following formula is an internal definition of $O^\Ical$:
        \[
        \phi_O(x) \equiv x = x^{(3,3,3)}.
        \] ~\\
        \noindent
        \textbf{Claim 4: Let $p$ be a prime number. 
        Then the following predicate is an internal reference to $D_=$:} 
        \[
        \phi_{=,p}(x) \equiv x = x^{(2,1,3)} 
        \]
        \textbf{if $p=2$, and }
        \begin{align*}
        \phi_{=,p}(x) \equiv & 
        \bigexists_{i\in [p-2]}^4 y_i
        \bigexists_{i\in [p-3]}^5 z_i \Biggl(
        y_1^{(1,2,3,3)} = x
        \wedge y_1^{(3,1,2,3)}= x
        \wedge y_{p-2}^{(2,3,1,3)}= x
        \Biggr) \bigwedge \\
        & \Biggl(
        \bigwedge_{i\in [p-3]}
         z_i^{(1,2,3,3,3)} =
         z_i^{(3,3,1,2,3)} \wedge 
         y_i = z_i^{(1,2,3,4,4)} \wedge
         y_{i+1} = z_i^{(1,2,4,3,4)}        
        \Biggr)
        \end{align*}
        \textbf{if $p>2$. }
       Let us show that $\phi_{=,p}(x)$ is an internal reference. The case $p=2$ is straightforward. We assume that $p\geq 3$.
        Suppose that
        $\overline{\Cscr}/\Dcal \models
        \phi_{=,p}(\qclass{\overline{\omega_x}})$
        with $\qclass{\overline{\omega_{y_i}}}$
        as the witness for $y_i$ for each $i\in [p-2]$ and 
        $\qclass{\overline{\omega_{z_i}}}$
        as the witness for $z_i$ for each $i\in [p-3]$. 
        We repeat the intuition given before this chain of claims. The idea is that $\omega_x$ fixes a step inside some cyclic group. I.e., $\omega_x=(\{a\}, \{b\}, \ZZ_p \setminus \{a, b\})$. Then, up to equivalence in $\overline{\Cscr}/\Dcal$,
        each element $\omega_{y_i}$ is of the form $(\{a\}, \{b\}, \{c_i\}, \ZZ_q\setminus \{a, b, c\})$, where $c_{i+1}=c_i + b - a$, and $c_1= 2b - a$. Then the idea is that at the end we have come full circle, so to say, and obtain $c_{p-2} = 2a - b$. Now we proceed with the formal argument. \par
        First, we show that $|\omega_x(1)| = |\omega_x(2)|$. Suppose that
        $|\omega_x(1)|<|\omega_x(2)|$ for a contradiction (the reverse inequality can be dealt with analogously). We prove that
        $|\omega_{y_i}(1)| < |\omega_{y_i}(2)|
        < |\omega_{y_i}(3)|$
        for all $i\in [p-2]$ by induction on $i$. For $i=1$, we have that $\overline{\omega_{y_1}}^{(1,2,3,3)} \sim_\Dcal
        \overline{\omega_{x}}
        $, and
        $\overline{\omega_{y_1}}^{(3,1,2,3)} \sim_\Dcal
        \overline{\omega_{x}}$,
        so necessarily 
        $|\omega_{y_1}(1)| < |\omega_{y_1}(2)|
        < |\omega_{y_1}(3)|$. Now let $i>1$ and suppose that
        \[
         |\omega_{y_{i-1}}(1)| < |\omega_{y_{i-1}}(2)|
        < |\omega_{y_{i-1}}(3)|.
        \]
        We have that
        $\overline{\omega_{y_{i-1}}}\sim_\Dcal
        \overline{\omega_{z_{i-1}}}^{(1,2,3,4,4)}$, so 
        \[
         |\omega_{z_{i-1}}(1)| < |\omega_{z_{i-1}}(2)|
        < |\omega_{z_{i-1}}(3)|.
        \]
        Additionally, $\overline{\omega_{z_{i-1}}}^{(1,2,3,3,3)}\sim_\Dcal
        \overline{\omega_{z_{i-1}}}^{(3,3,1,2,3)}$,
        so 
        \[
         |\omega_{z_{i-1}}(3)|
        < |\omega_{z_{i-1}}(4)|.
        \]
        Finally, $\overline{\omega_{y_{i}}}\sim_\Dcal
        \overline{\omega_{z_{i-1}}}^{(1,2,4,3,4)}$,
        so we can conclude that
        \[
         |\omega_{y_{i}}(1)| < |\omega_{y_{i}}(2)|
        < |\omega_{y_{i}}(3)|,
        \]
        as we wanted to show. We have shown that
        \[
         |\omega_{y_{p-2}}(1)| < |\omega_{y_{p-2}}(2)|
        < |\omega_{y_{p-2}}(3)|.
        \]
        However, $\overline{\omega_{x}}\sim_\Dcal
        \overline{\omega_{y_{p-2}}}^{(2,3,1,3)}$
        implies that
        \[
         |\omega_{y_{p-2}}(3)| < |\omega_{y_{p-2}}(1)|,
        \]
        a contradiction. Hence, it must be that
        $|\omega_x(1)| = |\omega_x(2)|$ to begin with. \\~\\
        \noindent
     \textbf{Claim 5: Let $p$ be a prime, and let $\omega\in \Cscr_p(3)$ be such that
    $|\omega(1)|=|\omega(2)|=1$. Then
    $\overline{\Cscr} \models \phi_{=,p}(\overline{\omega})$.
    } Suppose that $p=2$. Then $\omega \sim_p (
    \{0\},\{1\}, \emptyset)
    \sim_p (
    \{1\},\{0\}, \emptyset)$, proves the statement. Suppose that $p\geq 3$.
    Without loss of generality we can assume that $\omega= ( \{0\}, \{m\}, \ZZ_p \setminus \{0,m\})$ for some $m\in \ZZ_p$. We find witnesses $\overline{\omega_{y_i}}$,
    $\overline{\omega_{z_i}}$, for every variable $y_i, z_i$. Given $i\in [p-2]$ we define
    \begin{align*}
    \omega_{y_i}=\bigl( \{0\}, \, \{m\},  \, \{(i+1)m\}, \, \ZZ_p \setminus \{ 0, m, (i+1)m \} \bigr).
    \end{align*}
    Given $i\in [p-3]$, we define 
    \begin{align*}
    \omega_{z_i}=\bigl(\{0\}, \, \{m\}, \, \{(i+1)m\}, \,
    \{(i+2)m\}, \, \ZZ_p \setminus \{ 0, m, (i+1)m, (i+2)m \} \bigr).
    \end{align*}
    Now it is routine to check that our choice of representatives satisfies $\phi_{=,p}(\overline{\omega})$. \\~\\ \noindent
    \textbf{Claim 6: Let $p$ be a prime, and let $m\geq 0$ be an integer. Then the following is an internal reference to $D_\div$:}
    \begin{align*}
    \phi_{\div,p,m}(x)\equiv &
    \bigexists_{i\in [m+1]}^3 y_i
    \bigexists_{j \in [m]}^4 z_j \Biggl(
    y_1^{(2,2,2)} = y_1^{(2,1,2)} \wedge
    x = y_{m+1}
    \Biggr) \wedge \\ & 
    \Biggl(
    \bigwedge_{i\in [m]}
    z_i^{(1,2,3,3)}= y_i \wedge
    z_i^{(1,2,2,3)} = y_{i+1} \wedge
    \phi_{=,p}(z_i^{(1,3,2,3)})
    \Biggl).
    \end{align*}
    Suppose that $\overline{\Cscr}/\Dcal \models 
    \phi_{\div,p,m}(\qclass{\overline{\omega}})$,
    with $\qclass{\overline{\omega_{y_i}}}$ and $\qclass{\overline{\omega_{z_i}}}$ as witnesses for each variable $y_i$, $z_i$. We prove that $\overline{\omega}\in D_\div$. In order to show this, we prove by induction on $i$ that $\overline{\omega_{y_i}}\in D_\div$ for each $i\in [m+1]$. This proves the result because $\overline{\omega} \sim_\Dcal \overline{\omega_{y_{m+1}}}$. For $i=1$, we have that
    $\overline{\omega_{y_1}}^{(2,2,2)} \sim_\Dcal
    \overline{\omega_{y_1}}^{(2,1,2)}$. Given that
    the first element must belong to $D_0$, so
    does the second one, and we obtain that $\omega_{y_1}(2)=\emptyset$. Hence $\overline{\omega_{y_1}}\in D_\div$ vacuously.
    Now let $i> 1$ and    
    assume that $|\omega_{y_i}(1)|$ divides
    $|\omega_{y_i}(2)|$. Using that
    $\overline{\omega_{y_i}} \sim_\Dcal
    \overline{\omega_{z_i}}^{(1,2,3,3)}$, we obtain that $|\omega_{z_i}(1)|$ divides $|\omega_{z_i}(2)|$. We also have that $|\omega_{z_i}(1)|=|\omega_{z_i}(3)|$ using that $\phi_{=,p}(x)$ is an internal reference to $D_=$. Hence
    $|\omega_{z_i}(1)|$ divides
    $|\omega_{z_i}(2)| + |\omega_{z_i}(3)|$.
    Finally, using that  $\overline{\omega_{y_{i+1}}} \sim_\Dcal
    \overline{\omega_{z_i}}^{(1,2,2,3)}$, we obtain that $\overline{\omega_{y_{i+1}}}$  belongs to $D_\div$, as we wanted to prove. This shows the claim. \\~\\ \noindent
    \textbf{Claim 7: Let $p$ be a prime and let $0 \leq m \leq p-1$ be an integer. Let $\omega_{p,m}\in \Cscr_p(3)$ be defined as $(\{0\}, \{1,\dots, m\}, \{m+1,\dots, p-1\})$. Then $\overline{\Cscr} \models \phi_{\div, p, m}(\overline{\omega_{p,m}})$.} We show the claim by defining suitable witnesses for each variable in $\phi_{\div,p,m}$. For each $i\in [m+1]$ we define
    $\omega_{y_i} = \omega_{p,i-1}$, and for each $i\in [m]$ we define
    \[
    \omega_{z_i} = (\{0\},
    \{1,\dots, i-1\}, \{i \}, \{i+1, \dots, p-1\}).\]
    Now it is routine to check that
    the elements
    $\overline{\omega_{y_i}}$, $\overline{\omega_{z_{i}}}$ are witnesses for the variables $y_i, z_i$. The key observation is that $\overline{\omega_{z_{i}}}^{(1,3,2,3)}$ satisfies $\phi_{=,p}(x)$ on   
    $\overline{\Cscr}$ for all $i \in [m]$
    by Claim 5. \\~\\ \noindent    
    \textbf{Claim 8: The predicate $D_1$ is internal at arity $5$.} We apply \Cref{le:internal_criterion}. Let $\omega_p\in \Cscr_p$ be the element $(\{0 \}, \ZZ_p \setminus \{0\})$. We define an internal reference to $D_1$ that is satisfied by $\overline{\omega_p}$. This interpretation is the following.
    \[
    \phi_{1, p}(x) \equiv
    \exists^3 y \left( 
    y= x^{(1,2)} \wedge 
    \phi_{\div, p, p-1}(y)
    \right).
    \]
    Let us show that $\phi_{1,p}$ is an internal reference to $D_1$. Suppose that
    $\overline{\Cscr}/\Dcal \models 
    \phi_{1,p}(\qclass{\overline{\omega_x}})$ with
    $\qclass{\overline{\omega_y}}$ as a witness for $y$. Observe that $\omega_y(3)$ must be the empty set. Indeed, we have that $\overline{\omega_y}\sim_\Dcal 
    \overline{\omega_x}^{(1,2)}$, so 
    $\overline{\omega_y}^{(2,2,1)}\sim_\Dcal 
    \overline{\omega_x}^{(2,2)}$, and 
     $\overline{\omega_x}^{(2,2)}\in D_0$, 
     so $\overline{\omega_y}^{(2,2,1)} \in D_0$ as well. Additionally, the formula $\phi_{\div,p,p-1}(x)$ is an internal reference to $D_\div$, so $|\omega_{y}(1)|$ must divide $|\omega_{y}(2)|$. Both these numbers must add up to some prime $q$, so the only possibility is that $|\omega_y(1)| = 1$ and $|\omega_y(2)| = q-1$. This means that $\overline{y}^{(1,2,2)}$ belongs to $D_1$.
     Finally, using again that $\overline{\omega_y}\sim_\Dcal 
    \overline{\omega_x}^{(1,2)}$, we obtain again that $\overline{x}$ belongs to $D_1$, as we wanted to prove. \par
    Now, in order to see that $\overline{\omega_p}$
    satisfies $\phi_{1,p}(x)$ on $\overline{\Cscr}$, just consider $\overline{\omega_y}$ as a witness for $y$, where $\omega_y$ is the tuple
    $(\{0 \}, \ZZ_p \setminus \{0\}, \emptyset ) \in \Cscr_p$, and apply the previous claim. \\~\\
    \noindent
    \textbf{Claim 9: Let $i\in [2]$. Then the predicate $E^\Ical_i$ is internal at arity $5$.} We apply \Cref{le:internal_criterion}, as usual. Let $\overline{\omega}\in E_i^\Ical$ be such that
    $\omega \in \Cscr_p$ for a prime $p$. Then, using the previous claim we obtain that
    \[
    \phi_{E_i, p}(x)  \equiv \phi_{1, p}(x^{(2,2,1,2)})
    \]
    is an internal reference to $E_i^\Ical$ and is satisfied by $\overline{\omega}$ on $\overline{\Cscr}$. \\~\\ 
    \noindent
    \textbf{Claim 10: The predicate $U^\Ical$ is internal at arity $5$.} Trivially, the following is an internal definition of $U^\Ical$:
    \[
    \phi_{U}(x) \equiv x=x.
    \]
~\\
    \begin{proof}[Proof of \Cref{th:minor_identities_main}-(1)(i)]
     Observe that $3\geq \arty(U^\Ical)$ and $3\geq \arty(P)$ for all $P\in \Dcal$.
    The claims in this section show that $\bGS_\Ical$ is finitely equivalent to $\bm{\Gamma}$, and 
    $\Ical$ is internal at arity $5$ w.r.t. $\Dcal$. Hence, the result follows from 
    \Cref{th:reduction_undecidability} together with
    \Cref{prop:grid}.
    \end{proof}
    
\subsubsection{Cyclic Polymorphisms: Interpreting Triangular Slices}
\label{sec:cyclic_exists}
The following interpretation $\Ical$ over $\overline{\Cscr}$ induces structures $\bGS_{\Ical_p}$
that are homomorphically equivalent to $\bm{\nabla}_{p+2}$ for each prime number $p$, where $\Ical_p= \Ical\vert_{\overline{\Cscr_p}}$. 

\begin{interpret}
    The $\Sigma_\nabla$-interpretation $\Ical$ over $\overline{\Cscr}$ is defined as in \Cref{interpret:cyclic}, by adding
    \begin{flalign*}
        & W^\Ical= \Bigl\{ \overline{\omega}\in \overline{\Cscr}(3) \, \Big\vert \, \omega(3) = \emptyset \Bigr\}, \quad \text{and} \quad  \Pi^\Ical_{W,1}= \id. &&
    \end{flalign*}
\end{interpret} 

~\\
\noindent
    \textbf{Claim 1: For each prime $p$, the structures $\bm{\nabla}_{p+2}$ and $\bGS_{\Ical_p}$ are homomorphically equivalent.} The maps defined in Claim 1 in \Cref{sec:all_cyclic} are homomorphisms in both directions between $\bm{\nabla}_{p+2}$ and $\bGS_{\Ical_p}$.  \\~\par

We define the description $\Dcal\subseteq 2^{\overline{\Cscr}}$ as in \Cref{desc:cyclic}. The interpretation $\Ical$ is internal at arity $5$ (w.r.t. $\Dcal$). The following claim together with the claims from \Cref{sec:all_cyclic} prove the statement.\\~\\
\noindent
\textbf{Claim 2: The predicate $W^\Ical$ is internal at arity $5$.} Clearly $W^\Ical$ is $\Dcal$-stable. Additionally, the following is an internal definition of $W^\Ical$:
    \[
    \phi_W(x) \equiv x^{(1,1,2)} = x^{(1,1,1)}.
    \]
~\\
\begin{proof}[Proof of items (1)(ii-iv) of \cref{th:minor_identities_main}]
Observe that $3\geq \arty(U^\Ical)$ and $3\geq \arty(P)$ for all $P\in \Dcal$.
The claims in this section and the previous one show that for each prime number $p$, the structure $\bGS_{\Ical_p}$ is homomorphically equivalent to $\bm{\nabla}_{p+2}$, where $\Ical_p= \Ical\vert_{\overline{\Cscr_p}}$, and $\Ical$ is internal at arity $5$ w.r.t. the description $\Dcal$. Then the result follows from \Cref{prop:triangles_undecidability} together with \Cref{th:reduction_undecidability_disjoint_union}. It is enough to consider templates of the form $(\bK^5_3, \B)$ in the statement. 
\end{proof}

\section{Discussion}
\label{sec:discussion}

This work represents a step towards understanding the relationship between search and decision in promise constraint satisfaction. It is important to insist that the relationship between our results, the Search vs Decision question for PCSPs, and other related questions in the literature \cite{brakensiek2019algorithmic,BGWZ20} is nuanced. While we prove that rounding the output of the algorithms $\BLP$, $\AIP$, and $\BLP+\AIP$ is hard (in the TFNP sense), it is still possible that, for instance, whenever $\BLP$ solves the decision version of $\pcsp(\A, \B)$, another linear programming relaxation, such as the ones introduced in \cite{brakensiek2019algorithmic}, can be used to solve the search variant of $\pcsp(\A, \B)$. A counterintuitive possibility is that even if both algorithms $\BLP$ and $\BLP+\AIP$ cannot be used for search separately, it could happen that $\BLP+\AIP$ could be adapted for search in the templates where decision is solvable via $\BLP$. In other words, it could be the case that $\spcsp_{\BLP+\AIP}(\A, \B)$ is tractable whenever $\BLP$ solves $\pcsp(\A, \B)$.

\paragraph{The Search vs Decision Question} 

We have shown that, conditional to TFNP $\not\subseteq$ FP, not every efficient decision algorithm for $\pcsp(\A, \B)$ can be turned into an efficient decision algorithm for $\spcsp(\A, \B)$ that accepts the same instances. So in this particular sense search PCSPs are more difficult to solve than decision PCSPs, although this is not a complexity-theoretic separation. We remark that TFNP $\not \subseteq $ FP is the weakest assumption under which these questions make sense. If $\Qcal$ is a polynomial-time algorithm solving $\pcsp(\A, \B)$, then $\spcsp_\Qcal(\A, \B)$ can be seen as a problem in TFNP if one considers rejections by $\Qcal$ as proper search certificates. Hence, $\spcsp_\Qcal(\A, \B)$ can be solved in polynomial time if TFNP $\subseteq$ FP. \par

We have considered the problem of obtaining search algorithms from efficient decision algorithms, but standing above is the open question of whether, in the finite-template setting, $\spcsp(\A, \B)$ has a polynomial-time solution whenever $\pcsp(\A, \B)$ does. We do not consider our results strong evidence to the contrary: There is some reason to believe that, say, the third level of the $\BLP+\AIP$ hierarchy \cite{cz23soda:minions} could be used to solve $\spcsp(\A, \B)$ for all templates used to prove \Cref{th:AIP_main,th:BLP_main,th:BLP+AIP_main}. We sketch the argument here. 
Given $\Qcal\in \{\AIP, \BLP, \BLP+\AIP\}$, those results show hardness of $\spcsp_\Qcal(\A, \B)$ by proving that $\sPMC_N(\Mscr_\Qcal, \Nscr)$ is hard, where $\Nscr$ is some manifold minion
built on top of a quotient of $\Mscr_\Qcal$. If, instead, we want to prove the stronger result that $\spcsp(\A, \B)$ is hard, then we need to show that $\sPMC_N(\mathscr P, \Nscr)$ is hard, where $\mathscr P$ denotes the minion of projections. The reason our proof fails to show this result is that the minor conditions $\Phi_{\bm G}$ that arise in our patterns do not hold in $\mathscr P$. In fact, they can be ruled out by a $3$-consistency check. Hence, there could be an efficient method of obtaining search certificates for minor conditions that are both accepted by $\Qcal$ and $3$-consistency in a way that does not involve solving the tiling problem encoded in $\Nscr$. \par

We further observe that a separation of search and decision PCSPs implies a non-dichotomy for search PCSPs unless the polynomial hierarchy collapses to its first level. 

\begin{theorem}
    Let $(\A, \B)$ be a finite template for which $\pcsp(\A, \B)$ has a polynomial-time solution. Suppose that NP$\neq$ coNP. Then $\spcsp(\A, \B)$ is not FNP-hard. In particular, if $\spcsp(\A, \B)$ has no polynomial-time solution, then it must be FNP-intermediate. 
\end{theorem}
\begin{proof}
    Let $\Qcal$ be a polynomial-time algorithm solving $\pcsp(\A, \B)$. Then, by our previous reasoning, $\spcsp_\Qcal(\A, \B)\in \text{TFNP}$. The problem $\spcsp(\A, \B)$ is trivially reducible to $\spcsp_\Qcal(\A, \B)$. We note that if we see $\spcsp_\Qcal(\A, \B)$ as a total problem, then this reduction must be seen as a generalized many-one reduction, as it may need to map some answers in  $\spcsp_\Qcal(\A, \B)$ to rejections in $\spcsp(\A, \B)$. It is known \cite[Theorem 2.1]{megiddo1991total} that, unless $\text{NP}=\text{coNP}$, there is not a problem in TFNP that is FNP-hard under generalized many-one reductions \footnote{This result is often misquoted as referring only to many-one reductions, but notice that in that case the statement is vacuous: a non-total search problem cannot be reduced to a total problem via many-one reductions.}. Hence $\spcsp(\A, \B)$ cannot be FNP-hard.
\end{proof}

\paragraph{Small Templates}

Our techniques produce templates $(\A, \B)$ where $\B$ can grow quite large. To construct $\B$ we start with a set of tiles $\bT$ that exhibits some interesting behavior. For instance, if we drop the origin constraint, it is known that the smallest \emph{aperiodic} tile set has $11$ elements \cite{jeandel2015aperiodic}. Or alternatively, we start with a problem $\Pi\in \text{TFNP}$, translate it into a problem about a particular Turing machine, and encode this machine in a tile set $\bT$. Then, construct an manifold minion $\Nscr$ using $\bT$ and a suitable quotient of an interesting minion $\Mscr$. Finally, $\B$ is obtained as some free structure of $\Nscr$, although this last step does not have a significant impact on $\B$'s size. Therefore it seems safe to assume that our techniques have little to say about templates where the domains of both $\A$ and $\B$ are small. For instance, it is still possible that $\AIP, \BLP,$ and $\BLP+\AIP$ can always be adapted to solve search in Boolean PCSPs. \par
A somewhat unsatisfactory aspect of our reductions is that they are oblivious to the structure inside TFNP. For example, can we obtain an explicit and relatively small template $(\A, \B)$ such that $\BLP$ solves $\pcsp(\A, \B)$ and the rounding problem $\spcsp_\BLP(\A, \B)$ is PPAD-hard (i.e., as hard as the problem of computing Nash equilibria \cite{daskalakis2009complexity})? What can we say about the templates $(\A, \B)$ for which $\spcsp_\BLP(\A, \B)$ is PPAD-hard? 

\paragraph{Infinitely Many Symmetric Polymorphisms}
It is known that the algorithm $\BLP+ \AIP$ solves (the decision variant of) $\pcsp(\A, \B)$ whenever
$\pol(\A, \B)$ contains symmetric polymorphisms of infinitely many arities. It is not known, however, whether $\spcsp(\A, \B)$ is also tractable for these templates \cite{krokhin2022invitation, brakensiek2019algorithmic, BGWZ20}. Of course, a negative answer would imply a separation between search and decision for PCSPs. One open possibility is that running the BLP algorithm on a different ring other than $\QQ$, such as $\ZZ[\sqrt{2}]$, is enough to obtain an output that can be rounded in polynomial time and to solve $\spcsp(\A, \B)$. The intuitive reason for this is that those rings avoid ``rounding boundaries'', which seemed to be the obstacle to widely-applicable rounding procedures. An insight from our results is that this difficulty is intrinsic: when we prove TFNP-hardness, the tiling problems are precisely encoded on the rounding boundaries of the corresponding PCSPs (for instance, on the values of the form $\frac{1}{2^n3^m}$ in \Cref{sec:grid_BLP}). If these boundaries are the root reason the rounding problems are difficult, then avoiding them may be key for obtaining efficient search algorithms. \par
The work \cite{brakensiek2019algorithmic} poses a question, later repeated in \cite{BGWZ20} for the Boolean setting, whose affirmative answer would yield efficient search algorithms for all templates with infinitely many symmetric polymorphisms. This question is whether these templates must also contain an infinite consistent family of so-called ``regional-periodic'' polymorphisms. We remark that our results do not rule out this possibility, but they restrict the ways in which it might hold. For example, let $(\A, \B)$ be a finite template satisfying that $\BLP$ solves $\pcsp(\A, \B)$ but all homomorphisms $\Mscr_\BLP \rightarrow \pol(\A, \B)$ are non-computable, as in \Cref{th:BLP_main}. Then the minion $\pol(\A, \B)$ cannot contain a consistent family of nicely-behaved symmetric polymorphisms of all (emphasis on the word \textit{all}) arities, because they would entail a computable homomorphism $\Mscr_\BLP\rightarrow \pol(\A, \B)$. Still, it could be the case that one of these families exists, but only containing polymorphisms of \emph{infinitely} many arities. \par

\paragraph{Other Algorithms} A natural next step in light of our results is to study other PCSP algorithms from the same perspective. An obstacle is that we do not know of explicit minion characterizations for $k$-consistency for $k\geq 3$ and other algorithmic hierarchies \cite{cz23soda:minions} for any level after the second. Two cases where relatively tame minion characterizations can be obtained are the SDP algorithm \cite{bgs_robust23stoc,cz23soda:minions}, named after the \emph{semi-definite programming relaxation} \cite{raghavendra2009approximating}, and extensions of singleton arc-consistency \cite{DebruyneB97} such as the CLAP algorithm, introduced in \cite{cz23sicomp:clap}. Out of these, SDP is the case most similar to the relaxations studied in this paper. In the minion $\Mscr_{\text{SDP}}$, introduced independently in \cite{bgs_robust23stoc,cz23soda:minions}, $n$-ary elements are $n$-tuples of finitely-supported orthogonal vectors in $\mathbb R^{\NN}$ that add-up to $e_1=(1,0,\dots,0,\dots)$. Minoring is defined by means of addition, as in the case of $\Mscr_\AIP$ and $\Mscr_\BLP$. We outline some intuition indicating that our methods may be difficult to apply to $\Mscr_{\text{SDP}}$. An observation is that our reductions exploit the \emph{lack of symmetry} of the minions $\Mscr_\Qcal$ for $\Qcal\in \{\AIP,\BLP, \BLP+\AIP\}$: an interpretation $\Ical$ over $\Mscr_\Qcal$ that is internal (w.r.t. some description) must, in particular, be invariant under the endomorphisms of $\Mscr_{\Qcal}$. This lack of symmetry also seems necessary to obtain internal references to interesting predicates. In the cases of $\Mscr_\AIP$ and $\Mscr_{\BLP}$ the only endomorphisms are the identity maps. However, any isometry of $\mathbb R^\NN$ that fixes the origin and $e_1$ induces an endomorphism of $\Mscr_{\text{SDP}}$, making this minion extremely symmetric compared to the ones studied in this paper. This theme of lack of symmetry leading to hardness is recurrent in the theory of constraint satisfaction.

\paragraph{Other Meta-Problems} We point out several meta-problems whose decidability is still open. 
All of these are well known in the area, but are most often posed as quests for characterizations. Maybe an equally promising direction would be to consider them as invitations to prove the \emph{absence of effective characterizations}. We believe that resolving those questions could shed light on the ways in which PCSPs may be too expressive, or pinpoint some structure that aids in the further development of the theory. \par
The meta-problems for virtually all PCSP algorithms referred to in this paper other than $\AIP, \BLP, \BLP+\AIP$ remain open. The exception is \emph{arc-consistency}, but tractability through this algorithm is equivalent to pp-constructability from a fixed tractable finite-template CSP (Horn $3$-SAT) which is decidable \cite{BBKO21}. An important case is the one of $k$-consistency. Both \cite{Atserias22:soda} and \cite{ciardo2024periodic} have found some sufficient conditions implying that $\pcsp(\A, \B)$ has linear width (and hence is not solved by any fixed level of the local consistency algorithm), but we do not know of any non-trivial necessary conditions. In this context, is also worth going back to the CSP setting. There, the meta-problems for BLP, $k$-consistency, $k$-Sherali Adams, are decidable (\cite{brady2022notes}, \cite{barto2014collapse}, and
\cite{atserias2009affine,atserias2012sherali} respectively) and the relationship between those algorithms is well understood. However, for finite-template CSPs we do not have a good understanding of $\AIP$ and its derived algorithms (e.g., $\BLP+\AIP$, cohomological $k$-consistency \cite{conghaile2022cohomology}, the $\BLP+\AIP$ hierarchy \cite{cz23soda:minions}, CLAP \cite{cz23sicomp:clap}), and the related meta-problems are open. 
\par
Other than meta-problems related to algorithms, an important open question is whether we can recognize the cases where $\pcsp(\A, \B)$ is 
\emph{finitely tractable} \cite{barto2019promises,asimi2021finitely}, 
meaning that there is a finite structure $\C$ such that $\A \rightarrow \C \rightarrow \B$
and $\csp(\C)$ can be solved in polynomial time. In \cite{kazda2022small} it has been shown that the size of the smallest witness $\C$ of finite tractability can grow quite large compared to $\A, \B$ (if P$\neq$ NP), suggesting that characterizing this phenomenon may be difficult. In the same direction, another open-problem is the one of recognizing the cases where $\pcsp(\A, \B)$ is solvable in first-order logic. Recently \cite{mottet2024promise} showed that this occurs precisely when 
$\A \rightarrow \C \rightarrow \B$ for some finite $\C$ such that $\csp(\C)$ is definable in first-order logic. The cases where $\csp(\C)$ is definable in first-order logic are decidable \cite{larose2007characterisation}, but in \cite{mottet2024promise} there is no obvious bound on the size of $\C$ in terms of $(\A, \B)$, suggesting that this also may be a difficult meta-question. \par
Finally, the different notions of reductions between PCSPs are another source of interesting problems. Reductions between finite-template PCSPs by means of pp-constructions are characterized by the existence of a homomorphism between the corresponding polymorphism minions \cite{BBKO21}, which can be shown to be a decidable condition. However, these reductions are provably not enough to obtain all NP-hard finite-template PCSPs from, say, $3$-SAT (see the discussion in \cite{krokhin2022invitation}). Other proposed reductions are the ones given by so-called \emph{$(d,r)$-homomorphisms} between minions \cite{Barto22:soda}, and the more general \emph{local consistency} reductions \cite{dalmau2024local}. These give rise to natural meta-problems: Given finite templates $(\A, \B)$, $(\A^\prime, \B^\prime)$ can we decide (1) whether there is a $(d,r)$-homomorphism from $\pol(\A^\prime, \B^\prime)$
to $\pol(\A, \B)$ for any $d,r\in \NN$? (2) whether $\pcsp(\A, \B)$ reduces to $\pcsp(\A^\prime, \B^\prime)$ via the $k$-consistency reduction for any $k\in \NN$? Another related, and perhaps more accessible question is: given a finite template $(\A, \B)$, can we decide whether $\pol(\A, \B)$ has \emph{bounded essential arity}? By this we mean that there is some $d\in \NN$ such that every polymorphism $f\in \pol(\A, \B)$ depends on at most $d$ variables \cite{BBKO21}. This implies the existence of a (strong version of a) $(d,1)$-homomorphism to the minion of projections $\mathscr{P}$, and is a condition that has been used to prove hardness of some PCSPs (e.g., \cite{austrin2017epsilon_sat}). 

\subsection*{Acknowledgements}
I am especially grateful to Lorenzo Ciardo for his comments on the manuscript, and our discussions on the minions characterizing CLAP and singleton arc-consistency. When I was trying to encode tiling problems in $\Mscr_\BLP$, he also suggested that $\Mscr_\AIP$ could possibly be an easier target, and that approach ultimately worked. I thank Andrei Krokhin for our later discussions on the literature surrounding the search-vs-decision question for PCSPs. I am also thankful to Silvia Butti,
Lorenzo Ciardo, Tamio-Vesa Nakajima, and Standa \v{Z}ivn\'{y} for our early discussions on rounding algorithms. 
I wish to thank Paul Goldberg for discussing with me the status of the class TFNP$_1$, and other notions related to TFNP. Thanks should also go to Antoine Mottet
for pointing out to me that the class of templates admitting symmetric polymorphisms of all arities was not known to be decidable, and discussing this and other meta-questions with me. Finally, I wish to thank the CWC 2024 attendees for their feedback on a talk presenting this work. 

{\small
\bibliographystyle{plainurl}
\bibliography{main}
}

\newpage

\appendix
\section{Proofs of \Cref{sec:sources}}
\label{ap:sources}
We assume familiarity with Turing machines and computability-related notions during this section.
We refer to e.g., \cite{arora2009computational} or \cite{cutland1980computability}. We introduce some related notation below. \par
A non-deterministic Turing machine $M$ is a tuple $(Q, \Omega, \Delta)$, where (1) 
$Q$ is a finite set of states, containing an initial state $q_0$, an accepting state $q_+$ and a rejecting state $q_-$, where $q_+\neq q_-$, (2) $\Omega$ is a finite alphabet containing a distinguished ``blank'' symbol $\square$, and (3) $\Delta \subseteq (Q\times \Omega)  \times (Q\times \Omega \times \{L, R\})$ is the transition relation. The machine $M$ is deterministic if the transition function $\Delta$ contains at most one pair whose first element is $(q,a)$ for each $(q,a)\in Q\times \Omega$. We adopt the convention that a Turing machine $M$ is executed on a single right-infinite tape, and the initial position of the head is always the start of the tape. We say a Turing machine is \emph{immortal} if it has an infinite run starting from any finite input word. 
\par
Let $M=(Q, \Omega, \Delta)$ be a (non-deterministic) Turing machine and $\#$ a fresh symbol. A \emph{configuration} of $M$ is given by a pair  $(\bm{\omega}, \bm{q})\in (\Omega \cup \{ \# \})^{\ZZ_{\geq 0}}   \times  (Q\cup \{ \# \})^{\ZZ_{\geq 0}}$. The $0$-th position of the configuration symbolizes a blank spot to the left of the tape, which starts at position $1$. Formally, (1) $\bm{\omega}$ describes the contents of $M$'s tape from left to right, starting with the symbol $\#$ which marks the extra position before the leftmost end of the tape, and (2) $\bm{q}$ contains a state $q\in Q$ in a single entry and the $\#$ symbol in all the others, denoting that the head of $M$ is at the given position in state $q$. 
Given $i\in \NN$ (observe here indices start at $1$), the $i$-th local description of $(\bm{\omega}, \bm{q})$ is the pair $((\omega_{i-1}, \omega_{i}, \omega_{i+1}), (q_{i-1}, q_{i}, q_{i+1}))$. We denote by $L_M \subseteq (\Omega \cup \{ \# \})^3   \times  (Q\cup \{ \# \})^3$ the set of local descriptions of (configurations of) $M$. If $(\bm{\omega}, \bm{q})$ is a local restriction, we will often use the indices $-1,0,1$ to access its elements, rather than $1,2,3$.
Hence, the sensible restrictions apply: $\#$ can only occupy the first position of $\bm{\omega}_t$, and at most one element in $\bm{q}_t$ can be different from $\#$.

\subsection{Proof of \Cref{prop:grid}-(2)}
\label{sec:grid_undecidability}
In order to prove the results from \Cref{sec:sources} we need some insights from the proof of \Cref{prop:grid}-(2) given in \cite{wang1990dominoes}. This work is in the context of the domino problem. In this problem we are given a finite set of square tiles $T$ of the same size where each side of each tile is colored, together with an initial tile $t_0$. Our task is to decide whether it is possible to tile the infinite upper-right quadrant of the plane using the tiles from $T$ (each tile can be use infinitely many times) in such a way that $t_0$ is placed at the origin and any adjoining edges have the same color. Observe this problem can be seen as a restriction of $\Hom(\Gamma, \cdot)$ where we only consider instances $\bm{T}$ where each element in $T$ is a tuple $t=(t_U, t_D, t_L, t_R)$, corresponding to the edge colors of the tile $t$, the relation $E_1^T$ consists of the pairs $(t,t^\prime)$ such that  $t_R=t^\prime_L$, and the relation $E_2^T$ consists of the pairs $(t,t^\prime)$ such that  $t_U=t^\prime_D$. In \cite{wang1990dominoes} this problem is shown to be undecidable, so by extension $\Hom(\Gamma, \cdot)$ is undecidable as well. \par
The proof in \cite{wang1990dominoes} is a reduction from the Halting Problem. The idea
is that any (non-deterministic) Turing Machine $M$ can be encoded into a finite set of tiles $T_M$ in such a way that a tiling of the upper-right quadrant of the plane corresponds to a non-halting execution of $M$ starting from an empty tape. The intuition behind this encoding is that the $i$-th row of a tiling should represent a configuration of $M$ (i.e., tape contents, head position, and machine state) at the $i$-th time step. This is fairly straightforward, but we sketch a construction that is slightly simpler from the one in \cite{wang1990dominoes}, since we do not need to consider only ``domino'' tiles. \par

Fix a Turing machine $M=(Q, \Omega, \Delta)$. We a $\Sigma_\Gamma$-structure $\bm{T}_M$ as follows. We allow ourselves to be slightly informal in the description of $\bm{T}_M$ and leave some small gaps.  Elements from $T_M$ are tuples $t=( \bm{\omega}_t, \bm{q}_t, b_t, c_t)$ where $(\bm{\omega}_t, \bm{q}_t)\in L_M$ is a local description of $M$,
$b_t\in \{0,1\}$ is a bit that equals $1$
when $t$ describes a local configuration to the right of $M$'s head, and $c_t\in \{0,1\}$ is a bit that equals $1$ when $t$ is describes part of $M$'s initial configuration. Hence, $\bm{q}_t=(\#,\#, \#)$, when $b_t=1$. Similarly, when $c_t=1$, only blank symbols $\square$ and end-of-tape symbols $\#$ are allowed in $\bm{\omega}_t$, and the head of $M$ can only be at the beginning of the tape. The relation $O^{T_M}$ contains only the initial tile $( (\#, \square, \square), (\#, q_0, \#), 0, 1)$. Two elements $(t,t^\prime)$ belong to $E_1^{T_M}$ if they are two 
consecutive local descriptions of some configuration of $M$. More explicitly, the last two elements of $\bm{\omega}_t$ must equal the first two elements of $\bm{\omega}_{t^\prime}$, the last two elements of $\bm{q}_t$ must equal the first two elements of $\bm{q}_{t^\prime}$, and $c_t= c_{t^\prime}$. If $\bm{q}_t=(q,\#,\#)$, for some $q\in Q$, then 
$\bm{q}_{t^\prime}=(\#,\#,\#)$ and $b_{t^\prime}=1$. The relation $E_2^{T_M}$ consists of the pairs $(t,t^\prime)$ that represent consistent local descriptions of two successive configurations of $M$. We apply the natural rules: (1) $c_{t^\prime}=0$, (2) if $\bm{q}_t=(\#,\#,\#)$ then $\bm{\omega}_t=\bm{\omega}_{t^\prime}$, and
$\bm{q}_{t^\prime}$ contains $\#$ in the second position, and (3) if $\bm{q}_t=(\#, q, \#)$ for some $q\in Q$, then $t^\prime$ describes the evolution of $t$ according to some transition in $\Delta$. \par
Given our definition of $\bm{T}_M$ and a homomorphism $F: \bm{\Gamma} \rightarrow \bm{T}_M$, it is easy to see that there is an infinite run of $M$ starting from an empty tape which is given by a sequence of configurations $(\bm{\omega}_1,\bm{q}_1),(\bm{\omega}_2,\bm{q}_2),\dots$ such that the element $F(i,j)$ contains the $j$-th local description of $(\bm{\omega}_1,\bm{q}_1)$. Conversely, given an infinite run of $M$ starting from an empty tape it is straight-forward to describe a homomorphism $F: \bm{\Gamma} \rightarrow \bm{T}_M$. This shows \Cref{prop:grid}-(2). \par

\subsection{Proof of \Cref{prop:grid}-(3)}

The same proof from \cite{hanf1974nonrecursive} essentially shows our result. This work is also in the context of the domino problem. Their main result states that there are sets of domino tiles that can tile the whole plane starting with a fixed tile at the origin, but satisfying that any such tiling must be non-computable. In our setting, we only need to tile the upper-right quadrant of the plane and we consider more general sets of tiles, other than domino tiles, but the proof from \cite{hanf1974nonrecursive} can be easily adapted. The starting point is the result that there exists a Turing machine $M$ that does not halt for some input words, but all those input words are non-computable. This is shown in \cite{hanf1974nonrecursive} for Turing machines on a two-way infinite tape, rather than just a semi-infinite tape as in our setting, but it is known that both models are equivalent. This can be seen, for example, by ``folding'' a two-way infinite tape into a semi-infinite tape as in \cite{hopcroft1969formal}. Then
\Cref{prop:grid}-(3) follows from considering a $\Sigma_\Gamma$-structure $\bm{T}_M$ derived from $M$ as in \Cref{sec:grid_undecidability}, modified suitably so that arbitrary input words are allowed in the initial configuration.

\subsection{Proof of \Cref{prop:triangles_undecidability}}

We build on the ideas from \Cref{sec:grid_undecidability} again. The idea is to build a $\Sigma_\nabla$-structure $\bm{S}_M$ starting from a Turing machine $M$ in such a way that
homomorphisms $F: \bm{\nabla}_n \rightarrow \bm{S}_M$ precisely encode accepting runs of $M$ which start from the empty input, take at most $n$-steps, and where the head of $M$ 
is placed left to the $(n-i+1)$-th position at the $i$-th step. This way, $\bm{\nabla}_n\rightarrow \bm{S}_M$ for any $n\in \NN$, if and only if $\bm{\nabla}_n\rightarrow \bm{S}_M$ for all but finitely many $n\in \NN$, if and only if $\bm{\nabla}_n \rightarrow \bm{S}_M$ for infinitely many $n\in \NN$, if and only if $M$ has a halting run starting from the empty word. Then \Cref{prop:triangles_undecidability} follows from the fact that the problem of determining whether an input Turing machine accepts the empty word is undecidable. \par

To construct $\bm{S}_M$ we start from the tile set $\bm{T}_M$ described in \Cref{sec:grid_undecidability}, and add an additional element $\blacksquare$. The new relation $W^{S_M}$ consists only of this element $\{ \blacksquare \}$. We also add the pair $(\blacksquare, \blacksquare)$ to both $E_1^{S_M}, E_2^{S_M}$. Given a tile $t= ( \bm{\omega}_t, \bm{q}_t, b_t, c_t)\in T_M$, we add
$(t,\blacksquare)$ to $E_2^{S_M}$ if $\bm{q}_t=(\#, \#, \#)$, or if the only state showing in $\bm{q}_t$ is the accepting state $q_+$. We also add $(t,\blacksquare)$ to $E_1^{S_M}$ if 
the last position of $\bm{q}_t$ contains $\#$. Let us give some intuition for this construction. When searching for a homomorphism $F:\bm{\nabla}_n\rightarrow \bm{S}_M$, we proceed as in \Cref{sec:grid_undecidability}, by tiling the plane in a way that describes a run of $M$ on the empty word. Now, at any point we may chose to stop describing this run and start placing $\blacksquare$ tiles instead, with the conditions that (1) $\blacksquare$ tiles must propagate right and up, and (2)
$\blacksquare$ tiles cannot replace local descriptions that contain the head of $M$, as long as $M$ is not in the accepting state. It is routine to check that $\bm{S}_M$ has the desired properties described in the previous paragraphs.

\subsection{Hardness Proofs}

Before moving on to showing \Cref{prop:grid}-(1) and \Cref{prop:3-grid}-(1) we describe some TFNP$_1$-hard and TFNP-hard families. Given a immortal Turing machine $M=(Q,\Omega, \Delta)$, in the problem $R_M$ we are given a pair $(\bm{x}, \bm{n})$ as an input, where $\bm{x}$
is an input word for $M$, and $\bm{n}\in \{1\}^*$ is a number in unary representation, and the task is to output a run of length $|\bm{n}|$ of $M$ on the input $\bm{x}$. Here a run is given as a sequence $\delta_1,\dots,\delta_n$ of transitions in $\Delta$.
Observe the problem $R_M$ belongs to TFNP. We also consider the ``tally'' version of the problem $R_M$. Let $M$ be a (non-deterministic) Turing machine that has an infinite run on the empty input. In the problem $R_M^1$ we are given a number $\bm{n}\in \{1\}^*$ in unary representation, and the task is to output a run of $M$ on the empty input of length $|\bm{n}|$. 

\begin{lemma} 
\label{le:TFNP_hard_Turing}
Let $\mathcal{F}$ be the the family consisting of the problems $R_M$
for each immortal Turing machine $M$ on the binary alphabet $\{0,1,\square\}$. Then
$\mathcal{F}$ is TFNP-hard.
\end{lemma}
\begin{proof}
    We can ignore the restriction to the alphabet $\{0,1,\square\}$ by noting that larger alphabets can be suitably encoded in binary by paying some small overhead \cite{arora2009computational}. Hence, it is enough to prove that the larger family consisting of the problems $R_M$ for each immortal Turing machine $M$ is TFNP-hard. \par
    Let $\Lambda_\Rfrak$ be a problem in TFNP, where $\Rfrak \subseteq U^* \times V^*$. Let $p$ be a polynomial, and let $N$ be a polynomial-time Turing machine satisfying that
    for each $\bm{x}\in U^*$ there is some $\bm{y}\in V^*$ such that $(\bm{x},\bm{y})\in \Rfrak$ and $|\bm{y}|\leq p(|\bm{x}|)$, and $N$ decides $\Rfrak$. We build a immortal Turing machine $M$, informally described as follows. The alphabet of $M$ contains both $U$ and $V$. The machine $M$ loops forever on input words $\bm{x}$ that are not in $U^*$. Given an input $\bm{x}\in U^*$, the machine $M$ guesses a word $\bm{y}\in V^*$ of length at most $p(|\bm{x}|)$ and then runs $N$ on $(\bm{x}, \bm{y})$. If $N$ rejects this pair, then $M$ also rejects, and if $N$ accepts the pair, then $M$ loops forever. Observe that $M$ must be immortal. Let $p^\prime$ be a polynomial such that, given $\bm{x}\in U^*$, it takes $M$ at most $p^\prime(|\bm{x}|)$ time steps to guess $\bm{y}$, simulate $N$, and continue its execution for one more step. 
    \par
    Now we claim that there is a polynomial-time many-one reduction from $\Lambda_\Rfrak$ 
    to $R_M$. The first part of this reduction sends the input $\bm{x}\in U^*$ to $\Lambda_\Rfrak$
    to the input $(\bm{x}, \bm{n})$ to $R_M$, where $\bm{n}$ is a unary representation of $p^\prime(|\bm{x}|)$. Observe that a valid answer to $(\bm{x}, \bm{n})$ in $R_M$ is a run of $M$ on the input $\bm{x}$ that lasts for $p^\prime(|\bm{x}|)$ steps. During such a run, $M$ must guess correctly a word $\bm{y}\in V^*$ satisfying $(\bm{x}, \bm{y})\in \Rfrak$. Indeed, otherwise, $N$ would reject, leading to $M$ halting prematurely. The second part of the reduction extracts the word $\bm{y}$ from the description of the run. This completes the proof.
\end{proof}

\begin{lemma}
    \label{le:TFNP1_hard_Turing} Let $\mathcal{F}$ be the the family consisting of the problems $R_M^1$
for each Turing machine $M$ on the binary alphabet $\{0,1,\square\}$ that has an infinite run on the empty input. Then $\mathcal{F}$ is TFNP$_1$-hard.
\end{lemma}
\begin{proof}
    We ignore the restriction on the alphabet, as in the previous lemma. Let $\Lambda_\Rfrak$ be a problem in TFNP$_1$, where $\Rfrak \subseteq \{1\}^* \times V^*$. Let $p$ be a polynomial
    satisfying that for any $n\in \NN$  there is some $\bm{y}\in V^*$ satisfying $(1^n,\bm{y})\in \Rfrak$ and $|\bm{y}|\leq p(n)$, and let $N$ be polynomial-time deterministic Turing machine deciding $\Rfrak$. \par
    We construct a Turing machine $M$ that has an infinite run on the empty input which satisfies that $\Lambda_\Rfrak$ has a many-one polynomial-time reduction to $R_M^1$. The machine $M$ is informally described as follows. On the empty input, $M$ keeps track of an integer counter $k$, whose value starts at one. Then, $M$ guesses a word $\bm{y}\in V^*$ with $|\bm{y}|\leq p(k)$, and then runs $N$ on the input $(1^k, \bm{y})$. If $N$ rejects, then $M$ rejects and halts. Otherwise, $M$ increases the value of $k$ by one and repeats the process again. Let $q$ be a polynomial such that it takes $M$ at most $p^\prime(m)$ steps to set its counter $k$ to the value $m+1$. 
    \par
    We describe a suitable reduction from $\Lambda_\Rfrak$ to $R_M^1$. The first part of the reduction takes an input $1^n$ to $\Lambda_\Rfrak$ and constructs the input $1^{p^\prime(n)}$ to $R_M^1$. A valid response to $1^{p^\prime(n)}$ in $R_M^1$ is a run of $M$ on the empty input that lasts for $p^\prime(n)$ steps. During such a run $M$ must guess correctly a word $\bm{y}\in V^*$ satisfying $(1^n, \bm{y})\in \Rfrak$ (otherwise $M$ would halt prematurely). The second part of the reduction extracts the word $\bm{y}$ from the description of this run.
\end{proof}

\subsubsection{Proof of \Cref{prop:grid}-(1)}
Let $M$ be a Turing machine that has an infinite run on the empty word, and let $\bm{T}_M$ be the $\Sigma_\Gamma$-structure constructed in \Cref{sec:grid_undecidability}. We show that the problem $R_M^1$ has a polynomial-time many-one reduction to $\spcsp(\bm{\Gamma}, \bm{T}_M)$. Observe that this reduction together with \Cref{le:TFNP1_hard_Turing} prove the result. \par
The first part of the reduction takes an input $1^n$ to $R_M^1$ and produces an input $\bm{\Gamma}_n$ to 
$\spcsp(\bm{\Gamma}, \bm{T}_M)$. The structure  $\bm{\Gamma}_n$ is the substructure of $\bm{\Gamma}$ induced on $[n]\times [n]$, i.e., the $n\times n$ grid.
A valid output to $\bm{\Gamma}_n$ in $\spcsp(\bm{\Gamma}, \bm{T}_M)$ is a homomorphism 
$F: \bm{\Gamma}_n\rightarrow \bm{T}_M$. The horizontal lines $F(i,j)$ for each $i\in [n]$ describe successive configurations (truncated to the first $n$ tape spaces) of $M$ starting from the empty input. Therefore the homomorphism $F$ contains the description of a $n$-step run of $M$ on the empty input. The second part of the reduction simply extracts this run from $F$.

\subsubsection{Proof of \Cref{prop:3-grid}-(1)}
This result can be shown similarly to \Cref{prop:grid}-(1), with some work. Let $M$ be an immortal Turing machine. We begin by modifying the construction from \Cref{sec:grid_undecidability} to obtain a $\Sigma_{\Gamma^+}$-structure $\bm{T}_M$ satisfying that any homomorphism $F:\bm{\Gamma}^+ \rightarrow \bm{T}_M$ admits the following description. We view the super-grid $\bm{\Gamma}^+$ as the positive quadrant of the $3$-dimensional grid, where we call the dimensions vertical, horizontal, and normal respectively. Each normal slice (i.e., each set consisting all tuples $(i,j,k)$ where $k$ is fixed) is called a \emph{floor}. Then the tiling $F(i,j,k)$ of the $k$-th floor describes an infinite run of $M$ on the input given by the binary representation $\lfloor i \rfloor$ of $i$. We adopt the convention that binary representations start from the least-significant bit. I.e., $\lfloor 6 \rfloor = 011$.

\par
We construct $\bm{T}_M$ as follows. Let $M=(Q, \Omega, \Delta)$ be a Turing machine on the binary alphabet $\Omega= \{ 0,1, \square\}$. Similarly to the $\Sigma_\Gamma$-structure constructed in \Cref{sec:grid_undecidability}, the elements in $T_M$ 
are tuples $t=(\bm{\omega}_t, \bm{q}_t, b_t, c_t, d_t, e_t)\in \Omega^3 \times Q^3 \times \{0,1\}^4$, where $(\bm{\omega}_t, \bm{q}_t)\in L_M$ is a local description of $M$. 
The bit $b_t$ keeps track of whether the tile $t$ describes a local configuration to the right of $M$'s head, bit $c_t$ keeps track of whether the tile $t$ is in the initial horizontal line of some floor (e.g., is in a position $(i,j,k)$ with $j=1$), and the bit $d_t$
keeps track of whether $t$ is in the first floor (e.g., its position is $(i,j,k)$ with $k=1$). Finally, the bit $e_t$ is a \emph{carry-over} bit, used during addition operations described below. The origin relation $O^{T_M}$ is the singleton containing the tile
\[
((\#, 1, \square), (\#, q_0, \#), 0,1,1,0). 
\]
The vertical and horizontal relations $E_1^{T_M}$, $E_2^{T_M}$ are described analogously to the ones in \Cref{sec:grid_undecidability}. The main difference is that we force the tape described by the first horizontal line in the first floor to be $(\#,1,\square, \square, \dots)$ (i.e., the input $1$), rather than the blank tape. We also need to describe how $E_1^{T_M}$ and $E_2^{T_M}$ interact with the carry-over bit $e_t$. The vertical relation $E_2^{T_M}$ ``forgets'' about the carryover, meaning that if $(t,t^\prime)\in E_2^{T_M}$, then $e_{t^\prime}=0$. However, the horizontal relation $E_1^{T_M}$ propagates the carryover. This means that if $(t,t^\prime)\in E_2^{T_M}$, and $e_t=0$, then $e_{t^\prime}=0$. However if $e_t=1$ then we have two options. If $e_{t^\prime}=0$ then the rightmost element of $\bm{\omega}_{t^\prime}$ must be $1$, or, otherwise, if $e_{t^\prime}=1$ then 
the rightmost element of $\omega_{t^\prime}$ is $0$. The doubling relations $\Ebb_1^{T_M}$, $\Ebb_2^{T_M}$ do not impose any constraint. The normal relations $E_3^{T_M}, \Ebb_3^{T_M}$ only constrain the tiles that lie in the first horizontal line of each floor, i.e., those $t\in T_M$ with $c_t=1$. The relation $\Ebb_3^{T_M}$ ensures that if the tape contains a binary representation $\lfloor n \rfloor$ in the horizontal line $(i,1,1)$, then the horizontal line at $(2i,1,1)$ displays $\lfloor 2n \rfloor$. This is achieved by 
by shifting right $\bm{\omega}_t$ in each local description $(\bm{\omega_t}, \bm{q}_t)$ and by adding a leading $0$. In a similar way, the relation $E_3^{T_M}$ ensures that if the horizontal line
$(i,1,1)$ displays $\lfloor n \rfloor$, then $(i+1,1,1)$ displays $\lfloor n +1 \rfloor$. This is requires addition to be performed on $\bm{\omega}_t$ from left to right in each local description 
$(\bm{\omega_t}, \bm{q}_t)$ following the standard \emph{column addition} algorithm, and using the carry-over bit $e_t$ as needed. This way, we obtain a structure $\bm{T}_M$ satisfying the high-level description given at the start of the section. \par
In order to complete the proof of \Cref{prop:3-grid}-(1), we describe a reduction from $R_M$
to $\spcsp(\bm \Gamma^+, \bm{T}_M)$. Let $(\lfloor n \rfloor, 1^m)$ be an input to $R_M$, where $\lfloor n \rfloor$ corresponds to the binary representation of $n\in \NN$, and $m\in \NN$. Without loss of generality we may assume that $m\geq n$: if $m<n$ we can act as if $m=n$.
The first map of our reduction constructs in polynomial time a structure $\I_{n,m}$ satisfying $\I_{n,m}\rightarrow \bm \Gamma^+$. First, observe that there is a sequence $\spadesuit_1,\dots, \spadesuit_\ell$ of length $\ell \in O(\log n)$ consisting of the operations $+1$ and $\times 2$ that yields $n$ starting from $1$. Such a sequence $\bm{L}$ can be obtained in polynomial time from $\lfloor n \rfloor$. For each $i\in [\ell+1]$ we define the number $n_i$ inductively as follows. We let $n_1=1$, and $n_i= n_{i-1} \spadesuit_{i-1}$ for $i>1$, where we recall that $\spadesuit_{i-1}\in \{ +1, \times 2\}$. This way, $n_{\ell+1}=n$. The universe $I_{n,m}$ is defined as
\[
\{ (i,j,n) \, \vert \, i,j\in [m] \} \bigcup \{ (i,1,n_k) \, \vert \, i\in [m], k\in [\ell+1] \}. \]
In other words, this includes the initial $m\times m$ quadrant in the $n$-th floor of $\NN^3$, plus
several horizontal segments of length $m$ that ``climb'' up to the $n$-th floor. The relations are defined so that $\I_{n,m}$ is the substructure of $\bm{\Gamma}^+$ induced on $I_{n,m}$. By construction $\I_{n,m}\rightarrow \bm \Gamma^+$. Now, for the second map of our reduction we need a polynomial-time procedure that computes a $m$-step run of $M$ starting from the input $\lfloor n \rfloor$ by accessing a homomorphism $F: \I_{n,m} \rightarrow \bT_{M}$. This is done in the intuitive way: By construction, image of the horizontal line $(i,1,n)$ through $F$ must describe the initial configuration of $M$ on the input $\lfloor n \rfloor$, and the image through $F$
of the quadrant $\{ (i,j,n) \, \vert \, i,j\in [m] \}$ must describe a $m$-step run of $M$
starting from this input. Hence, by accessing $F$ one can compute the desired run in polynomial time. This completes the proof. 

\section{Characterizing Polymorphism Minions up to Isomorphism}
\label{app:pol_characterization}
We extend the ideas of \cite[Section 6.2]{BG21:sicomp}, where function minions corresponding to polymorphisms are characterized. The main result there is that a function minion is a polymorphism minion if and only if it has bounded \emph{finitized arity}, as we described in \Cref{sec:minion_closures}. The following characterizes abstract minions that are isomorphic to the polymorphism minion of some finite template.

\begin{theorem}
\label{th:characterization_pol}
Let $\Mscr$ be a minion. Then there exists a finite template $(\A, \B)$ such that $\pol(\A, \B)$ is isomorphic to $\Mscr$ if and only if $\Mscr$ is locally finite and $\Mscr$ is isomorphic to $\Mscr^{(h)}$ for some $h\in \NN$.    
\end{theorem}

We need the following short auxiliary results.

\begin{lemma}
 Let $\Mscr$ and $\Nscr$ be minions, let $h\in \NN$, and let $F:\Mscr \overpartialmap{h} \Nscr$ be a partial minion isomorphism defined up to arity $h$. Then $\Mscr^{(h)}$ is isomorphic to $\Nscr^{(h)}$.
\end{lemma}
\begin{proof}
 We define an isomorphism $H: \Mscr^{(h)} \rightarrow \Nscr^{(h)}$. Given a system $\xi \in \Mscr^{(h)}(n)$, its image
 $H(\xi)$ is the system defined by $H(\xi)(\pi)= F(\xi(\pi))$ for all $\pi \in [h]^{[n]}$. Checking that this is a minion homomorphism is routine, and its inverse is given by $H^{-1}(\rho)(\pi)=F^{-1}(\rho(\pi))$ for all $n\in \NN$ all $\rho \in \Nscr^{(h)}(n)$ and all $\pi\in [h]^{[n]}$
\end{proof}

As a corollary we obtain that the $h$-closure of a $h$-closure is isomorphic to the original $h$-closure.

\begin{corollary}
    \label{cor:closure_of_closure}
    Let $\Mscr$ be a minion, let $h\in \NN$ and let $\Nscr=\Mscr^{(h)}$. Then, $\Nscr$ is isomorphic to $\Nscr^{(h)}$.
\end{corollary}
\begin{proof}
      By \Cref{th:closure_equivalent}, $\Mscr$ and $\Mscr^{(h)}=\Nscr$ are isomorphic up to arity $h$, so the statement follows from last lemma. 
\end{proof}

Now we proceed with our main proof.

\begin{proof}[Proof of \Cref{th:characterization_pol}]
We show both directions. Let $(\A, \B)$ be an arbitrary finite template. Let $r=|A|$, and let $h\in \NN$ be at least as large as $r$ and as large as the number of tuples in any relation of $\A$. We identify $A$ with $[r]$. A consequence of the definition of polymorphism is that a function of the form $f:A^n \rightarrow B$
belongs to $\pol(\A, \B)$ if and only if $f^{\pi}\in \pol(\A, \B)$ for every function $\pi\in [h]^{[n]}$. Define $\Mscr=\pol(\A, \B)$. Clearly, $\Mscr$ is locally finite. We show that $\Mscr$ is isomorphic to $\Mscr^{(h)}$. The isomorphism is the canonical map $\mathrm{Cl}_h:\Mscr \rightarrow \Mscr^{(h)}$. We already know that $\mathrm{Cl}$ is a minion homomorphism, so all that is left is to show that it is bijective. An argument that is convenient to remember is that, under the usual identification of tuples from $[r]^n$ with maps in $[r]^{[n]}$, the evaluation of a function $f: [r]^n\rightarrow B$ on a tuple $\bm a \in [r]^n$ can be seen as the evaluation $f^{\bm a}(\id_{[r]})$.
\par
\emph{$\mathrm{Cl}_h$ is injective}. Let $n\in \NN$, and let $f, g\in \Mscr(n)$. Suppose that $f\neq g$. Then by the remark above, it holds that $f^{\pi}\neq g^\pi$ for some $\pi \in [r]^{[n]}$. Composing $\pi$ with the inclusion 
$[r]\hookrightarrow [n]$ we obtain $\pi^\prime \in [h]^{[n]}$ such that $f^{\pi^\prime} \neq g^{\pi^\prime}$. Thus, by the definition of the canonical map, it holds that
\[
\mathrm{Cl}_h(f)(\pi^\prime)= f^{\pi^\prime} \neq g^{\pi^\prime} = \mathrm{Cl}_h(g)(\pi^\prime),
\]
meaning that $\mathrm{Cl}_h(f)\neq \mathrm{Cl}_h(h)$.\par
\emph{$\mathrm{Cl}_h$ is surjective}. Let $n\in \NN$ and let $\xi\in \Mscr^{(h)}(n)$ be a system. 
The idea is that $\xi$ is given by a family of consistent $h$-ary polymorphisms in $\Mscr$, 
so there is a $n$-ary function $f_{\xi}: [r]^n \rightarrow B$ whose $h$-ary minors are given by $f_{\xi}^\pi=\xi(\pi)$. Constructing such $f_\xi$ would show that $\mathrm{Cl}_h$ is surjective: Indeed, under this assumption all $h$-ary minors of $f_\xi$ belong to $\Mscr$ forcing $f_\xi\in \Mscr$ as well. Finally, by definition, $\mathrm{Cl}_h(f_\xi)=\xi$, so this shows that $\mathrm{Cl}_h$ is surjective. We define the desired function $f_\xi: [r]^n \rightarrow B$ as follows. Fix two maps $\alpha: [r] \rightarrow [h]$ and $\beta: [h]\rightarrow [r]$ such that $\beta \circ \alpha = \id_{[r]}$. Given a tuple $\bm a\in [r]^n$, we define
\[
f_\xi(\bm a)= \xi(\alpha \circ \bm a)(\beta).\]
To parse this equality, recall that $\xi(\alpha \circ \bm a)$ is a  $h$-ary polymorphism in $\Mscr$, which can be evaluated at a tuple in $[r]^h$. All that is left to show is that $f_\xi^\pi=\xi(\pi)$ for all $\pi \in [h]^{[n]}$. Indeed, in this case, given a tuple $\bm a\in {[r]}^{[h]}$ it holds that
\[
f_\xi^\pi(\bm a)= f_\xi(\bm a \circ \pi)= \xi(\alpha\circ \bm a \circ \pi)(\beta)= 
(\xi(\pi))^{\alpha \circ \bm a}(\beta) = 
\xi(\pi)( \beta \circ \alpha \circ \bm a) = \xi(\pi)(\bm a). 
\]
This proves that $\mathrm{Cl}_h$ is a surjective homomorphism, and together with the fact that it is also injective, completes the proof that it is an isomorphism. \par
We have shown that if $\Mscr$ is the polymorphism minion of some finite template, then it is locally finite and is isomorphic to $\Mscr^{(h)}$ for some $h\in \NN$. Now we prove the other direction. Let $\Mscr$ be a locally finite minion, and suppose that it is isomorphic to $\Nscr=\Mscr^{(h)}$ for some $h\in \NN$. An observation is that the rank of $\Nscr$ is at most $h$. Indeed, let $n\in \NN$ and let $\xi_1, \xi_2\in \Nscr(n)$ be two systems of $h$-ary minors. Suppose that $\xi_1^\pi= \xi_2^\pi$ for all $\pi\in [h]^{[n]}$. We also have that $\xi_1(\pi)= \xi_1^\pi(\id_{[h]})$
and analogously with $\xi_2$. This means that $\xi_1(\pi)=\xi_2(\pi)$ for all $\pi\in [h]^{[n]}$ and $\xi_1=\xi_2$. 
Then, by \Cref{cor:minion_to_template}, it holds that $\Nscr^{(h)}$ is isomorphic to $\pol(\bK_h^h, \bF_{\Nscr}(\bK_h^h))$. Also, $\Nscr^{(h)}$ is isomorphic to $\Nscr=\Mscr^{(h)}$ by \Cref{cor:closure_of_closure}, and $\Mscr^{(h)}$ is, by assumption, isomorphic to $\Mscr$. Putting together all these isomorphisms we obtain that 
$\Mscr$ is isomorphic to  $\pol(\bK_h^h, \bF_{\Nscr}(\bK_h^h))$, completing the proof. 
\end{proof}

An alternative, and perhaps better way to present the second part of this proof could have been to show, using essentially the same reasoning as in the proof of \Cref{th:minion_to_template}, that given a minion $\Mscr$ and a number $h\in \NN$, it holds in general that $\Mscr^{(h)}$ is isomorphic to $\pol(\bK_h^h, \bF_{\Mscr}(\bK_h^h))$.

\end{document}